 \documentclass[reqno,11pt,a4paper]{amsart}

\pdfoutput=1

\usepackage{amsfonts,amsmath,amssymb}
\usepackage[latin1]{inputenc}
\usepackage{dsfont}
\usepackage{enumerate}
\usepackage{xcolor}
\usepackage[all]{xy}
\def\notshow#1\notshowend{} %
\usepackage{hyphenat}

\usepackage{comment}

\usepackage{graphicx}

\usepackage{tikz}
\usepackage{tikz-3dplot}
\usepackage{pgfplots}
\pgfplotsset{compat=1.14}

\usepackage{multirow}
\usepackage{array}

\usepackage{amssymb}
\usepackage{amsfonts}
\usepackage{mathrsfs}
\usepackage{hyperref}
\hypersetup{
    colorlinks=true,
    linkcolor=blue,
    urlcolor=blue,
    pdftitle={Título},
}

\usepackage[normalem]{ulem} 

\def\bb#1\eb{\textcolor{blue}{#1}} 
\def\br#1\er{\textcolor{red}{#1}} %
\def\bm#1\em{\textcolor{magenta}{#1}} %

  \newcommand{\N}{\mathds{N}}     
  \newcommand{\R}{\mathds{R}}     
 \newcommand{\C}{\mathcal{C}}
    \newcommand{\Lo}{\mathds{L}}     
\newcommand{\LL}{{\mathbb{L}}}
\newcommand{\Ver}{\mathcal{V}}
\newcommand{\Hor}{\mathcal{H}}
\newcommand{\Ng}{\mathcal{N}} 

\newcommand{\torN}{\mathrm{Tor}}
  \newcommand{\RN}{\mathcal{R}}   
  
\newtheorem{thm}{Theorem}[section]
\newtheorem{prop}[thm]{Proposition}
\newtheorem{lemma}[thm]{Lemma}

\theoremstyle{definition}
\newtheorem{defi}[thm]{Definition}

\newtheorem{rem}[thm]{Remark}
\newtheorem{note}[thm]{Note}

\newcommand{\noi}{\noindent}
\newcommand{\be}{\begin{equation}}
\newcommand{\ee}{\end{equation}}
\newcommand{\ben}{\begin{enumerate}}
\newcommand{\een}{\end{enumerate}}
\newcommand{\bit}{\begin{itemize}}
\newcommand{\eit}{\end{itemize}}
\newcommand{\edoc}{\end{document}}

\newcommand{\bq}{\begin{quote}}
\newcommand{\eq}{\end{quote}}

\title[]{
On the foundations  and  applications   of \\ Lorentz-Finsler Geometry  }

\author[M. S\'anchez]{Miguel S\'anchez} \address{IMAG \&  Departamento de Geometr\'{\i}a y Topolog\'{\i}a, Facultad de Ciencias, \hfill\break\indent Universidad de Granada,\hfill\break\indent Campus Fuentenueva s/n, \hfill\break\indent 18071 Granada, Spain}\email{sanchezm@ugr.es}

\begin{document}

\begin{abstract} 
Finslerian extensions of Special and General Relativity---commonly referred to as Very Special and Very General Relativity--- necessitate the development of a unified Lorentz-Finsler geometry. However, the scope of this geometric framework extends well beyond relativistic physics. Indeed, it offers powerful tools for modeling wave propagation in classical mechanics, discretizing spacetimes in classical and relativistic settings, and supporting effective theories in fundamental physics. Moreover, Lorentz-Finsler geometry provides a versatile setting that facilitates the resolution of problems within Riemannian, Lorentzian, and Finslerian geometries individually. This work presents a plain introduction to the subject, reviewing  foundational concepts, key applications, and future prospects. 

The reviewed topics include (i) basics on the setting of cones, Finsler and  Lorentz-Finsler metrics and their (nonlinear, anisotropic and linear) connections, (ii) the global structure of Lorentz-Finsler manifolds and its space of null geodesics, (iii) links among Riemannian, Finsler and Lorentz geometries, (iv)  applications in classical settings as   wildfires and seisms propagation, and  discretization in classical and relativistic settings with quantum prospects, and (v)  Finslerian variational approach to Einstein equations. The new results include the splitting of globally hyperbolic Finsler spacetimes, in addition to the analysis  of several extensions of the Lorentz setting, as  the case of timelike boundaries.

\vspace{5mm}

\\


\end{abstract}
\maketitle

\newpage

{\small
\tableofcontents
}

\newpage

\section{Introduction}\label{s_Intro}

Recently, Lorentz-Finsler Geometry has attracted the attention of many researchers due to the convergence of several active areas of development, including:
(a) Finslerian modifications of Special and General Relativity, which require the establishment of a standard Lorentz-Finsler geometric framework;
(b) rheonomic applications of Finsler Geometry, which naturally lead to cone structures akin to those found in Lorentz-Finsler manifolds; and
(c) the interplay among Riemannian, Lorentzian, and Finslerian geometries, from which Lorentz-Finsler metrics emerge as a compelling unifying concept.

This survey aims to provide an overview of these advances
---necessarily non-exhaustive. Readers are assumed to have a basic background in Riemannian Geometry and elementary Lorentz/relativistic concepts, but not necessarily prior training in Finsler Geometry.

The article is structured into five main sections, in addition to this Introduction.
Section \ref{s2} provides a summary of foundational Finsler concepts and their extension to the Lorentz-Finsler setting.
Sections \ref{s3}---\ref{s6} 
  develop four  broad topics, which can be read in an essentially  independent way.
Sections~\ref{s3}-\ref{s5} are more geometrically oriented, covering global Finsler structures and applications to classical geometries.
Sections \ref{s4}-\ref{s6} focus on applications, including  phenomena in applied classical Physics and Finslerian approaches to gravity.


In Section 2, we adopt a pedagogical, non-technical approach to introduce key concepts such as cone structures $\mathcal{C}$, Lorentz-Finsler metrics $L$, various types of connections and the viewpoints of Hamiltonian mechanics and contact geometry.
In particular, we stress the viewpoint  of 
anisotropic connections, developed by MA Javaloyes and coworkers \cite{J19, J20, JSV_Gelocor, JSo},  which is  intermediate between non-linear and Finsler connections and is used in some of the references. 

Several choices must be made among the many possibilities found in the literature. Our guiding principle is to establish a reasonable foundational framework that can be adapted for specific studies (e.g., low regularity, singularities, physical models). Some of our choices, further justified throughout the article, include:
(i) smooth elements, typically at least $C^2$ 
 thus avoiding  singular directions (but see for example footnote~\ref{f_smooth});
(ii)~connected manifolds with semi-Finsler metrics, that is, requiring the fundamental tensor $g$ is non-degenerate and of constant signature (in parallel with the semi-Riemannian framework of O'Neill \cite{O});
(iii)~strong convexity of Finsler indicatrices;
(iv)~in Lorentzian signature, restriction to causal vectors within a cone structure to focus on causality and distance-maximizing properties, avoiding the arbitrariness of metric extension beyond the cone;
(v)~retention of the timelike/spacelike terminology in the semi-Finsler case (as in \cite{O}), although this distinction is purely conventional for non-Lorentzian signatures.

Section \ref{s3} develops globally hyperbolic Finsler spacetimes. Notice that, as in the Lorentz case, Finsler global hyperbolicity draws some analogies with metric Riemannian completeness.  In \S \ref{s3_1}, we focus  on their global structure and prove the possibility to split it in an orthogonal way, extending the Lorentz case (Theorem.~\ref{t_splitting}); this  includes the case of  cone structures  -with-timelike-boundary (Theorem.~\ref{t_splitting2}). With this aim, we briefly  review the techniques used so far   and explain detailedly how those in \cite{BS03, BS05, AFS} permit to extend the results to the Lorentz-Finsler case, including  additional possibilities of these splittings. In \S \ref{s3_2}, we consider the space of cone geodesics $\Ng$, introduced   (in the Lorentz case)   by R. Low \cite{Low1989} following seminal ideas by R. Penrose \cite{Penrose0}.    Very recently, J. Hedicke    \cite{Hedickepreprint} (see also  \cite{HedickePhD}) has considered     $\Ng$ systematically in the Lorentz-Finsler case, and the case with timelike boundary has been studied in \cite{HS25}.      Finally,  Finslerian extensions of Lorentzian singularity theorems are also briefly reviewed.


Section \ref{s5} highlights a special interplay between Lorentzian and Finslerian viewpoints. In particular \S\ref{s5_2}, Finsler geometry aids in describing the causality of stationary spacetimes, and wind Finsler structures extend this to SSTK spacetimes. Conversely \S\ref{s5_1}, Lorentzian cone structures provide insights into Finsler and Riemannian problems, such as the classification of Randers metrics with constant flag curvature and the strong refocusing of Riemannian geodesics. Finally \S \ref{s5_3}, the causal boundary for spacetimes is then shown to encompass classical (Cauchy, Gromov) Riemannian boundaries  and their Finslerian counterparts.

 Section \ref{s4} begins with a review on  the Fermat principle, Zermelo navigation problem and Snell's law within the general framework of cone structures, \S \ref{s4_1}. This yields  multiple applications, among them wildfire and seismic monitoring,  \S \ref{s4_2}.  Furthermore, the discretization of these classical problems opens avenues for applications in Numerical Relativity and other foundational  theories in Physics.    
 
Section \ref{s6} briefly reviews Finslerian approaches to Relativity, emphasizing geometric aspects.  In \S \ref{s6_1} we start at the origins  of Very Special and Very General Relativity and compare them with classical frameworks. Then, we focus in the Finsler variational approach analog to Einstein Hilbert \S \ref{s6_20} and review briefly its vacuum solutions \S \ref{s6_3}. Finally,    a comparison with  the Einstein-Palatini approach is carried out, \S \ref{s6_4}.

Summing up, this survey aims to bridge the gap between classical  and emerging Finslerian approaches, highlighting both theoretical richness and practical relevance.  Ultimately, Lorentz-Finsler Geometry stands as a promising field for future exploration, offering new insights into the structure of spacetime and its applications in Geometry and Physics.

\newpage

\section{
Setting  on cones, Lorentz-Finsler metrics  and connections}\label{s2}

The main references along this section are \cite{JS20,   JSV_Gelocor,  Hedickepreprint,  Min15} and \cite{Dahl} (the latter considers only the Finsler case, but it is easily adaptable to the Lorentz-Finsler one).

\subsection{Minkowski,  wind Minkowski and  Lorentz-Minkowski norms} 

In this introductory section,  we primarily follow  \cite{JS20}, adopting its Euclidean background,  conventions and intuitions derived from \cite{JS14, CJS24}. It  is worth pointing out that some low regularity issues on norms extend to the more general framework of Banach spaces \cite{Day} but we will not deepen on this. 

\subsubsection{Emergence of non-symmetric norms from  anisotropic propagation}\label{sss_motivation}

Consider a moving object  or a propagating wave  respect to an observer in a classical (non-relativistic) setting, for example, a Zeppelin or sound wave in a (possibly windy) air respect to an observer on earth. The velocity depends on the oriented direction, thus,  at each point, the space of (maximum) velocities of propagation  yields a closed (compact without boundary) hypersurface $S$, thus, enclosing an open domain $D$. Let us assume that $D$ includes the  velocity 0 (in our example, this means that  the wind is not too strong, we will discuss this case later, see \S \ref{sss_wind}). Notice that $S$ may be non-symmetric respect to $0$, as travelling in the oriented direction of the wind may be  faster than on the reversed one.  
We will be interested in the  infimum time to arrive from a first point  $x$ to a second one $\bar x$. When the anisotropic velocities are independent on the point, a reasoning in affine geometry   (as around  Fig.~\ref{f_tineq} below),  shows that the original $S$ can be effectively replaced  by its convex hull (Fig.~\ref{fig_VelocidadConvexa}).
Consistently, we will assume that $S$ is convex in what follows.

	\begin{figure} 
		
		\includegraphics[width=0.6\textwidth]{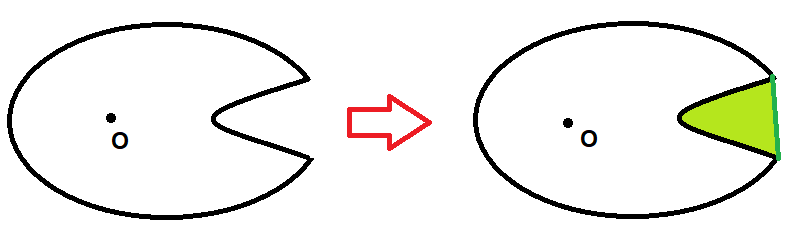}
		\caption{Hypersurface  $S$ enclosing 0 and its convex hull. }\label{fig_VelocidadConvexa}
	\end{figure}

Some remarks on   convexity for  a  {\em smooth} hypersurface $S= \partial D$   (in a real   finite $n$-dimensional vector space) $V$ are in order. Let $\xi$ be any vector field transverse to $S$ and  pointing to $D$, and
$\sigma^{\xi}$ the corresponding second fundamental form\footnote{
Notice that $\sigma^{\xi}$ is obtained at each 
$v\in S$ in a standard way by decomposing second derivatives in the direction 
of $\xi_v$ and $T_vS$. The property of being positive definite 
or semi-definitive for $\sigma^{\xi}$ is independent of the chosen $\xi$, in particular, one can  
use any auxiliary Euclidean scalar produt and take  $\xi$ 
as a unit vector field ortogonal to $S$.}.  
 The following notions on convexity for $S$ (Fig. \ref{f_3convexities}) are increasingly restrictive:
\bit\item
 {\em Infinitesimally convex} (or simply {\em convex}):  $\sigma^{\xi}$ is positive semidefinite at each    $v\in S$.   
\item {\em Strictly convex}:  each tangent hyperplane $T_vS$   intersects $S$ only at $v$. 
\item {\em Strongly convex}:  $\sigma^{\xi}$ is positive semidefinite  at each    $v\in S$. 
\eit  In terms of the interior domain $D$,
 $S$ is convex when, for each $v,w\in \bar D$, the segment joining them lies in $\bar D$ and  strictly convex if the segment lies in $D$ except at most the endpoints $v$, $w$ (see Fig. \ref{f_convexD}).

\begin{center}
	\begin{figure}
\includegraphics[width=1.0\textwidth]{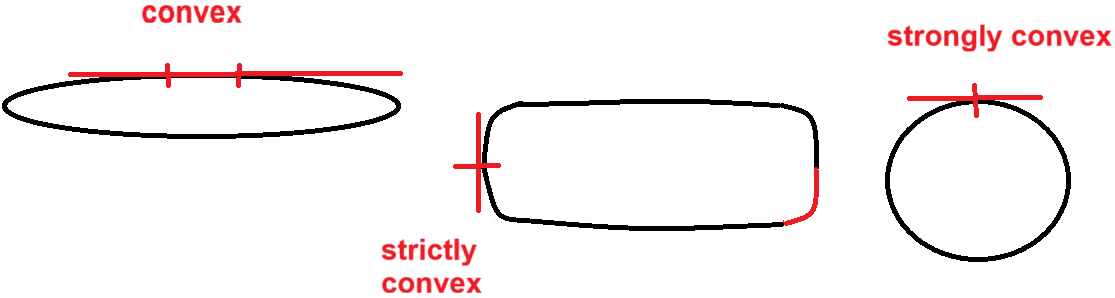}
	\caption{Types of convexity and implications: \newline
		  Infinitesimal: permits flat parts in $S$. \newline
		 Strict: norm with strict triangle ineq.,  bijective $\flat$. 	\newline
		 Strong: smooth $\sharp:= \flat^{-1}$. 
	} \label{f_3convexities}
\end{figure} 
\end{center}

\begin{center}

\begin{figure} 
	$$	\includegraphics[width=.5\textwidth]{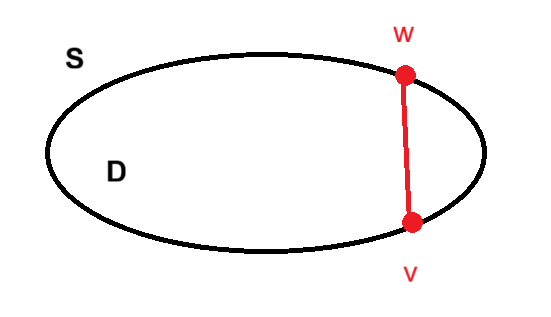}
	 $$
\caption{(Strict) conv. of $S$   equivalent to (strict) conv. of $D$.}\label{f_convexD}
\end{figure}
\end{center}
$S$ can be regarded as the indicatrix (unit sphere) for a non (necessarily) symmetric norm.

\begin{prop}\label{p_norm}
Let $S$  be a smooth convex closed hypersurface of $V$ such that its  open interior domain $D$   is precompact and $0\in D$ (in particular, $S=\partial D$). 
For  each $v\in V$, there exists a unique scalar  $\parallel v\parallel \geq 0$ satisfying that $v$ belongs to the hypersurface
$\lambda S$. Then, the    map $\parallel \cdot \parallel: V \rightarrow \R$ is a possibly non-symmetric norm, that is, it satisfies:

\bit \item  Positiveness: $\parallel v \parallel \geq 0$ with equality if and only if $v=0$
\item  Positive homogeneity: $\parallel \lambda v \parallel =  \lambda  \parallel v \parallel $,  for $\lambda \geq 0$   
\item Triangle inequality: $\parallel v+w\parallel \leq \parallel v \parallel + \parallel w\parallel$
\eit
\end{prop}
 To check that  the (strict) triangle inequality is equivalent  to the (strict) convexity of $S$, a geometric proof for unit vectors $v,w$ is depicted in Fig. \ref{f_tineq} (see also, for example, \cite[Prop.2.3 ]{JS14} or \cite{BOZ} for infinite dimension). 

\begin{figure}

\begin{center} 		
		
		\includegraphics[width=1.05\textwidth]{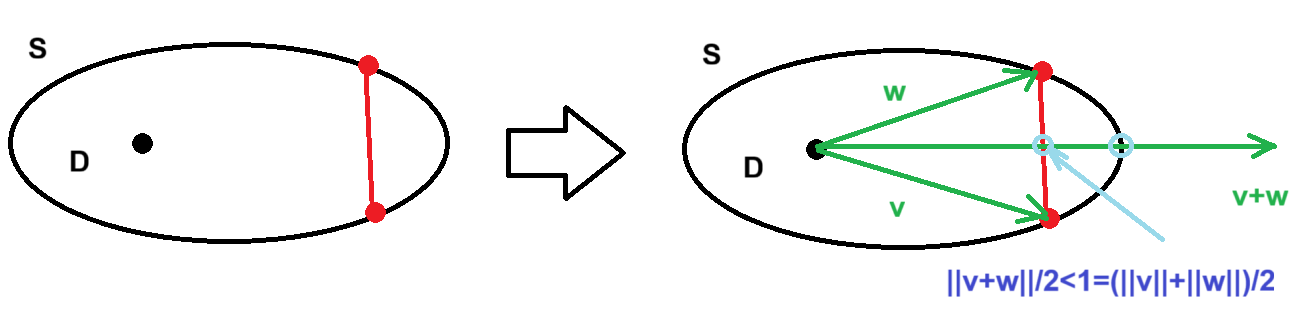}			
	\end{center} 
\caption{(Strict) convexity implies that the center of the paralelogram with sides $v,w$ lies in (the interior of) $\bar D$, thus yielding (strict) triangle inequality.  analogous picture}\label{f_tineq}
\end{figure}

\subsubsection{Minkowski norms}
Norms are never smooth at 0 and, noticeably:

\begin{lemma}\label{l_norms}
Let $\parallel \cdot \parallel$ be a non-symmetric norm. Its square $\parallel \cdot \parallel^2$ is smooth at 0 if and only if it comes from an Euclidean scalar product.
\end{lemma} 

\begin{proof} For the necessary condition, use that the Hessian of $\parallel \cdot \parallel^2/2$  at 0,  is well defined, symmetric and it agrees with $\parallel \cdot \parallel^2$, see \cite[Prop. 4.1]{Wa}.  
\end{proof}

Then a natural regularity for 
$\parallel \cdot \parallel$  is smoothness (at least $C^2$) away from $0$ and, then, so  will be its indicatrix $S=\parallel \cdot \parallel^{-1}(1)$. In this case, putting  $ L:= \parallel \cdot \parallel^2$ 
the {\em fundamental tensor} of the norm is:
\begin{equation}\label{e_tensor_fundamental}
   g_v:= \frac{1}{2} \, \hbox{Hess} \, _v  L    \qquad \forall v\neq 0 .
\end{equation}
It is easy to check that   $g_v$ restricted to $TS$ is the affine second fundamental form of $S$ respect to $-v$  and $g_{\lambda v}=g_v$ for $\lambda >0$. In particular, {\em the indicatrix of $\parallel \cdot \parallel$ is strongly convex if and only if $g_v$ is positive definite} (for all $v\neq 0$), which will be our standard choice for norms.

\begin{center} 
		\begin{figure}

		$$
		\includegraphics[width=.4\textwidth]{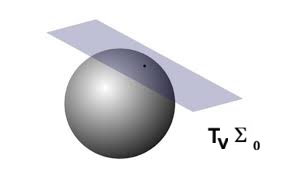}
		$$
		\caption{For   $v\in \Sigma_0$: $T_v\Sigma_0=(v^{\flat})^{-1}(\Sigma_0)$
		  }\label{f_flat}
	
		\end{figure}
		
	\end{center}

\begin{defi}
A {\em Minkowski norm}  $L_0= \parallel \cdot \parallel_0^2$ is a possibly non-symmetric norm (according to Prop.~\ref{p_norm}, now $\parallel \cdot \parallel_0$ is positive homogeneous and $L_0$ is 2-homogeneous) which is smooth outside 0  and has a strongly convex indicatrix $\Sigma_0=L_0^{-1}(0)$, that is, its {\em fundamental tensor} defined as in \eqref{e_tensor_fundamental}
is positive definite.  
\end{defi}

The following property of Minkowski norms will rely on our choice of convexity. For any non-symmetric norm $\parallel \cdot \parallel$ smooth away from 0,  the {\em Legendre} or {\em (non-linear) flat map} is
\begin{equation}\label{e_Legendre}
\begin{array}{ccc} \flat : V\rightarrow V^*, & &  v \mapsto v^{\flat}:=(d L)_v= g_v(v,\cdot ). 
\end{array}
\end{equation}
Here $v^\flat$ is the 1-form characterized by:  $v^\flat(v)= \parallel v \parallel^2 $, $0^\flat =0$ and, whenever $v\neq 0$, $\hbox{ker}(v^\flat)$ is the vector hyperplane parallel to $T_{v/\parallel v \parallel}S$ (see Fig. \ref{f_flat}).

This map is clearly continuous, positive homogeneous and onto. It also  extends the standard lowering index map of scalar products. However it may be non linear and (see \cite[Sect. 1.2]{Dahl}):

\bit\item $S$ is strictly convex 
 $\Longleftrightarrow $  the flat map $\flat$ is injective. 
 Then, the inverse, or {\em sharp} map  $\sharp : V^*\rightarrow V$, $\omega\mapsto \omega^\sharp$ exists, and it is continuous.
  
\item $D$ is strongly convex  
$\Longleftrightarrow$ $\flat$ (thus $\sharp$) is a diffeomorphism.
Indeed, Minkowski norms are characterized by this property among all the non-symmetric norms.  
\eit
The first item is geometrically obvious and, then, the second one comes from the non-degeneracy of $g$, which implies that \eqref{e_Legendre} is a local diffeomorphism (see \cite{Min17}, or  \cite{BBC} to include the infinite dimensional case). 

\begin{rem}\label{r_gradiente}
For Minkowski norms, the gradient of a (smooth) function $f:V\rightarrow \R$ is naturally defined as grad$f:=(df)^\flat$ and it becomes smooth when doesn't  vanish. 
In general, it is non-linear and only positive homogeneous. 
Whenever  $v:=$grad$_x f\neq 0$, it is orthogonal (for the scalar product $g_v$) to the tangent space to $S$ in the oriented direction of $v$. Indeed, this direction  selects the maximum of $df_x$ on $S$, which can be checked by using a Lagrange multiplier for the constraint $g_u(u,u)-1
=0$.  
\end{rem}

\subsubsection{Wind norms 
}\label{sss_wind} 
Coming back to our original problem on propagation \S \ref{sss_motivation}, we can wonder what happens if the 0 velocity is not included in the open domain $D$, that is, in the cases of either  {\em critical wind}, when $0\in S=\partial D$ or {\em strong wind}, when $0\not\in \bar D:=D\cup \partial D$. In the latter case, a natural cone with vertex at $0$ and tangent to $S$ emerges and $S$, up to the tangency points, is the union of two open connected subsets $S_{\hbox{\small{conv}}}$ and $S_{\hbox{\small{conc}}}$  (see Fig. \ref{f_wind}).

In a natural way, $S_{\hbox{\small{conv}}}$ can be regarded as the indicatrix of a {\em conic Minkowski norm} defined only in the interior of the cone. It also gives a strict triangle inequality for vectors in the interior domain of the cone (see Def. 2.14 and Section 2 in \cite{JS14}). Analogously, $S_{\hbox{\small{conc}}}$ gives a {\em Lorentz norm} and a similar reverse triangle inequality holds  (consider a similar picture as in Fig. \ref{f_tineq} with a piece of concave indicatrix), see  \cite[\S 3.1, Appendix B]{JS20}. In both cases, we can define the square of the norm $L_0$ and, thus, the fundamental tensor  $g_v$ (as in \eqref{e_tensor_fundamental}) for $v\neq 0$ in the open conic region,  by means of the  positive homogeneous extension:
\begin{equation}\label{e_L0}
L(\lambda  v)= \lambda^2, \quad \hbox{ $\forall \lambda>0$ and either all $v \in S_{\hbox{\small{conv}}}$ or all $v\in S_{\hbox{\small{conc}}}$} 
\end{equation}
In the case of the conic Minkowski norm, $g_v$ is positive definite but in the Lorentz one, $g_v$ has Lorentzian signature $(+,-, \dots , -)$.

In the Lorentz part, the reverse triangle inequality has a natural interpretation even in a classical setting. Namely, in our example,  the Lorentz norm correspond to the case when the zeppelin engine is trying to move in the direction opposite to the wind (describing a straight line in the conic region), however, as it is drifted by the wind, the arrival time is    maximized.

The case of critical wind corresponds to the limit situation  when the cone degenerates into a hyperplane tangent to $S$. Hence, the 0 velocity for the zeppelin corresponds to  the velocity of the wind compensated by the maximum velocity of the engine in the opposite direction.   This leads to the concept of wind Minkowski norm \cite{CJS24}.

\begin{figure}
\begin{center} 		
		\includegraphics[width=0.5\textwidth]{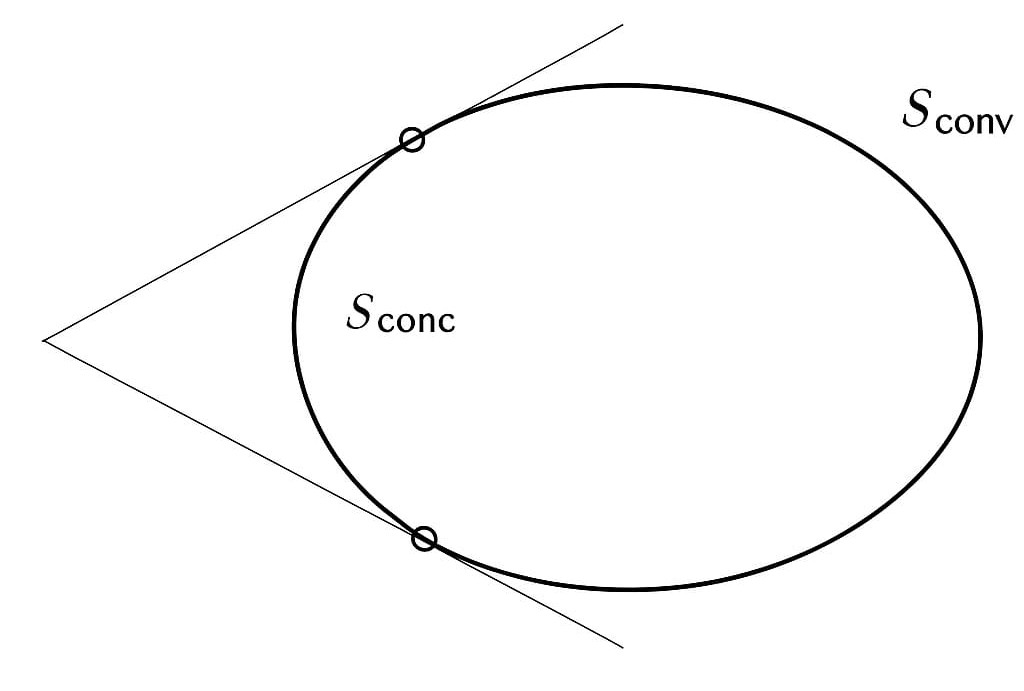}
	\caption{Case of strong wind: emergence of convex and concave conic indicatrices in the interior of a cone.}\label{f_wind}			
	\end{center} 
\end{figure}

\begin{defi}\label{d_wind_norm}
A {\em wind Minkowski norm} on $V$ 
is a smooth connected compact strongly convex hypersurface $\Sigma_0 =\partial D_0$, where $D_0$ is the interior domain of $\Sigma_0$. Depending on whether $0\not\in \bar D_0$, $0\in \Sigma_0$ or $0\in D_0$ the wind is called {\em strong, critical} or {\em weak}, resp. (the latter  being a Minkowski norm too).
\end{defi}

\subsubsection{Rheonomic emergence of cones}\label{sss_cones} 
The rheonomic viewpoint  considers the   
time evolution of the motion for a wind norm after a unit of time.  It naturally yields  cones, as in Fig. \ref{f_conewind}. 
This viewpoint:
\bit 
\item desingularizes  wind Finsler structures 
 and
\item is adapted to geometrically model time-dependent  velocities.
\eit

\begin{figure}
\begin{center}
		\includegraphics[width=0.45\textwidth]{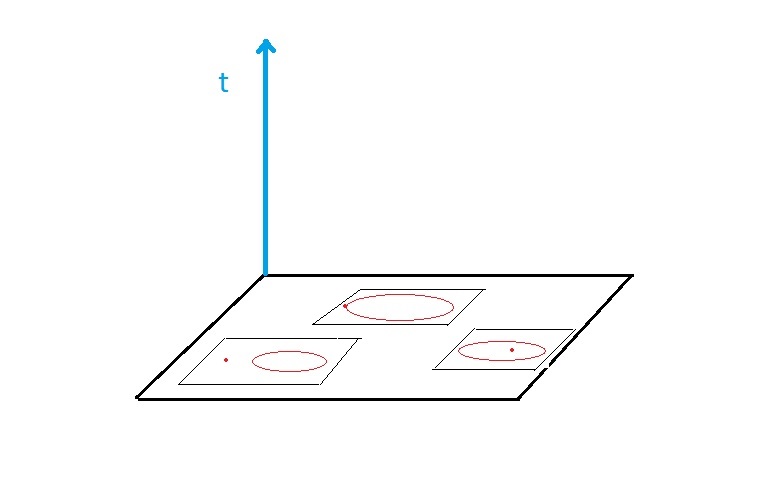}
		\includegraphics[width=.45\textwidth]{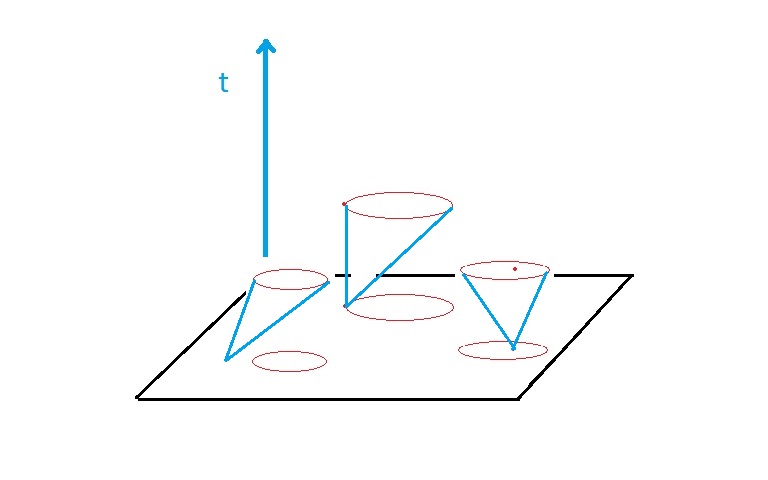}
\end{center}		
\caption{Cones from strong/critical/weak wind Minkowski  norms for three different cases  in the affine plane $t=0$.}\label{f_conewind}
\end{figure}

\begin{defi}\label{d_cone} A (strongly convex) {\em cone} in $V$ is  a smooth  hypersurface $\C_0$ embedded in   $V\setminus\{0\}$   satisfying the properties: 
 \begin{enumerate} 
 \item Conic:  $v\in \C_0 \Rightarrow \lambda v \in \C_0, \forall \lambda>0$.
 
 \item  Salient:  if $v\in \C_0$, then $-v\notin \C_0$.
 
\item Convex interior: $\C_0= \partial A_0$ in $V\setminus \{0\}$ 
with $A_0 \subset V\setminus \{0\}$  an open,  convex subset (as around Fig. \ref{f_convexD}), called the {\em interior}  (or cone domain)  of the cone. 

\item   (Non-radial) strong convexity: the 2nd fundamental form $\sigma^\xi$ of $\mathcal{C}_0$ (with $\xi$ pointing to $A$) is positive semi-definite with radical at each $v\in \C_0$ spanned by  $\{\lambda v: \lambda>0\}$.
\een
Consistently, 
\bit\item A vector $v\in V$  is called  {\em  lightlike} (resp. {\em timelike}, {\em causal} if $v$ lies in $\C_0$ (resp. $A_0$, $\bar A_0$). 
\item A vector hyperplane $H\subset V$ is
called  {\em  lightlike} (resp. {\em timelike}, {\em causal} or {\em spacelike}) if $H$  is tangent to $\C_0$ (resp. contains a timelike vector, lies in one of the previous cases, lies in none of them).
\item  A one-form $\omega: V\rightarrow \R$ is called {\em timelike} (resp.  causal) if $\omega(w)>0$ for all $w\in \bar A_0$ (resp. $\omega(w)>0$ for all $w\in  A_0$). When $\omega $ is causal but non-timelike then it is called  {\em lightlike}  (and ker $\omega$ must be a lightlike   hyperplane). $\bar A^*_0,A^*_0$ and $\C^*_0$ will denote the set of causal, timelike and lightlike forms, resp. 
\eit
\end{defi}
\begin{rem}
(1) The assumptions in the definition of cone can be widely optimized, see \cite[\S 2.1]{JS20}. 

(2) Under our convention, the 0 vector is excluded from the cone. In a natural way $\C_0$ and $A$ generalize the set of    future-directed lightlike and timelike vectors,  resp., in Special Relativity.  However, as we have only one cone,  we have dropped the usual relativistic expression ``future-directed'' for  causal vectors  $v$. 

(3)  Later, causality relations for $\C_0$ will behave as if $-v$ were ``past directed''. As a consequence, the natural definition for a vector $v$ (in $V$ or $V\setminus\{0\}$) to be spacelike would be that neither $v$ nor $-v$ is causal. 

(4) One can check that a vector hyperplane $H$ is spacelike if and only if for one (and then any) causal vector $v$ the affine hyperplane $v+H$ intersects transversely all the directions of $\C_0$ (then, yielding a strongly convex hypersurface in $H$ whenever $n>2$); moreover, any  cone must admit  spacelike  hyperplanes \cite[Lemma 2.5, Prop. 2.6]{JS20}.

(5) Finally, when the kernel of a 1-form $\omega$ is spacelike (resp. lightlike) then either $\omega$ or $-\omega$ is timelike (resp. lightlike).  

\end{rem}

\subsubsection{Lorentz-Minkowski norms}

In \S \ref{sss_wind}, the possibility to define a Lorentz norm associated with a concave indicatrix and a reverse triangle inequality was explained. In the particular case considered therein,  the fundamental tensor $g_v$ had Lorentzian signature but it was defined only in the (timelike) directions in the interior of the cone (see \eqref{e_L0} and Fig. \ref{f_wind}). Indeed  the Lorentz part of the norm could not be extended even  continuously by 0 to the cone $\C_0$, because the indicatrix touched tangently  the cone.  Nevertheless, in the case of a Lorentzian scalar product the unit timelike vectors in one of the cones $\C_0$ do yield a Lorentz norm which can be extended by 0 to $\C_0$ and such that the corresponding fundamental tensor $g_v$ is well defined on $\C_0$ and non-degenrate (thus, with Lorentzian signature therein).  Such  properties are trivial in the case of a scalar product because this product coincides with $g_v$ whenever defined, but they serve as a guide for Lorentz norms with convenient regularity.   

\begin{defi}\label{d_LorentzNorm} A (proper) 
{\em  Lorentz-Minkowski  norm} 
$L_0: \bar A_0 \rightarrow \R$  is a smooth function defined on a cone and its interior domain (i.e.,   $\bar A_0=A_0 \cup \C_0$ $( \subset V\setminus\{0\}$) satisfying:  
\ben \item $L_0(v)\geq 0$ with equality iff $v\in \C_0$
\item Two-homogeneity, that is, $L_0(\lambda v)=\lambda^2 L_0(v)$   for $\lambda>0$.
	
\item The fundamental tensor 	$g_{v}:= $ Hess$_{v} L_0/2$ is non-degenerate  with Lorentz signature $(+,-,\dots . -)$ for all $v\in \bar A_0$.

\een  \end{defi}
   Again, the hypotheses are not optimized, namely, (3) implies that $\C_0$ is a (strongly convex) cone even if one only assumes that $\C$ is the boundary of a domain invariant by positive homotheties  (see \cite[\S 3.1]{JS20}). In any case, the indicatrix $\Sigma_0:= L^{-1}_0(1)$ will behave qualitatively as in Lorentz-Minkowski and it is asymptotic to the cone  $\C_0$ (Fig.~\ref{f_LorentNorm}).

   \begin{rem}  
  As explained in \S \ref{sss_wind}, any Lorentz-Minkowski norm $L_0$ satisfies the strict reversed triangle inequality
  $$ \sqrt{L_0(v+w)} \geq \sqrt{L_0(v)}+\sqrt{L_0(w)}, \qquad \forall v,w\in \bar A_0,  
  $$
 with equality if and only $v,w$ are colinear (then necessarily pointing in the same orientation). Indeed, this property relied in the strict concaveness of the indicatrix. It is worth pointing out that this reverse  triangle inequality permits to re-prove and extend some known inequalities (Acz\'el, Popoviciu, Bellman), see \cite{MPV}.
  \end{rem} 
  
  Beem's pioneering definition \cite{Beem}  changes the hypothesis (1) above by assuming that  $L_0$ is defined on the whole $V$ (with only continuity at 0). 
Our restriction to $\bar A$ is motivated by the following facts:

\bit\item  Physically, the only available directions are the causal (timelike or lightlike) ones. 

\item Mathematically, the causal directions present properties of maximization resembling those of minimization in Riemannian and Finslerian geometries. However, spacelike directions behave very differently.

\item  For Lorentzian scalar products, the behavior of the quadratic form on  any non-empty open set of causal vectors 
determines the behavior on all the vectors.  However, Lorentz-Minkowski norms cannot by  any means be determined in a similar way.

\item    Both definitions are naturally compatible:
\bit 
\item Any Lorentz-Minkowski norm as in our definition can be extended  to a norm as Beem's (in a highly non-unique way, see Remark \ref{r_cone_triples} below or \cite{Min16}). 
\item Conversely, Beem's definition implies the existence of at least one cone fulfilling our definition\footnote{The number of cones for dimension 2 can be arbitrarily big. However, for dimensions $\geq 3$, it is  bounded by 2 at least in the reversible case (that is, when the Lorentz-Minkowski norm is not only positive homogeneous but fully homogeneous), see \cite{Min15}.}. 
\eit
\item From a technical viewpoint, the non-degeneracy of $g_v$ when $v\in \C_0$ ensures that it maintains Lorentzian signature when $L_0$ is smoothly extended to a small neighborhood of $\bar A_0$. This property (generalizable to Lorentz-Finsler metrics below) permits to avoid the burden of considering boundary points for local computations.

\eit
   Summing up, Lorentz-Minkowski norms resembles the behavior of  vectors in a causal cone for a Lorentz scalar product (including the strict reverse  triangle inequality for timelike vectors) and other geometric satisfactory properties, as the following extension of  \eqref{e_Legendre}. 

\begin{defi}\label{d_LegendreLorentz} Let 
$L_0: \bar A_0 \rightarrow \R$ be a {\em  Lorentz-Minkowski  norm} and let $\bar A^*_0$ be the set of all the causal forms (see Def. \ref{d_cone}). Its {\em Legendre} (or {\em flat})  map is defined formally as in \eqref{e_Legendre} adapting the domain and codomain, that is:
$$
\flat : \bar A_0 \rightarrow \bar A^*_0 \qquad (d L_0)_v=v\rightarrow v^\flat:=g_v(v,\cdot).
$$
\end{defi}  
Reasoning as around \eqref{e_Legendre}, 
the indicatrix $\Sigma_0$ yields now a positive-homogeneous diffeomorphism between timelike vectors and forms which extend to lightlike ones, making the inverse $\sharp$ well defined. 
 Notice, however, that for $v\in \Sigma_0$, $v^\flat$ can be identified with the affine hyperplane $v+T_v \Sigma_0 $, while for $u\in\C_0$ the map $u\mapsto u+T_u\C_0$ does not yield and analogous identification (as all the radial affine hyperplanes $\lambda u+T_{\lambda u}\C_0, \lambda>0,$ are equal).

    \begin{rem}\label{r_gradiente_temporal}
  A smooth function $f: V\rightarrow \R$ will be called {\em temporal} when its differential $df$ is a timelike 1-form at each point. Then, its gradient grad$f:=(df)^\sharp$ is well defined and smooth (as in Remark \ref{r_gradiente}). A  natural geometric interpretation also emerges here for the direction of  grad$f$, because, at each point $x\in V$,  ker~$df_x$ is parallel to the tangent space to $\Sigma_0$ at  grad$f /\sqrt{L_0(\hbox{grad} f)}$.
  \end{rem}


	\begin{figure}
\begin{center}

		\includegraphics[width=.7\textwidth]{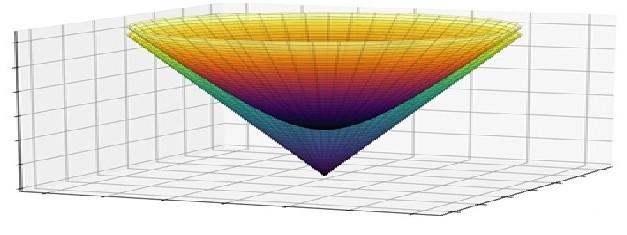}
	\caption{Cone and indicatrix of a Lorentz-Minkowshi norm }\label{f_LorentNorm}

\end{center}	

	\end{figure}

	

\subsection{Cone structures and Lorentz-Finsler metrics}
Here we follow \cite{JS20}, and  \cite{Min15} is also recommended.

\subsubsection{From vector spaces to  manifolds} 
Essentially, a Finsler metric, a wind Finsler structure, a cone structure and a Lorentz-Finsler metric on a (connected) manifold $M$ are smooth pointwise assignations (to each point $p\in M$) of, respectively, a Minkowski norm, a  wind Minkowski norm, a cone and a Lorentz-Minkowski norm (defined above).  This smooth dependence can be neatly expressed in the case of a Finsler metric $F$ with the following definition. Notice that the latter is also is adapted to our preferences of 2-homogeneity ($L=F^2$) and smooth domain, thus using the {\em slit tangent bundle}, obtained by removing the 0 section from $TM$. 

\begin{defi} 
	$L: TM\setminus  \mathbf{0}  \rightarrow \R$ is a {\em Finsler metric} if it satisfies:
	\ben\item Positive : $L(v)>0$, $\forall v\in TM\setminus  \mathbf{0}$. 
	\item  Positive 2-homogeneus : $L(\lambda v)= \lambda^2 L(v)$ for $\lambda>0$ and $v\in TM\setminus  \mathbf{0}$.
	\item  Smooth with positive definite {\em fundamental tensor} $g_v$  
	$$
	g_{v_p}(u_p,w_p)=\frac{1}{2}(\hbox{Hess}_{v_p}L)(u_p,w_p)=\frac{1}{2}\left.\frac{\partial^{2}}{\partial  r  \partial  s} L(v_p+ r  u_p+sw_p)
	\right|_{ r,s=0},$$
	for each $u_p, w_p \in T_pM \setminus{0}$ and $p\in M$.
	
	\een
	\end{defi}
In case of wind and cone structures, however, it is natural  to regard them as hypersurfaces $\Sigma$, $\C$ of $TM$ such that their  restrictions $\Sigma_p$, $\C_p$  to each  $T_pM$ are wind Minkowski norms and cones, resp. There is, however, a subtlety here, because the smoothness of $\Sigma$ or $\C$ would not be enough to ensure the smooth dependence of the structure with $p$ (see Prop. 2.12 and Fig. 2 in \cite{CJS24}) and an additional transversality condition is required.

\begin{defi}\label{d_wind_cone_structure} A {\em Wind Finsler} (resp. {\em cone}) structure  on $M$ is a  hypersurface  $\Sigma \subset TM$ 
	(resp. $\C \subset TM \setminus \mathbf{0}$) such that, for each $p \in M$: 

\bit\item $\Sigma_p$ (resp. $\C_p$)  is a wind norm (resp. cone) at $T_pM$
\item    The submanifolds $\Sigma$ (resp $\C$) and $T_pM$ of $TM$ intersect  transversely.  
\eit 
 Consistently, $\C$ contains the lightlike vectors, the {\em interior} (or cone domain) $A$ of $\C$ is the union of the interiors $A_p$ at all $p\in M$, that is, $A$ contains all the timelike vectors, $\bar A=A\cup \C$ contains the causal vectors, and $A^*,\bar A^*$ and $\C^*$ contain, resp., the timelike, causal and lightlike forms.
\end{defi}

Finally, Lorentz-Finsler metrics can be defined as follows (we refer to the comments to Def. \ref{d_LorentzNorm} as well as  \cite[\S 3.2]{JS20} for optimized hypotheses).

\begin{defi} $L: \bar A\rightarrow \R$ is a 
	{\em Lorentz-Finler} metric, and 
	then $(M,L)$ a {\em Finsler spacetime}, if 
  $\bar A$ is the union of a cone structure $\C$ and its  interior $A$   and the restriction $L_p$ of $L$ 
  to $\bar A\cap T_pM$ is a Lorentz-Minkowski norm for each $p\in M$. 
\end{defi}

\begin{rem}\label{r_gradiente_final}
From Remarks \ref{r_gradiente} and \ref{r_gradiente_temporal}, the definition of  gradient is straightforward for any smooth function $f:M\rightarrow \R$ in the Finsler case and,  in the Lorentz-Finsler one, when $f$ is a {\em temporal function} (i.e. $df_p$ is a timelike form for each $p \in M$)  as well as when $df_p$ is a causal form for all~$p$.
\end{rem}
\subsubsection{Causality and cone triples}
Cone structures $\C $ naturally extend the {\em future} cones of the relativistic spacetimes, that is, time-oriented Lorentzian manifolds. 
Accordingly,  {\em piecewise smooth} curves $\gamma$ are called timelike or causal when so is its velocity everywhere (including the two limit velocities at the breaks). For $p,q\in M$, we write $p\ll q$ (resp. $p < q$) and say that $p$ lies in the chronological (resp strict. causal) past of $q$ or $q$ lies in the chronological (resp strict. causal) future of $p$ if there exists a timelike (resp. causal) $\gamma$ from $p$ to $q$. This permits to define the chronological and causal futures of $p$ as well as its future horismos as, resp.: 
\begin{equation}
\label{e_IJE}
 I^+(p)=\{q\in M: p\ll q\}, \; J^+(p)=\{q\in M: p \leq q\}, \; E^+(p):=J^+(q)\setminus I^+(q)\end{equation}
(and analogously for the notions on {\em past}).
This opens the possibility to  formally extend all the notions of Causality for spacetimes, including the {\em ladder} or {\em hierarchy} of spacetimes (see for example,  \cite{MinSan}). 
Thus, ee will use typical notions of relativistic spacetimes such as {\em causal spacetime} --the one admitting no closed causal curve-- or time function $t$ --continuous function strictly increasing on (future-directed) causal curves-- for cone structures in the remainder, with no further mention.
Moreover:

\begin{defi}\label{d_cone_geodesic} Let $\C$ be a cone structure. A curve\footnote{In general, all the elements are assumed smooth. However, at this point  the continuity of $\alpha$ would suffice.} $\alpha: I\rightarrow M$ is a {\em cone geodesic} if it is {\em locally horismotic}, that is,  
  for each $t\in I$ there exists an open neighborhood $U$ of $\alpha(t)$ such that, whenever $t_1 < t_2$ and $\alpha([t_1,t_2]) \subset U$, one has $\alpha(t_2) \in E^+_U(\alpha(t_1))$ (where $E^+_U$ denotes the future horismos in $U$ computed by considering the restriction $\C_U$ of  $\C$ to $U$ as a cone structure in its own).
\end{defi}

In relativistic spacetimes, $\C$ determines the Lorentzian metric $g$ up to a conformal factor, and the cone geodesics coincide with the future-directed lightlike pregeodesics (that is, the $g$-lightlike geodesics  up to a future-directed reparametrization).  We will see that this last property can be extended to any Lorentz-Finsler metric $L$, even if $L$ is not by any means determined by its cone structure. 
About the latter, notice first the following  nice 
description of cone structures (see  Fig. \ref{f_cone_triple}).

\begin{defi}\label{d_cone_triple} Let $M$ be a manifold. A {\em cone triple $(\Omega, T , F)$ 
for  $M$} 
is composed by a non-vanishing one-form $\Omega$, a vector field $T$ such that $\Omega(T)\equiv 1$ and a Finsler metric $F$ on ker $(\Omega)$.
A cone triple $(\Omega, T , F)$ is  {\em associated with a cone structure  $\C$} on $M$ if: 
\ben  \item $\Omega$ is   $\C$-timelike, that is, $\Omega(v)>0$ for all $v$ causal (Def. \ref{d_cone}). 
\item  $T$ is  $\C$-timelike\footnote{\label{foot_wind_triple} This hypothesis can be dropped by suitably replacing the Finsler metric $F$ in the item (3) by  a wind Finsler structure, see around  \cite[Fig. 4]{CJS24}.}, and
\item $F$ is a Finsler metric $F$ on ker $\Omega$ such that 

$v_p\in \C  \Longleftrightarrow v_p=F(\pi_T(v_p))T_p+\pi_T(v_p)
\qquad \forall v_p\in TM\setminus
\mathbf{0}$,
  
\noindent where $\pi_T: TM\rightarrow \hbox{ker} \, \Omega$ is the natural projection in the decomposition $TM=$ ker $\Omega \oplus$ Span$(T)$.
\een
\end{defi}

For the following result, and discussion, see \cite[\S 2.4, \S 5]{JS20}

\begin{thm}\label{t_cone_triple} 
(A) Any cone structure $\C$ admits admits a  timelike 1-form $\Omega$ and a timelike vector field $T$ satisfying $\Omega(T)\equiv 1$ (both highly non-unique). Then, there exists a unique $F$  so that $(\Omega, T,F)$ is a cone triple for $\C$.

(B) Given a cone triple $(\Omega, T, F)$ for $M$: 

\ben\item There exists a unique cone structure $\C$ such that $(\Omega, T, F)$ is a cone triple  associated with  $\C$.
\item The function $\mathcal{G}: TM\setminus \mathbf{0}\rightarrow \R$, 
$ \mathcal{G}=\Omega^2-(F\circ \pi_T)^2$,
satisfies: 

(a) It is two-homogeneous  and smooth everywhere,  but in  Span$(T)$. 

(b)  $\mathcal{G}^{-1}(0)$ is composed by two cones  with cone triples $(\pm \Omega, T, F)$. 

(c) Whenever $\mathcal{G}$ is smooth, its fundamental tensor $g$ has Lorentzian signature. 

(d) $\mathcal{G}$ is smoothable in any arbitrarily small neighborhood of Span$(T)$ maintaining the Lorentz signature of $g$. Thus, the smoothening is a Lorentz-Finsler metric defined on the whole slit tangent bundle (in agreement with Beem's definition).
\een
\end{thm}

\begin{figure}
\begin{center}
		\includegraphics[width=0.72\textwidth]{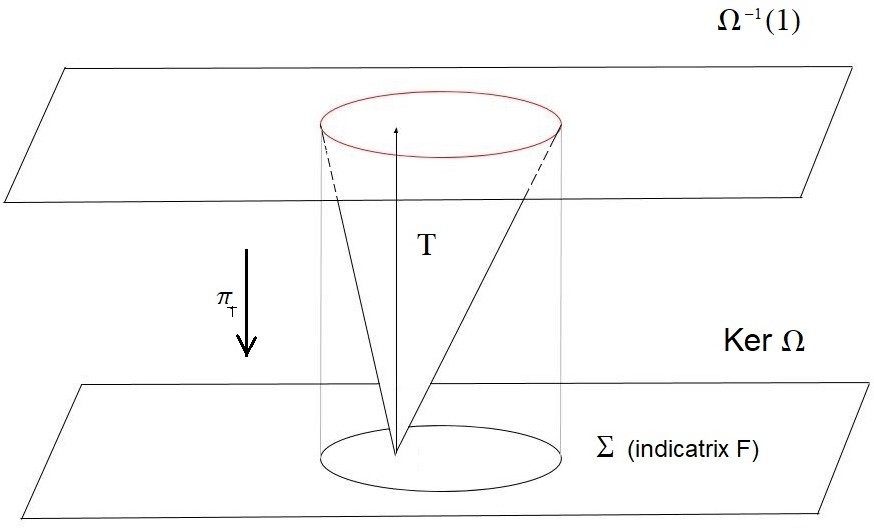}
\end{center}		
\caption{$\C$ determined by a cone triple $(\Omega,T,F)$. }\label{f_cone_triple}
\end{figure}

As we will see, cone triples are especially  useful  when $\Omega=d\tau$ for a temporal function 
(see definition  in Remark \ref{r_gradiente_final}) 
and $T$ is regarded as a physical field of  observers.

\begin{rem}\label{r_cone_triples}
It is worth pointing out
that,  given $\C$ and the timelike 1-form 
$\Omega$, there is a unique timelike direction 
which crosses the barycenter of 
$\bar A \cap \Omega^{-1}(t)$ for all $t>0$. 
However, the above generality for  $T$ (even amplified as suggested in footnote \ref{foot_wind_triple}) will be convenient. 

In particular, the smoothened metric $\mathcal{G}$ obtained from different choices of $T$   in Th. \ref{t_cone_triple} (last item) exemplifies the high-non uniqueness of the Lorentz-Finsler metrics compatible with a single cone structure.  
\end{rem}

\subsection{Finsler distance and   exponential,   Lorentz-Finsler causality}\label{s2_3_1}

\subsubsection{Finsler distance and geodesics}\label{ss_distance_geodesics}
Extending Riemannian notions, the {\em length} and {\em energy} functionals on curves $\alpha: [a,b]\rightarrow M$ for a Finsler metric $L=F^2$ are, respectively,  
$$\ell (\alpha)= \int_a^b F(\alpha'(s))ds, \qquad \mathcal{E} (\alpha)= \frac{1}{2}\int_a^b L(\alpha'(s))ds,$$ 
and the length is invariant under monotonically {\em increasing} parametrizations. This permits to define 
a generalized (possibly non-symmetric) distance, that is, a usual distance $d$ except for  $$d(p,q)\neq d(q,p) \qquad \qquad \hbox{in general, for } p,q\in M.$$ As in the Riemannian case, geodesics can be defined as the  critical points for the energy functional on curves connecting any two prescribed points.
Their Euler-Lagrange equation is then formally analogous to the Riemannian one, 
\begin{equation}\label{e_geodesicas}\ddot x^k +  \gamma_{ij}^k (x,  \dot x )\dot x^i \dot x^k=0. \end{equation}
This expression can be derived directly (see for example \cite[Section 3]{Dahl})  and will agree with the Lagrange equations for the Lagrangian $L/2$ in \S \ref{ss_contact}. 
Here, the {\em formal Christoffel symbols} $\gamma^k_{ij}$ are computed with the same  expression as in the (semi-) Riemannian case, but applied now to the fundamental tensor $g$. However, notice that  $g=g_{ij}(x,y)$,  depends not only on the point $p\in M$   we are considering (with coordinates $x=(x^1, \dots x^n)$), but also on the oriented direction of the tangent vector $v_p= \sum_i y^i\partial_{x^i}$ (with natural coordinates $(x,y)$ in $TM$).  
As in the Riemannian case, Finslerian geodesics locally minimize the distance. 

\subsubsection{The exponential map, local and global causality}
Using geodesics, one can define the exponential map at a point  $\exp_p$ for any semi-Finsler metric in a similar way as in the Riemannian case, and the existence of convex neighborhoods is ensured by a result of Whitehead \cite{Wh}. However,  the lack of regularity at 0 of $F$ implies that the symbols  $\gamma_{ij}^k(x,y)$ may not even be continuously extended to $y=0$ (indeed, $\exp_p$ is smooth at $0$ only in very particular cases, see \cite{AZ88}). 
Anyway, the exponential map $\exp$ is always\footnote{This is common knowledge at least since F.W. Warner \cite[Th. 4.5]{Wa}, see also explicit H .~Koehler's  \cite[Prop.~4.1]{Koe14}, where  S.~Ivanov is cited as the author.} $C^{1,1}$ and, following R. Bryant\footnote{\label{footRobert}See   \cite[Th. 4.5]{Wa} and   
\href{https://mathoverflow.net/questions/479094/how-badly-does-the-geodesic-exponential-map-fail-to-be-c2-on-finsler-manifold}{MathOverflow (10/02/2026)}.  Expansions of the metric in normal coordinates using curvature 
can be found in \cite{Pfeifer}.}, all the derivatives can be also bounded. Indeed, there exists a smooth map
  $E:(-\epsilon,\epsilon)\times \Sigma_p\rightarrow \R^n$ such that 
 \begin{equation}\label{e_Bryant1}
\exp_p(ru)=ru+r^2E(r,u) \qquad \forall u\in \Sigma_p, \quad   r\in [0,\epsilon), 
 \end{equation}
and   the curve 
 \begin{equation}\label{e_Bryant2} \gamma_u(t)= t u + t^2 E(t,u), \qquad \hbox{for} \,  |t|<\epsilon , 
 \end{equation}
is a smooth geodesic. Such estimates can be transmited to the Lorentz-Finsler case, see the proof of \cite[Lemma 6.6]{JMPS25}.

Then, the extension of this Finsler setting to the Lorentz-Finsler one turns out similar to the extension from Riemann to Lorentz manifolds,   even though it is  technically more involved, see  \cite{AJ16,  CJ_Gelomer,  Min15b, Min17b} (some of them sharpening usual Lorentzian regularity).  In particular, when $L$ is Lorentz-Finsler, timelike geodesics become the local maximizers of the energy functional. Morever, lightlike $L$-pregeodesics agree with cone geodesics, thus, the former depends only on the associated cone structure\footnote{It is worth pointing out that the conformal invariance of lightlike pregeodesics is valid in any signature (where causality properties cannot be claimed). This occurs both,  in the semi-Riemannian and in the semi-Finsler cases (in the latter, even anisotropically conformal changes are  permitted), see \cite{JS20}.}  $\C$. 
The regularity of $\exp$ is enough to ensure that, locally, the geometry and causality of a Lorentz-Finsler manifold mimics the one of a Lorentz Minkowski norm (thus, the one of a Lorentz-Minkowski scalar product). 
  In particular, (as in 
\cite[Prop. 5.34, Lemma 14.2]{O}):  
\begin{prop}
Let $(M,L)$ be Lorentz Finsler,  $U$  a normal neighborhood  of $p\in M$ and $q=\exp_p u\in U$. 
\bit\item
{(\em Local causality.)} Exists a timelike (resp. lightlike) curve from $p$ to $q$ in $U$ if and only if  $u$ is timelike (resp. lightlike) 
\item {\em (Local maximization.)} In this case, the longest causal curve (up to parametrization) from $p$ to $q$ in $U$ is the geodesic 
$[0,1] \ni t\mapsto \exp_p(tu)$.  
\eit \end{prop} 
Then, the basic causality relations $I^\pm, J^\pm, E^\pm $ (see \eqref{e_IJE}) and {\em time-separation} (Lorenzian distance) maintain the properties for relativistic spacetimes. The {\em causal ladder} is also consistently defined and mimicks the standard  Lorentzian one.

\subsection{Nonlinear, anisotropic and Finsler connections}
Next, our basic reference is \cite{JSV_Gelocor} (which takes into account  progress in previous \cite{J19, J20}), apart from standard books on Finsler Geometry such as \cite{BCS,Sh}.

For the tangent bundle $TM\rightarrow M$, natural coordinates  $(x,y)$ $=(x^1, \dots x^n, $ $ y^1, \dots y^n)$,  induced from  coordinates $x=(x^1, \dots , x^n)$ on $M$, are used.
The derivatives along the fiber (type $\partial_{y^k}$) will be called {\em vertical}. We  will consider  $\bar A\subset TM$, the (causal) domain of a Lorentz-Finsler metric $L$, but also   applies to the slit tangent bundle $TM\setminus \mathbf{0}$ of a (semi-)Finsler one.

\subsubsection{Geodesic spray}
Starting at geodesics, a number of geometric structures emerge by taking vertical derivatives (then, decreasing the order of homogeneity), we will follow \cite[\S 6.1]{JSV_Gelocor}.
The geodesic equation \eqref{e_geodesicas} is naturally associated with the {\em geodesic spray} $G$, a vector field defined on the tangent space  to the cone domain  ($T\bar A$ included in the second order tangent space  $T(TM)$)\footnote{The following notions are defined locally, so that they make sense  on  open subsets of $TM$. This includes not only $A$ but also $\bar A$ because $L$ can be locally extended beyond $\bar A=\C$ (recall the last item below Def. \ref{d_LorentzNorm}).}.  $G$ is defined  as:

\begin{equation} \label{e_Geod}
	{G}_{(x,y)}=y^{a} \frac{\partial}{\partial x^{a}}|_{(x,y)} - 2 {G}^{a}(x,y)\frac{\partial}{\partial y^{a}}|_{(x,y)}, \quad \hbox{where} \; 2{G}^{a}(x,y)=\gamma_{ij}^{a}(x,y)y^{i}y^{j} 
\end{equation}
(always $i,j,k,a= 1, \dots , n$, and sum in repeated indices). 
Notice that  each function  $G^a$ is positive 2-homogeneous in $y$  (i.e., $G^a(x,\lambda y)=\lambda^2 G^a(x,y)$ for $\lambda >0$.). Then,  there exist positive  1-homogeneous functions ${N^a_i}$ such that 
\begin{equation} \label{e_Non_linear}  {N^a_i}(x,y)= \frac{\partial G^a}{\partial y^i}(x,y), 
\end{equation}
and the {\em Euler theorem} ensures
$2G^a(x,y)=N^a_i(x,y)y^i$.

\subsubsection{Nonlinear  connections}
Now, the functions ${N^a_i}$ can be regarded as the components of a {\em nonlinear connection} ${\nu}$  on 
the  fibered manifold   $\bar A\rightarrow M$. Indeed, $T\bar  A$ admits a natural vertical subbundle $\Ver \bar A$ (composed by the velocities of curves included in each fiber) and a  nonlinear connection can be defined as a vector bundle homomorphism 
 $\nu: T\bar A\rightarrow \Ver \bar A$ such that $\nu|_{\Ver \bar A}$ is the identity. Then, the distribution $\Hor \bar A:=$ ker$(\nu)$ also characterizes $\nu$ and gives a decomposition $T\bar A=\Hor \bar A \oplus \Ver \bar  A$. In our case, using natural coordinates $(x,y,\dot x, \dot y)$ in $T(TM)$:
\begin{equation}
\label{e_delta}
\Hor_{(x,y)}A= \hbox{Span}_{\{i=1,\dots ,n\}}\left\{\left. \frac{\delta}{\delta x^i}\right|_{(x,y)} := \left. 
\frac{\partial}{\partial x^i}\right|_{(x,y)} - N^a_i(x,y) \left. \frac{\partial}{\partial y^a}\right|_{(x,y)}
\right\},
\end{equation}
which is invariant by the homotheties $h_\lambda: A \rightarrow A$, 
$v_p \mapsto \lambda v_p$, 
i.e., 
$\Hor_{(x,\lambda y)}A$ $=
d(h_{\lambda} )_{(x,y)}(\Hor_{(x,y)}A)$, for $\lambda>0$.

\subsubsection{Cartan tensor}
Computing explicitly from \eqref{e_Non_linear}, \eqref{e_Geod},
\begin{equation}\label{e_nonlinear_explicit}
{N^a_i}=\gamma^a_{ij}y^j-g^{aj}C_{ijk} \gamma^k_{lm}y^ly^m,
\end{equation}
where $C_{ijk}=\partial_{y^k}g_{ij}/2$ are coordinates of the {\em Cartan tensor} $C$ of $L$. Notice that, in the semi-Riemannian case $g_{ij}$ depends only on the point, thus $C_{ijk}=0$. In general, at each point $p$, $C$ is the differential (vertical derivative) of $g_{ij}$ restricted to $T_pM$ that is, $C_p$   measures how far  $g_{ij}(x(p),\cdot )$ is  from being semi-Euclidean.

\subsubsection{Anisotropic connections}\label{ss_nonlinearconncetions}
As the functions $N^a_i$ are positive homogeneous, Euler theorem also yields 
\begin{equation}\label{e_homogN}
	N^a_i(x,y)=\partial_{y^j} N^a_i(x,y) y^j  =  \mathring{\Gamma}_{ij}^a(x,y) y^j ,
\end{equation} 
where
$ \mathring{\Gamma}_{ij}^a(x,y)  :=\partial_{y^j} N^a_i(x,y)$ is positive $0$-homogeneous in $y$, that is, $\mathring{\Gamma}_{ij}^a(x,\lambda y)=\mathring{\Gamma}_{ij}^a(x,y)$ for $\lambda>0$. One could  define a covariant derivative by using the nonlinear  connection of a (causal) vector field $Z$ in the direction of a coordinate vector, namely\begin{equation}\label{e_dc}
\begin{array}{rl}
    D_i Z(x) & = \left( \frac{\partial Z^a}{\partial x^i}(x) + N^a_i(x, Z(x))\right) \partial_{x^a} 
    \\ & = 
 \left( \frac{\partial Z^a}{\partial x^i}(x) + \mathring{\Gamma}_{i j}^a(x, Z(x)) Z^j(x) \right) \partial_{x^a}.
\end{array}
\end{equation}
Notice, however, that this covariant derivative would not be linear in $Z$ because  $\mathring{\Gamma}_{i j}^a$ is also dependent of  $y=Z(x)$.\footnote{Such a derivative has some applications. For example, in \cite[\S 7]{JSV_Gelocor}),  it is used to define a (non-linear) parallel transport $Z$  for a causal vector $Z_x$ along a curve $\gamma$, regarding $Z(x)$ as an unknown. Then, the so-obtained parallel-transported  $Z$  is used to define a {\em linear} parallel transport along $\gamma$ by using $\nabla^Z$ defined below. } 

In case that  $\mathring{\Gamma}_{ij}^a (x,y) \equiv \Gamma_{ij}^a(x) $  
is fully independent of $y$, we can put 
$$
N^a_i(x,y)=
{\Gamma}_{ij}^a(x) y^j.$$ Then,  the functions  $\Gamma_{ij}^a(x)$ obey the standard rule of transformation of the Christoffel symbols, and  the non-linear connection becomes  {\em affine}
(or {\em linear}). The expression \eqref{e_dc} can be seen as a covariant derivative $\nabla$ on vector fields 
$(X,Y) \mapsto \nabla_X Y$ satisfying the expected elementary properties of  linearity  and Leibniz rule. In particular, it turns out  that $(\nabla_X Y)_p$ only depends of the value of $X$ at $p$, i.e., $X_p=\bar X_p \Longrightarrow (\nabla_X Y)_p=(\nabla_{\bar X} Y)_p$.

An {\em anisotropic connection} on the bundle $\bar A \rightarrow M$ can be defined in coordinates, by using functions $\Gamma^a_{\,\,  ij}(x,y)$
which depend on the oriented direction but satisfy the usual cocycle rule
\begin{equation}
\label{e_cocycle Chris}
 \bar \Gamma^a_{\,\, ij }(\bar x(x,y), \bar y(x,y))=
\frac{\partial  \bar x^a}{\partial  x^m}(x) \left(
\frac{\partial^2    x^m}{\partial \bar x^i \partial \bar  x^j}(x) + \frac{\partial  x^k}{\partial \bar x^i}(x) \, \frac{\partial  x^l}{\partial \bar x^j}(x) \, \Gamma^m_{\,\, kl } (x,y)
\right). 
\end{equation}
Indeed, now the anisotropic connection $\nabla$ applied to the coordinate vector fields $\partial_i, \partial_j$
is a {\em vertical vector field on $\bar A$} defined as:

\begin{equation}\label{e_christoff}
\nabla_{\partial_i}\partial_j|_{(x,y)}=\Gamma^a_{\,\, ij  }(x,y)\, \partial_{y^a}|_{(x,y)} \in \mathcal{V}(TM),
\end{equation}
with independence of the chosen   natural coordinates $ (x,y)$ on  $\bar A$.
By imposing usual linearity and Leibniz rules,  one defines $\nabla_X Y$ for any vector fields $X,Y$ on $M$ obtaining an {\em anisotropic connection} $\nabla$ (in coordinates). 

Now, let us rewrite \eqref{e_christoff}  in a fashion closer to affine connections. Choose
any  vector field  $Z: M \rightarrow  \bar A$ (that is, a causal vector field in the Lorentz-Finsler case or any non-vanishing vector field on $M$  in the Finsler case), and consider then the composition $$\nabla_XY \circ Z: M\rightarrow \mathcal{V}(TM).$$ Identifying naturally each vertical vector $\nabla_XY\circ Z|_{p}$ with a vector, denoted 
$\nabla^Z_XY|_p$ which belongs to $T_pM$ we obtain finally:  
 whenever a (causal) vector field  $Z$ is prescribed, one has  an {\em affine connection,
\begin{equation}\label{e_ca}
(X,Y) \rightarrow \nabla^Z _X Y,
\end{equation}
depending smoothly on $Z$}. Notice the notation and  identifications: 
$$\bar A \ni Z_p\mapsto \nabla^{Z_p}_XY:= 
\nabla_XY|_{Z_p} \in \mathcal{V}_{Z_p}(TM)
\equiv T_pM  \subset TM.$$

Indeed, \eqref{e_ca} can be used to define an {\em anisotropic connection} in an abstract (coordinate free) way, the  formal  definition can be found in \cite[\S 3]{JS20}. 
This abstract definition also recovers the fact (evident from above) that  $\nabla^Z _X Y$ depends, at each point $p$,  on the local value of $Y$ around $p$ and  only on the values of $Z$ and $X$ at $p$.  

Next, our aim is to see first that any non-linear connection $N^a_i$ has a naturally associated anisotropic connection $\Gamma^a_{\,\, ij }$ but, when $N^a_i$ comes from a semi-Finsler metric, a second anisotropic connection appears.

\subsubsection{Anisotropic connections of Berwald and Chern} \label{ss_Anisoropicconnections}
 A semi-Finsler metric  naturally yields two anisotropic connections defined  as in \eqref{e_christoff}, in terms of  Christoffel symbols (satisfying \eqref{e_cocycle Chris}): 
 
 \bit\item {\em Berwald anisotropic connection}. ts Christoffel symbols are  $\mathring{\Gamma}_{ij}^a(x,y) $ as in \eqref{e_homogN}  and they can be defined for any nonlinear connection $N^a_i$.
  In our case, recall that $N^a_i$ came from the formal ones $\gamma_{ij}^k(x,y)$   obtained from the fundamental tensor $g_{ij}(x,y)$   as in semi-Riemannian geometry (indeed, they appeared in the  spray \eqref{e_Geod}, thus in \eqref{e_Non_linear},  \eqref{e_nonlinear_explicit}).

When the anisotropic Berwald connection is  affine, the (semi-) Finsler manifold is called {\em  Berwald}.  This is characterized by the vanishing of the Berwald tensor (positive homogeneous of degree -1), obtained by taking the vertical derivative of the Berwald anisotropic connection.  In the positive definite case, Szab\'o \cite{Sz} showed that the associated affine connection of a  Berwald manifold is the Levi-Civita one of a Riemannian metric (see \cite{FHPV} for the indefinite case).
  
 \item {\em Chern anisotropic connection}.  its Christoffel symbols  $ ^C\Gamma_{ij}^a $  are computed as the  formal ones\footnote{Notice 
 that the  $\gamma_{ij}^a(x,y)$'s do 
 not satisfy the cocycle rule in \eqref{e_cocycle Chris}, but the 
 transformation obtained replacing in the left hand side therein 
 $\bar y(x,y)$ by $y$.} 
 $\gamma_{ij}^a(x,y)$
but now replacing the derivatives $\partial/\partial {x^i}$ by the corresponding derivatives $\delta/\delta x^i$ in \eqref{e_delta}.
It plays a role similar to the Levi-Civita connection in semi-Riemannian geometry. Indeed, it  is the unique symmetric anisotropic connection parallelizing $g$, \cite[Sect. 4.1]{J19}. Another interpretation in terms of $\nabla^Z$ can be found in \cite[Proposition 3.9]{J20}; in 
particular, if $Z$ is taken parallel at a point $x$ (a property which depends only on the nonlinear connection, see \eqref{e_dc}), then 
the Christoffel symbols of the semi-Riemannian metric $g_Z$ agree at $x$ with those of Chern's anisotropic connection at $(x,Z_x)$, see 
\cite[formula (7.17)]{Shen01}.  

The difference between the Chern and Berwald anisotropic connections is (up to metric equivalence) the so-called   {\em Landsberg tensor}. When this tensor vanishes, the semi-Finsler manifold is called {\em Landsberg}.  (The vertical derivative of the Chern anisotropic  connection also yields a tensor, which is the sum of the Berwald tensor and the vertical derivative of the Landsberg one.) 

  \eit
  Explicitly, the Berwald and Landsberg tensors, as well as  the contraction of the latter called {\em mean Landsberg}, are written in coordinates (starting at the geodesic spray in \eqref{e_Geod}), respectively as follows: 
  \begin{equation}\label{e_Berwald_et_al}
B^a_{ijl}:= 
{G}^{a}_{\cdot i\cdot j\cdot l}  = \mathring{\Gamma}_{ij\cdot l}^a  , \qquad 
  {\mathrm{Lan}_{ijl}}:=\frac{1}{2} y_a B^a_{ijl}, \qquad {\mathrm{Lan}_l}:=g^{ij} \mathrm{Lan}_{ijl}
\end{equation}
  where the dot symbol denotes vertical derivative (that is,  $\cdot \alpha$ means $\partial_{y^\alpha}$). 
  Notice that  $y_a:=y^b g_{ab}$,  showing the Landsberg dependence on the semi-Finsler metric and not only on the geodesic spray. 
  Moreover, the difference between  the anisotropic connections $\mathring{\Gamma}_{ij}^a$ and $^C{\Gamma}_{ij}^a$  becomes $g^{ab} \mathrm{Lan}_{ijb}$. 

  
   Summarizing, these tensors vanish  for  semi-Riemannian manifolds, 
the Berwald tensor measures  to  what extent the nonlinear connection is truly non-linear (i.e., non-affine) and the Landsberg one  to  what extent the aniso\-tropic Chern connection is not determined only by the non-linear connection.  

  Trivially, a Berwald manifold ($B^\mu_{\nu\rho\sigma}\equiv 0$) is a Landsberg one (${\mathrm{Lan}_{\mu\nu\rho}}\equiv 0$), and the latter is a weakly Landsberg one (${\mathrm{Lan}_\mu}\equiv 0$),   however, the converse to the first one is a major open question   in Finsler Geometry.
   Indeed, following D.~Bao, Landsberg non-Berwald manifolds are called {\em unicorns}. Unicorns  {\em with singularities} have been found in both Finsler \cite{Elgendi} and Lorentz-Finsler manifolds \cite{Heefer, HPRF} (see \S \ref{ss_unicorn} below). However, the inexistence of (regular) unicorns, or {\em Landsberg-Berwald conjecture}, is  a major open question in  Finsler Geometry, see \cite{Crampin}. 

	\subsubsection{Finsler connections}\label{ss_Finslerconnections}
A step forward is to define {\em Finsler connections}. Recall that the nonlinear connection $\nu$ was defined for the bundle $\bar A\mapsto M$ and provided the decomposition  $T\bar A=\Hor \bar A \oplus \Ver \bar A$. Finsler connections are linear connections for the bundle $\Ver \bar A \rightarrow \bar A$, also naturally invariant by homotheties $h_\lambda$. In order to specify the linear connection as a horizontal distribution, the decomposition of $T\bar A$ provided by $\nu$ reduces it to locally specify {\em horizontal} and {\em vertical} functions $H_{ij}^k, V_{ij}^k$ on $\bar A$ satisfying, respectively, the cocycle transformation of coordinates for an anisotropic  connection  and  a  tensor. The natural options for   
$H_{ij}^k$ are the Christoffel symbols of  the  Berwald and Chern anisotropic connections, while those for $V_{ij}^k$ are the zero and Cartan tensors. This yields the four classical linear connections:
\ben\item Berwald: $H_{ij}^k= \mathring{\Gamma}_{ij}^k$ (Christoffel symbols for Berwald), $V_{ij}^k=0$. 
\item Chern: $H_{ij}^k= \; ^C \Gamma_{ij}^k$ (Christoffel symbols for Chern), $V_{ij}^k=0$.   
\item Hasiguchi:  $H_{ij}^k= \mathring{\Gamma}_{ij}^k$,   $V_{ij}^k=  C_{ij}^k$   (Cartan tensor). 
\item Cartan: $H_{ij}^k= \; ^C  \Gamma_{ij}^k$,   $V_{ij}^k=C_{ij}^k$. 
\een
Berwald and Chern linear connections are {\em vertically trivial} (i.e. $V_{ij}^k=0$), thus,  equivalent to their anisotropic counterparts. 
    Indeed, the (anisotropic) Christoffel symbols $ \Gamma_{ij}^k(x,y)$ for  the latter ones can be regarded as the corresponding (linear)
      Christoffel symbols for the former ones (see formula (30) in \cite{SV}, where $\Delta$ therein vanishes for a vertically trivial connection). Given a curve  $s\mapsto Z(s)=(x(s),y(s))$ in $A$, one has naturally  a parallel transport associated with the linear connection,  as well as an anisotropic parallel transport associated with $\nabla^Z$, the latter with further possibilities see \cite[Sect. 7]{JSV_Gelocor}. 
\subsubsection{Curvatures}\label{ss2_5}

In Riemannian geometry, curvature arises uniquely from the Levi-Civita connection, but in Finsler geometry curvature must be defined either via the  nonlinear connections or via the non-unique linear connections on $\Ver \bar A \rightarrow \bar A$.

In general, the \textit{curvature}, the \textit{(anisotropic) Ricci scalar} and the \textit{torsion} of a homogeneous nonlinear connection $N$ can be regarded as homogeneous 
tensors with coordinates (always depending on $(x,y)$):
\begin{equation}
	\RN_{ij}^a=\delta_j N_i^a-\delta_i N_j^a,\qquad \mathrm{Ric}=  y^b\,\RN_{ba}^a   ,\qquad\torN_{ij}^a= \partial_{y^j}N_{i}^a- \partial_{y^i} N_{j}^a .
	\label{eq:curvature and torsion}
\end{equation}
The tensor $\RN_{\;  j}^k= y^b \RN_{bj}^k$, whose contraction yields the Ricci scalar, may be  regarded as a {\em predecessor} of the flag curvature below  (see \cite[Ex. 2.5.7]{BCS}). By \eqref{e_Non_linear}, the torsion vanishes for the nonlinear connection  from a geodesic spray.  It is worth noticing
$$
\left[\frac{\delta}{\delta x^i},
\frac{\delta}{\delta x^j}\right] = \; 
\RN_{ij}^a \; \frac{\partial}{\partial y^a} \; , 
$$
so that the horizontal distribution is integrable around a point $(x,y)\in \bar A$ if and only if $\RN_{ij}^a \equiv 0$ around $(x,y)$. 

For Finsler linear connections, consistently with the horizontal and vertical splitting in their definitions, the curvatures have three parts, labelled $hh$, $hv$, $vv$ \cite[p. 50]{BCS}. Focusing on the 
  Chern connection, the $vv$ part vanishes and the $hh$ one is described by  the {\em first Chern curvature tensor} \(R_j{}^i{} _{k  l}(x,y)\). The (0-homogeneous) components of this tensor are computable from the Chern Christoffel symbols $\Gamma^k_{ij}$ in a way formally analog to the Riemannian one, but again replacing $\partial/\partial{x^i}$ by $\delta/\delta {x^i}$ (see formulas (3.3.2) and (3.3.3) in \cite{BCS} for the explicit expression of \(R_j{}^i{} _{k  l}(x,y)\) as well as the remaining {\em hv} part of the curvature). Remarkably, \(R_j{}^i{} _{k  l}(x,y)\) agrees  (up to sign) with the {\em curvature tensor for the Chern anisotropic connection}, as defined in \cite{J19}. 
  
  From this tensor, 
   one can define the Finslerian analogue of sectional curvature, namely,  the \emph{flag curvature},
\[
K(x,y,\Pi) = \frac{v^i ( y^j R_{j i k l} \, y^l) v^k }{ g_y(y,y)g_y(v,v) - g_y(y,v)^2 },
\]
Note that $K(x,y,\Pi)$
 depends on the direction (flagpole) $y$
and the plane \(\Pi = \mathrm{span}\{y,v\}\) (flag) but not on the choice of $v$. The Ricci scalar $\mathrm{Ric}(x,y)$ can be computed as an average of the flag curvatures with flag pole $y$. In the literature, a Ricci tensor $\mathrm{Ric}_{ij}(x,y)$ introduced by Akbar-Zadeh and containing the same information of the Ricci scalar is also used \cite[p. 192]{BCS}.
Moreover, a weighted Ricci  curvature was introduced by Ohta \cite{Ohta, Ohta2},  in order to adapt  the well-known N-Bakry-\' Emery-Ricci tensor on a weighted Riemannian manifold (see for example \cite{Lott}) to the Finsler setting.
It plays a relevant role in comparison  theorems and yields splitting results \cite{Ohta1} extensible to the Lorentz-Finsler case \cite{LuMinOhta2}. 
    
\subsection{Symplectic and contact viewpoints}\label{ss_contact} 
This part follows the classical setting  in Mechanics \cite{AbrahamMarsden, Arnold}, adapted when necessary to the Finsler \cite{Dahl} and Lorentz-Finsler \cite{Hedickepreprint} cases.   
Notice that the usual notation  in Mechanics for coordinates in the  bundles $TM$ and $T^*M$ are rewriten here as $q^i,\dot q^i\equiv x^i,y^i$ and  $q^i, p_i\equiv x^i,p_i$, resp. Morever, in Finslerian Mechanics one tipically uses $TM\setminus \mathbf{0}$ and its dual bundle   and we will be particularly interested in the set of causal vectors 
$\bar A\subset TM\setminus \mathbf{0}$ and its   dual $\bar A^*$ of causal forms.  

\subsubsection{Basic Hamiltonian and Lagrangian approach}
The cotangent bundle $\hat \pi: T^*M\rightarrow M$  admits the {\em tautological form}, written in bundle  coordinates $(x^i,p_i)$ (with sum in repeated indexes)   
 
$$\theta= 
p_i dx^i  
, \quad \hbox{that is}, \quad \theta(y^i \partial_{x^i}+\dot p_j \partial_{p_j})|_{(x,p)}= p(y^i \partial_{x^i}|_x)= p_i y^i ,
$$
where $(x^i,p_i, y^i,   \dot p_i)$ are coordinates in $T(T^*M)$ and $y^i \partial_{x^i}=\hat \pi_*(y^i \partial_{x^i}+\dot p_j \partial_{p_j})$, making $\theta$ independent of coordinates. Then, the {\em Poincar\'e  2-form},
$$
\omega = -d\theta= dx^i \wedge dp_i . 
$$
 is {\em symplectic} i.e., closed with $\omega \wedge  \dots ^{(n)}\wedge \omega \neq 0$.  By Darboux theorem, every symplectic form looks  like the above in suitable coordinates. The fiber of $TM^*$ at each $p\in M$ is a {\em Lagrangian submanifold}, i.e., $\omega\equiv 0$ on the submanifold and it has maximum dimension $n$.
 
 Moreover,  $\omega$
also yields a linear isomorphism ({\em symplectic gradient})
$$
T(T^*M) \longrightarrow T^*(T^*M), \quad X\mapsto i_X\omega (\equiv \omega(X,\cdot )). 
$$  
Notice that this map restricts naturally when $T^*M$ is replaced by $\bar A^*$.  Thus, it permits to define, for any function $H:  \bar A^*\subset  T^*M\rightarrow \R$ its {\em Hamiltonian vector field $X_H$}, with integral curves $t\to (x(t),p(t))$,  characterized by
 $$\begin{array}{c}
 dH=
 i_{X_H}\omega, \; \;\; \hbox{then,} \; \; \; X_H=
 \frac{\partial H}{\partial p_i} \frac{\partial }{\partial x^i} 
 -
 \frac{\partial H}{\partial x^i} \frac{\partial }{\partial p_i}, 
\; \hbox{and }  \; y^i=\frac{\partial H}{\partial x^i}, \, \dot{p}_i= -
 \frac{\partial H}{\partial x^i},  
 \end{array}$$   
the last pair being the {\em Hamilton equations}. The flow of $X_H$ preserves the symplectic structure and the {\em Hamiltonian} $H$ becomes constant along the integral curves of $X_H$.

A (2-homogeneous) Finsler or Lorentz-Finsler metric   
must be regarded as a {\em Lagrangian} $L/2: \bar A\subset TM\rightarrow \R$ (the factor 1/2 being convenient). 
It permits to define the {\em generalized momenta} in coordinates
$$p_i:=\frac{1}{2} \frac{\partial L}{\partial y^i}= \frac{1}{2}\frac{\partial g_{jk}}{\partial y^i} y^jy^k+ g_{ik}y^k,$$ 
so that $p_idx^i$ recovers the Legendre map (Def. \ref{d_LegendreLorentz})  in our case:
\begin{equation}\label{e_LegendreLorentz2}
\flat: \bar A\rightarrow \bar A^*, \qquad  y^j \partial_{x^j}|_x
 \mapsto p_i dx^i|_x = g_{ij}(x,y) y^j  d x^i|_x ,
\end{equation}
just noticing that, by the $0$-homogeneity of $g_{ij}$ and Euler's theorem, 
$$
\frac{1}{2} \frac{\partial g_{jk}}{\partial y^i} y^jy^k \left(dx^i(y^l \partial_{x^l})\right)= \frac{1}{2}
y^jy^k \left( \frac{\partial g_{jk}}{\partial y^i} y^i\right)=0. 
$$
The standard {\em Hamiltonian $H_L$ associated with $L$} is its {\em Legendre transform},
$$
H_L(x,p)=p_i \, y^i(x,p)-\frac{1}{2}L(x,y(x,p))=\frac{1}{2}L(x,y(x,p)).
$$ 
Thus $H_L=L\circ \sharp$ and the Legendre map sends bijectively 
the  solutions $\gamma$ of the Lagrange equations  \eqref{e_geodesicas}  
into solutions $(\gamma')^\flat$ of Hamilton's preserving the (kinetic and total) energy $E=L(\gamma')/2=H((\gamma')^\flat)$.  

\subsubsection{The Hilbert form}\label{ssHilbert_form} Formula \eqref{e_LegendreLorentz2} can be alternatively interpreted as   defining   a  1-form $\eta _g$ on $\bar A$,   the {\em Hilbert form},   obtained   as   the pullback\footnote{In the Finsler case some times this is divided by $F(x,y)$. However, this would introduce an issue in the Lorentz-Finsler case for lightlike geodesics. } 
$$ \eta _g:=\flat^*\theta= 
g_{ij}(x,y)y^i dx^j|_{(x,y)}.
$$ 

In a natural way, $\omega_g=-d\eta_g$ provides a symplectic form on $T\bar A$ (see  \cite[Prop.7.13]{Dahl}). Then, the $\omega_g$-Hamilton equations for $L$ (regarded as a Hamiltonian function) are equivalent to the geodesic equations for $(M,L)$, which turns out the Lagrange equations for the Lagrangian $L/2$. 
Indeed,  
the $L$-geodesic vector field $G$ on $\bar A$ 
in \eqref{e_Geod} 
is the $\omega_g$-Hamiltonian one $X_L$, see \cite[Proposition 7.19]{Dahl}.

As the integral curves of $G$ are the velocities of the geodesics,  $\eta_g(G|_{(x,y)})=g_{ij}(x,y)y^i y^j$ is constant along each geodesic and vanishes only on the lightlike ones. In particular, $G$ is tangent to the indicatrix $\Sigma=L^{-1}(1)$  and $\eta_g(G)|_{\Sigma}\equiv 1$ for  both the Finsler and Lorentz-Finsler cases.

\subsubsection{Contact geometry}
 $\Sigma$
is a hypersurface of $T A$  with odd dimension $2n-1$, where the standard contact setting for the Finsler case \cite{Dahl} extends  to the Lorentz-Finsler one \cite{HedickePhD, Hedickepreprint}. 

A contact form $\eta$ is a 1-form  satisfying $\eta \wedge
d \eta \wedge \dots \wedge ^{(n-1)}d\eta \neq 0$ so that ker $\eta$ is a {\em contact structure}, that is, a distribution of hyperplanes which is maximally non-integrable, in the sense that the distribution can be locally expressed from some  $\eta$ as above. We will be interested in the case that this expression is global (i.e.,   the distribution is {\em co-orientable}). Anyway,  the contact structure is also determined   by $f\eta$, with $f\neq 0$ (a non-vanishing function). The primary example of a contact structure is  the {\em spherical cotangent bundle} $ST^*M$ of  $M$ (i.e.  the quotient of $T^*M\setminus \{0\}$ by the $\R_{>0}$ action of positive homotheties), where the tautological 1-form $\theta$ induces the required hyperplane distribution. The fibers of this bundle are {\em Legendrian submanifolds}, i.e., the contact form vanishes on them and they have 
the maximum  dimension $n-1$. 
 
A contact form provides a unique {\em Reeb vector} field, characterized by  
$$\eta(R)= 1, \qquad i_R 
d\eta=0.$$     For each function $f$ on $\Sigma$,  
it also yields the   {\em contact Hamiltonian 
vector field } $X_f$   satisfying:  
$$\eta(X_f)= f, \qquad i_{X_f} 
d\eta=  df(R) \eta    -df.$$ 
$X_f$ is a {\em contact vector field}, i.e., its  flow is composed by contactomorphisms (which preserve the contact structure but not necessarily the contact form). Notice that, when $f\not = 0$, then $X_f$ is the Reeb vector of $  \eta /f$.  
     
The Hilbert form $\eta_g$ vanishes on the the radial vector field $y^i\partial_{y^i}$, which is transversal to $\Sigma$. Its restriction $\eta$ to $\Sigma$ has rank $2n-2$ and becomes the natural contact form $\eta$ on $\Sigma$  (see \cite[Example 2.1.1]{HedickePhD} and references therein).
Its Reeb vector $R$ is just the restriction 
of the geodesic field $G$ to $\Sigma$. 
Moreover, one can check that $\eta=R^\flat$, where the  isomorphism 
$\flat:TA\to TA^*$ is computed   using the {\em Sasaki metric} (see \cite[Sect. 4.1]{Dahl}) restricted to $\Sigma$.

Notice that $\eta_g$ can also be restricted to $\C$, however, the setting changes as $\eta_g(G)|_{\C}\equiv 0$). We postpone its  study to the space of cone geodesics in \S \ref{s3_2}. 
    
\subsubsection{About the Hamilton-Jacobi equation}\label{ss_HJ} When studying the propagation of light in a classical setting, the relation between  particle trajectories and  wavefronts is well established in classical Mechanics from Hamilton-Jacobi equation (see \cite{AbrahamMarsden, Arnold} or the  thesis \cite{Vaquero}). 

When considering light propagation or shortest paths in Riemannian geometry and geometric optics, Hamilton-Jacobi reduces to the eikonal  equation,  namely, $|\nabla u|(x)=n(x)$, where $n(x)$ is a positive function, tipycally either $n(x)\equiv 1$ or the refraction index $n(x)=c/v(x)$ ---the quotient between the speed of light in vacuum and in (isotropic) media. 

The function $n(x)$ can be regarded as a cost function to be locally minimized by the trajectories, so that the wavefronts become the level sets of an appropriate distance function. 
This setting will underlie in the applications to both Lorentz Geometry and Classical Mechanics to be developed next.

\section{Global Lorentz-Finsler Geometry}\label{s3}

As explained in \S \ref{s2_3_1}, most of classical local Lorentz Geometry and global  causality can be transplanted to the Lorentz-Finsler setting. Next, we focus on globally hyperbolic Finsler spacetimes  to explain some  topics of interest concerning their global structure. 

Recall that these spacetimes are defined as causal  with compact diamonds $J(p,q):=J^+(p)\cap J^-(q)$, as in the Lorentz case \cite{BS07}.

\subsection{Finsler globally hyperbolic splittings}\label{s3_1} In this subsection, the main aim is to achieve the   Finslerian splittings in Theorems \ref{t_splitting}, \ref{t_splitting2} ensuring also additional properties listed below (Remark \ref{r_splittings}).  Reference \cite{Sa22}  expands our Lorentzian setting here and other Finsler developments.   

\subsubsection{The Lorentzian splitting}\label{s_items}
Classical Geroch's theorem \cite{Ge} asserts the equivalence,  in the Lorentz case, between the  global hyperbolicity of a spacetime $(M,g)$ and the existence of a topological acausal Cauchy hypersurface $S$. What is more, he  constructed a Cauchy time function $t$ with all the levels $t=$ constant Cauchy hypersurfaces (with the same properties as $S$) and, then, a homeomorphism between $M$ and $\R\times S$.  The long-standing question whether $S$ or even the Cauchy time  $t$ could be found smooth, was answered in the positive by Bernal and the author 
in \cite{BS03} and \cite{BS05}, respectively. The last reference, however, not only proved the existence  of a smooth Cauchy $t$ but also the existence of a Cauchy temporal function $\tau$  (i.e, smooth with timelike gradient 
$\nabla \tau$). This not only implies the improvement of the topological splitting $\R\times S$ into a smooth one with spacelike Cauchy hypersurfaces, but also the existence of an associated smooth  global orthogonal splitting, that is, with no cross terms between the $\R$ and $S$ parts. 
On  the one hand such splitting admits   first-principles implications about the possibility to recover globally the notions of time and space \cite{S25} and, on the other,  
Cauchy temporal functions $\tau$ also exhibit more technical advantages: 

\ben\item \label{item1}   $\tau$ is also Cauchy temporal for  
$C^0$ close metrics (an intuitive fact, which holds even when  timelike boundaries are permitted, see \cite{AFS}). In particular, the corresponding splitting  shows the stability of global hyperbolicity too (a fact claimed in  Geroch's original article \cite{Ge}). 

\item \label{item2}  Such $\tau$'s are flexible in terms of Cauchy initial data, namely: any  compact spacelike acausal hypersurface $S_0$ can be extended to  a spacelike Cauchy one $S$, which can  be regarded then as the slice $S=\tau^{-1}(0)$ for a Cauchy temporal function \cite{BS06}.  

\item \label{item3} The hypotheses on $\tau$ can be  strengthened   into  {\em steep} Cauchy temporal function (i.e., $\tau$ satisfies additionally $g(\nabla \tau,\nabla \tau)\leq -1$).

 This implies the possibility to obtain a Nash-type \cite{Nash, Gromov} isometric embedding  of the spacetime in Lorentz-Minkowski $\Lo^N$ for large $N$, see \cite{MS}. 

\item \label{item4} As proven by Burstcher and  Garc\'{\i}a-Heveling \cite{AnnegretLeonardo}, the temporal $\tau$ can be chosen  $h$-steep with respect to a complete Riemannian metric (a notion introduced by Bernard and Suhr \cite{BS} see below), which is  a condition more restrictive than being Cauchy temporal (additionally, it can be found also steep, see  
\cite[Appendix C]{BFS25}).

This is important in the context of Sormani and Vega's  null distance 
 $\hat d$ associated with $\tau$, see \cite{SormaniVega}. Indeed,  $h$-steepeness  implies the completeness of $\hat d$,   a crucial  property  for the theory of convergence of metrics. Moreover, $h$-steepness also  becomes equivalent to the completeness of the Wick rotated metric of the conformal  metric $g/|\nabla \tau|^2$, which also has implications for this theory \cite[Sect. 4]{BFS25}.

\item \label{item5} $\tau$ (and, then,  the splitting) can be chosen invariant by any compact group of isometries of the spacetime \cite{Mu16} (among other properties, see \cite{Mu12}).

\item \label{item6} In the case of {\em globally hyperbolic spacetimes-with-timelike-boundary}\footnote{\label{f_gh_con_borde}This means that the manifold $\bar M$  is permitted to admit a boundary $\partial M$ which turns out timelike at all the points, and the compactness of $J(p,q)$ in the definition of global hyperbolicity is imposed by including also the points in $\partial M$ (and taking not only piecewise smooth causal curves but also with locally Lipschitz regularity), see also \S \ref{sss_with_bd} below.}, $\tau$ can be chosen {\em adapted to the boundary}, that is,  with $\nabla \tau|_{\partial M}$ is  tangent to $\partial M$. Then,  $\bar M$ splits orthogonally as $\R\times \bar S$, where $\bar S$ is a spacelike Cauchy hypersurface with boundary, see \cite{AFS}.

\item \label{item7} In case of analytic metrics, the Cauchy temporal function, as well as the associated splitting, can be chosen also analytic,  see \cite[Th. 2.10]{Sa22}. Obviously, the extensibility  from a compact hypersurface $S_0$ to a Cauchy hypersurface seen in item  \eqref{item2} cannot hold in an analytic way. However, the smooth extension $S$ obtained in that item can be approximated by analytic Cauchy spacelike hypersurfaces.

\een

\subsubsection{The Lorentz-Finsler setting} Starting at 
 Fathi and Siconolfi's  article \cite{FS} published in 2012, a number of new approaches  have been developed extending the existence of  Cauchy time functions  to cone structures $\C$, even lowering the regularity considered here. Indeed, they consider cones just     convex, closed, salient  with non-empty
interior and varying continuously, and they  apply   weak KAM theory. Under global hyperbolicity, such a $\C$ is proved to admit a    smooth Cauchy time function. Noticeably, there is  a previous related approach in dynamical systems by Sullivan  \cite{Sullivan}, as well as a posterior one by Monclair   \cite{Monclair1, Monclair2}, the latter  linking     attractors, chain recurrent points and the existence
of time functions in Lorentzian manifolds by using Conley theory. This theory was used shortly after by Bernard and Suhr \cite{BS, BS2} to prove the existence of  temporal functions  satisfying  {\em $h$-steepness} (as mentioned above, this condition   becomes  stronger than Cauchy temporal  in the Lorentz case but independent of steepness,   \cite{BFS25}). Their results hold   for more general  closed cone structures; with this purpose,  a suitable extension of the notion of global hyperbolicity is introduced. 

Summing up,  at least for  regular cone structures, the existence of Cauchy temporal functions and some of the properties  listed in \S \ref{s_items} (as the stability in item \eqref{item1} or the  flexibility in item \eqref{item2}, see  \cite{Min17b}), is well established now by   methods alternative to the original Lorentzian ones ---however,  the possibility of low regular splittings does not seem to have been explored.  Next, we emphasize that the original approach  \cite{BS03, BS05, BS06, MS, AFS} extends  to Lorentz-Finsler manifolds yielding a global splitting and ensuring the properties listed in \S \ref{s_items}. 

\subsubsection{Case without boundary}
Recall that, given a cone structure $\C$, a smooth function $\tau$ is  {\em temporal} when $d\tau$ is timelike everywhere and then,  the (pointwise) map $\sharp$ defines  the  (necessarily timelike) gradient $\nabla \tau=d\tau^\sharp$ (Remarks~\ref{r_gradiente_temporal}, \ref{r_gradiente_final}).

\begin{thm}\label{t_splitting}
Any globally hyperbolic Finsler spacetime admits a Cauchy temporal function $\tau$ and, then, a splitting $\R\times S$ such that:

\bit\item Under the splitting,  $\tau: \R\times S \rightarrow \R$ is the natural projection on $\R$.
\item $\nabla \tau$ is everywhere collinear to  the natural vector field $\partial_\tau$.
\eit
Thus, the factors $\R$ and $S$ of the product are everywhere $g_{\partial_t}$-orthogonal;  moreover, all its  smooth Cauchy hypersurfaces are diffeomorphic.
\end{thm}

\begin{proof}
The technique used in \cite{BS05} to construct $\tau$ uses the following ingredients  which can be transplanted to the Lorentz-Finsler setting\footnote{Reference  \cite{Ringstrom} is also strongly recommended for a proof in book format as well as further implications on the initial value problem in General  Relativity.}.  
First, Geroch's Cauchy time function $t:M\rightarrow \R$, whose existence in  \cite{Ge} is proved by using standard properties of causality valid for cone structures. Second, a smoothing procedure which relies on the existence of functions labelled $h_p$, $p\in M$  (introduced in \cite{BS03}) defined as follows.

Let $t_1<t_2$ and consider Geroch's Cauchy slices $S_1:=t^{-1}(t_1), S_2:=t^{-1}(t_2)$.   
For each $p\in S_2$ and convex neighborhood $U_p$ of $p$, the function  $h_p : M \rightarrow [0,\infty)$
must satisfy: (i) $h_p(p) = 1$, (ii) the support of $h_p$ (i.e., the closure of $h_p
^{-1}((0,\infty))$) is compact and included in
$U_p \cap I^+(S_1)$, and 
(iii) for each  $q \in J^-(S_2)$ such that $h_p(q)>0$, $\nabla h_p|_q$ is  timelike and past-pointing. Notice that all these conditions can be transplanted to the Finsler setting, just taking into account that the hypotesis (iii) is stated assuming that  $d h_p|_q$ is timelike and, thus $h_p$ is temporal around $q$ with gradient $\nabla h_p|_q$.

To construct $h_p$, one chooses  $p' \in I^-(p)\cap I^+(S_1)$ such that $J^+(p') \cap J^-(S_{2}) \subset U_p$ and define $h_p$ on 
$I^-(S_{2})$ as the function:
\be \label{hp} 
h_p(q) = 
e^{\mathbf{d} (p',p)^{-2}} \; \cdot \, e^{-\mathbf{d} (p',q)^{-2}} 
\ee
and construct any smooth extension of $h_p$ out of $I^-(S_{2})$ such that the support of $h_p$ is included in $U_p$ and $h_p\geq 0$. 
Here $\mathbf{d}$ is the time-separation on $U_p$ regarded as a spacetime
($h_p$ is regarded as 0 on $I^-(S_{2})\backslash  U_p$).   

The key for the regularity of $h_p$ in the Lorentzian case, is that the function $q \mapsto \mathbf{d} (p',q)^{2}$ satisfies:  on $h_p^{-1}(0,\infty)$, it is smooth with  bounded derivatives of any order  (indeed, it is a quadratic polynomial of degree 2) 
on  $h_p^{-1}(0,\infty)$) and, thus,   the negative exponentiation in \eqref{hp} smoothens it at 0. This  is extensible to the Lorenz-Finsler case too, because now     
$q \mapsto \mathbf{d} (p',q)^{2}$ becomes equal to $L_{p'}$ in normal coordinates 
on $h_p^{-1}(0,\infty)$. Thus, $L_{p'}$ is smooth therein and, moreover,  the derivatives of any order are bounded   because $L_{p'}$ can also be extended smoothly beyond the boundary.

The other steps in the references \cite{BS03, BS05} use algebraic manipulations of these functions that do not affect smoothability and, then, they are extensible to the Lorentz-Finsler case. It is worth pointing out that only finite sums of these types of functions are used  in  \cite{BS03}, while in    \cite{BS05} infinite sums are used. Anyway, no additional issue appears because the simple argument which proves that a converging  sum of timelike vectors is also timelike (see \cite[Lemma 3.10]{BS05}) is valid  for Lorentz-Minkowski norms too.       

Once $\tau$ has been constructed, the splitting $\R\times S$ is obtained by using that $\nabla \tau$ is a timelike vector field (this holds in the Lorentz-Finsler case by the consistency of Def. \ref{d_LegendreLorentz}) and, then, its integral curves must cross all the slices of $\tau$. The last assertion follows because the splitting procedure flows $S=\tau^{-1}(0)$ by using the vector field  $\nabla \tau/|\nabla \tau|^2$, then guaranteeing that $\partial_t$ and $\nabla \tau$ points out in the same oriented direction; this also provides a diffeomorphism between any two smooth Cauchy hypersurfaces.
\end{proof}

\begin{rem}
The orthogonality in the last assertion of Theorem \ref{t_splitting} is the best possible one, as $g_v$ varies with the causal vector $v$ and it is not defined in the directions tangent to the factor $S$.
\end{rem}
\subsubsection{Case with timelike boundary for $\C$ and summary}\label{sss_with_bd} Let us consider the case of a manifold with boundary $\bar M$, as in \cite{AFS,HS25}.
For a cone structure $\C= \partial \bar A$ on $\bar M$, the boundary is called {\em timelike} when so are the hyperplanes tangent to $\partial M$ according to Def. \ref{d_cone} (that is,
when $T_p{\partial M}\cap A_p\neq 0$ for all $p\in \partial M$). Accordingly, if $L$ is a Lorentz-Finsler metric on $\bar M$, the boundary is defined as {\em timelike} when so is  its cone structure. 

Now, all the notions like the definitions of inextensibility for  causal curves $\gamma$, the sets $I^+(p), J^+(p)$ and the compactness of $J(p,q)$, must take into account the points of $\partial M$ and: (a) this forces to consider  locally Lipschitz regularity instead of piecewise smoothness for the causal curves $\gamma$ used to compute $J^\pm(p)$ (by the reasons argued in \cite[Appendix B]{AFS}, appliable  even if $\partial M$ is $C^\infty$), and (b) necessarily, Cauchy hypersurfaces $\bar S$ must have a boundary, eventually met by  inextensible $\gamma$ (which may travel along $\partial M$).

For the extension of the splitting in Theorem \ref{t_splitting} to the case with timelike boundary,  
the  procedure in \cite{AFS} for the Lorentzian case becomes  quite technical now. However, the optimal result for cone structures can be solved as a simple consequence of the introduced tools.

\begin{thm}\label{t_splitting2}
Let $(\bar M,\C)$ be any cone structure which is  globally 
hyperbolic-with-timelike-boundary: 

(a) It admits a Cauchy temporal function $\tau$. 

(b) For any  Cauchy temporal function $\tau$ on $\bar M$, there exists a  
Lorentz-Finsler metric $L$ such that $\tau$ is  {\em adapted to the boundary} for $L$,  that is,  $\nabla \tau|_{\partial M}$ is  tangent to $\partial M$. 

(c) Any Lorentz-Finsler metric $L$ endowed with a Cauchy temporal function $\tau$  adapted to the boundary   admits a global Cauchy splitting $\bar M=\R\times \bar S$ satisfying all the properties stated in Theorem \ref{t_splitting}.
\end{thm}

\begin{proof} (a) Notice first that, in the Lorentzian case, the technique to obtain a (non-adapted) Cauchy temporal function for the case without boundary also works for the case with boundary \cite[Remark 4.3(2)]{AFS}. Thus, this holds in the Lorentz-Finsler case by the reasons in Theorem \ref{t_splitting}.

(b) Choose any future directed timelike vector field $T$ on 
$\bar M$ such that $T|_{\partial M}$ is tangent 
to $\partial M$,\footnote{This can be carried out just by noticing that both $M$ and $\partial M$ are manifolds without boundary, 
	thus admitting timelike vector 
fields $T_M$, $T_{\partial M}$, resp. for the inherited cone structures.  $T_{\partial M}$ is tangent to $\partial M$ and it is also extendible to 
a timelike vector $T_W$ on neighborhood $W$ of  $\partial M$. Then, $T=: \mu T_M + (1-\mu)T_W$, where $\{\mu, 1-\mu\}$ is a partition of the unity subordinated to $\{M,W\}$, suffices.} and assume $d\tau(T)=1$ with no loss of generality. Using Theorem \ref{t_cone_triple}, consider the associated cone triple $(\Omega=d\tau, T,F)$ and function $\mathcal{G}=d\tau^2-(F\circ \pi_T)^2$. 

By construction, at each point $p\in \bar M$, $T$ attains the minimum of $d\tau$ restricted to $\Sigma_{\mathcal{G}}|_p:=\mathcal{G}^{-1}(1) \cap {T_pM}$. 
Then, one can smoothen  $\Sigma_\mathcal{G}$ (thus $\mathcal{G}$) close to 
the direction of $T$ to obtain a smooth and strongly convex $\Sigma$  following the procedure in \cite[\S 5.2]{JS20}. It is important to notice that this procedure 
permits to maintain $T$ as the direction of the minimum of $d\tau|_{\Sigma}$. Indeed, revising the procedure, 
it is enough to ensure that 0  remains as the minimum for the smoothed function 
$\tilde t_0$ 
in \cite[Lemma 5.4]{JS20}. This can be achieved directly  by using the general argument by D. Azagra explained in \cite[Remark 5.5]{JS20}. Thus, taking the corresponding Lorentz-Finsler metric 
$L$ (with the so-smoothened indicatrix $\Sigma$), the gradient $\nabla \tau$ will point 
everywhere in the direction of $T$, in particular along $T(\partial M)$ on $\partial M$. 

(c) The procedure to obtain the splitting in Theorem \ref{t_splitting} also works now,  because the flow of $\nabla \tau$ preserves $\partial M$.
\end{proof}

\begin{rem}\label{r_splittings}
The strengthening of the conditions on temporal functions for the Lorentzian case in items \eqref{item1} to \eqref{item5}  of \S  \ref{s_items}   rely on arguments similar to those in the proof of Theorem \ref{t_splitting}. It  is worth pointing out that cone triples $(\Omega, T, F)$  provide a simple technical way to implement operations in cone structures as narrowing or widening  them (say, replacing $F$ by a Finsler metric with indicatrix inside  or enclosing the one of $F$, resp.) to be used for stability.    Thus, these five items can  be extended to the Lorentz-Finsler case as before. Moreover, the strengthening in item \eqref{item6} corresponds directly to the Lorentz-Finsler result in Theorem \ref{t_splitting2}.

Notice that some of these stregthenings  are relevant in its own, but items \eqref{item3} and \eqref{item4} are also important for  applications to Nash-type isometric embeddings and convergence of Lorentz metrics. Thus, it would be interesting to explore Lorentz-Finsler issues in these fields, taking into account the subtleties of the purely Finsler problem (see for example \cite{BI}). 

Finally, the analytic case in item \eqref{item7} can also be extended to the Lorentz-Finsler case reasoning as in \cite{Sa22}. Namely,    Grauert  
\cite[Prop. 8]{Grauert} and Whitney \cite[Lemmas 6, 7]{W} (see also \cite{Azagra}) established that   any $C^k$ 
  function, $k\in \N$, on a  real analytic manifold can be $C^k$ approximated by analytic functions. 
  Then, a  $C^2$ approximation of a (eventually steep and $h$-steep) Cauchy temporal function by analytic ones   suffices to maintain the spacelike character of the slices and its Cauchy character (apart from the other eventual properties), as well as to construct the  splitting in Th. \ref{t_splitting}.   
\end{rem}

\subsection{The space of cone geodesics}\label{s3_2} The study of the \emph{space $\Ng$ of null geodesics}   or \emph{lightrays} or {\em cone geodesics} (the latter used here to emphasize its exclusive dependence on the cone structure) comes from
 Penrose's seminal ideas   on   twistor space, leading to consider lightrays as a fundamental physical structure from where spacetime emerges \cite{Penrose0, Penrose}.
Inspired by this approach, Low studied $\Ng$ introducing the  viewpoint of contact geometry in the Lorentz setting \cite{Low1989, LowSurvey}.  Hedicke   \cite{HedickePhD, Hedickepreprint} has   developed  the Lorentz-Finsler case,   thus   providing background for this section.  


  Roughly,   Low proved that $\Ng$ turns out well behaved   and it is identifiable to   a  natural smooth (Hausdorff) manifold in all globally hyperbolic spacetimes \cite{Low1989}. 
  Hedicke and Suhr \cite{HedickeSuhr}   extended this result to some   classes of causally simple spacetimes. 
 The structure of  $\Ng$ can also be fully determined for   globally hyperbolic cone structures-with-timelike-boundary, in particular, relativistic spacetimes admitting  a timelike conformal boundary such as (conformally asymptotic) anti-de Sitter \cite{HS25}. 
 
   Remarkably,   Low \cite{LowSurvey} also developed  the symplectic and contact viewpoints connecting causality with topological linking,   posing a   first version   for a crucial conjecture solved in positive by Chernov and Nemirowski \cite{ChernovNemirovski10}. 

  Next, we deepen in some details and prospects on this  topic.   From a physical viewpoint, it is also worth mentioning the study by Hasse and Perlick  \cite{HassePerlick} about physical implications of $\Ng$ in order to detect possible Finslerian corrections to standard Relativity. 

\subsubsection{Setting and Low's conjecture}\label{s_Low}
Let   $(M^{n+1},\C)$ be a  cone structure.  Extending the Lorentz case,  consider the \emph{space  of  cone geodesics $\Ng$} 
\[
\mathcal{N}= \{ \gamma \ \text{ inextensible cone geod. in } M \},
\]
 computable by taking all the inextensible lightlike geodesics  for any compatible Lorentz-Finsler $L$
 and identifying those differing in the parametrization.

 Both the geodesic spray $G$ and the radial field $\mathcal{R}:=y^k \frac{\partial}{\partial {y^k}}$ on $\bar A \subset TM$ preserve the tangent to the cone structure $T\C$ and, moreover,  the distribution $\mathcal{F}:=$ Span $\{\mathcal{R},G\}$ is involutive   ($[\mathcal{R},G]=G$).   This makes consistent the definition of $\mathcal{N}$ as the quotient topological space $T\C/\mathcal{F}$.
When $\C$ is strongly causal, the quotient is a possibly non-Hausdorff manifold  of the same dimension, see   Low's   \cite[Prop. 2.1]{Low1989} (or   Lorentz-Finsler   \cite{HedickePhD, HS25}).

\begin{rem}
 Non-Hausdorffness arises easily. Indeed, this happens in the  space of 
(Riemannian) geodesics for $\R^2\setminus\{0\}$ and this is transmited to the space of cone geodesics in   $\R^1_1\times 
(\R^2\setminus\{0\})$
 (an open subset of Lorentz-Minkowski 
$\LL^3$).
Anyway the  structures to be  considered next  make sense whenever $\C$ is a (possibly non-Hausdorff) manifold.

Below strong causality, 
classes of compact cone structures with well behaved $\Ng$ are  Zollfrei spacetimes \cite{Suhr}, that is,
compact Lorentzian manifolds such that the geodesic flow restricted to the 
cones induces a fibration by circles (thus, the orbit space being a smooth manifold),   and some classes of lens spaces nicely related to  Engel geometries  \cite{MarinSalvador2024}. 
\end{rem}

When the Finsler spacetime is globally hyperbolic, $\Ng$ becomes a smooth manifold. Indeed, choose any 
Cauchy hypersurface $S$, and notice that cone geodesics will cross $S$ exactly once, then  providing an identification between $\Ng$ and the set of lightlike directions on $S$. This set can be characterized by taking any cone triple $(\Omega,T,F)$ for $\C$ and noticing that lightlike directions are in bijective correspondence with the indicatrix for $F$ on $\Sigma$, that is, rewritting the conclusions in \cite{Low1989, HedickePhD}:

\begin{prop} \label{p_Ng_glob_hyp} The space of cone geodesics $\Ng$ for any globally hyperbolic cone structure $(M^{n+1},\C)$ is diffeomorphic to $\Sigma_{S}, \, 
 \hbox{where}$  
\bit\item $S\subset M$ is any smooth spacelike Cauchy hypersurface, and 
\item    $\Sigma_{S}$ is the indicatrix  determined by the Finsler metric $F$ of any cone triple  $(\Omega,T,F)$ for $\C|_{S}$.  
\eit
In particular, dim $\Ng=2n-1$.
\end{prop}
The independence of  $\Ng$ on the chosen $S$ is consistent with Th. \ref{t_splitting} and  $\Sigma_{S}$ becomes also a convex bundle $\Sigma_{S}\rightarrow S$

   The  Hilbert   contact structure 
 (see \S \ref{ssHilbert_form}) is preserved on $\C$  by the flows of $\mathcal{R},G$ and it can be induced on $\Ng$,   whenever $\Ng$ is a manifold.   Anyway,  to provide a    better insight,  
 $\Ng$ will be regarded  as a quotient in the cotangent bundle $T^*M$  by using the 
Legendre map (Def. \ref{d_LegendreLorentz}). 
Following Low \cite{LowSurvey} in the Lorentz case (or 
\cite[Sect. 3.1]{HedickePhD} in the Lorentz-Finsler one),      first   consider 
$\C^*$ as a subset of $ T^*M\setminus\mathbf{0}$ (the $H_L$-energy level $E=0$), and  then regard $\Ng$ as the quotient $\C^*/\mathcal{F^*}$ with $\mathcal{F^*}:=$ Span $\{\mathcal{R^\flat},G^\flat\}$   (which works specifically for $E=0$). 
Then, the standard {\em contact form} on $T^*M$ induces a contact form $\eta$ on $\Ng$ \cite[Prop. 3.1.4]{HedickePhD}. This  provides a tool which will link
the differential topology of $\mathcal{N}$ with the causal 
geometry of $M$. 

Indeed, each point $p\in M$   induces   a submanifold in $\Ng$, namely, its \emph{sky}:
\[
\mathfrak{S}_p = \{ \gamma \in \mathcal{N} \mid p \in \gamma \},
\]
Such a sky becomes a {\em Legendrian sphere} for the  contact form.


Low conjectured  that, in a globally hyperbolic spacetime, causal relations between events in $M^{n+1}$,  can be detected by topological linking of their skies in $\mathcal{N}$    for $n=2$.
  Subsequent work by Nat\'ario and Tod \cite{Natario}   proved partially it   and  clarified that the proper invariant for $n>2$ is \emph{Legendrian} linking. Chernov and Nemirovski proved the 
reformulated conjecture   first in the case of a Cauchy hypersurface $S$ diffeomorphic to an open subset of $\R^n$ \cite{Chernov_2010} and then 
requiring only non-compactness for $S$, namely: 

\begin{thm}[Th. 10.4 in \cite{ChernovNemirovski10}] \label{t_CN}   Consider a spacetime $M^{n+1}$ whose universal cover has a non-compact Cauchy hypersurface $S$ and $n> 2$ (resp. $n=2$). 
The points $p,q\in M^{n+1}$ are causally related (i.e., $p<q$ or $q<p$)   
if and only if $\mathfrak{S}_p$ and $ \mathfrak{S}_q$ are non-trivially Legendrian linked (resp. topologically linked) in the contact manifold $\mathcal{N}$.
\end{thm}
  Notice   that the  non-compactness of a Cauchy hypersurface  \cite{Low01} or just of its universal covering \cite{ChernovRudyak, Bauermeister}, implies that 
different points have different skies , that is,  the spacetime is  {\em non strongly
(null) refocusing}. This property permits  to reconstruct $M$ from the space of skies as a point set; to recover it topologically, one should prevent the weaker condition of  {\em refocusing}, that is, forbid that  the cone geodesics departing from one point $p$ can arrive at points arbitrarily close to a second one $q$. 

 For further results and   open questions in the case of 
compact $S$, see  \cite{Chernov18}.   
 Other relevant related topics (eventually extensible to the Lorentz-Finsler case, see  \cite{Hedickepreprint}) include 
the relation of non-refocusing  with  isotopies by skies   \cite{ChernovRudyak, Kinlaw,  Bautista}, causal boundary constructions \cite{LowSurvey, Bautista_2022, Chernov_2020} or inverse and  reconstruction problems \cite{Feizmohammadi2019, Bautista_2014}. 
 
The wealth of this framework also opens new issues connecting contact and Lorentz (-Finsler) geometries, with interest  for the contact side too,   where  isotopies that run transverse to   $\eta$ appear naturallly \cite{Polterovich}.   In particular, Chernov and  Nemirovski \cite{Chernov_2020} found an analogy between the  {\em orderability}  provided by the group of positive contactomorphisms (preserving a co-oriented contact structure) Cont$(M, \eta)$  and the Lorentz {\em causal condition} for spacetimes. They introduced the {\em interval topology} in  Cont$(M, \eta)$  which resembles  Alexandrov's topology in the Lorentz setting, and determined properties such as Hausdorffness.
  
Moreover, Hedicke \cite{Hedicke2024}, inspired in ideas from symplectic geometry as
Shelukhin's  Hofer norm on contactomorphisms \cite{Shelukhin},  defined a Lorentzian type distance on the connected part of the identity $\hbox{Cont}_0(M, \eta)$. Among its nice properties, this distance is continuous, makes $\hbox{Cont}_0(M, \eta)$ a Lorentzian pre-length space and it can be calculated for
the Reeb flow of the contact form
---then permitting to answer a question by
Shelukhin  on the diameter of the contactomorphism group under the assumption of
orderability \cite[Question 18]{Shelukhin}.  When applied to the space of skies of 
globally hyperbolic spacetimes, the manifold topology coincides with the topology induced by
the distance of Legendrian skies in the space of cone geodesics \cite[Th. 6.1]{Hedicke2024}.

\subsubsection{Causal simplicity and the case with timelike boundary} \label{s_HS25}
The first natural question on $\Ng$ is its Hausdorffness, only ensured above by global hyperbolicity. The step in the standard causal  ladder of spacetimes  just below global hyperbolicity is causal simplicity, that is,  
being causal with all $J^+(p), J^-(p)$  closed, \cite{MinSan}. Low characterized the compactness of the intersections $J^+(p)\cap J^-(p)$ for global hyperbolicity in terms of {\em null pseudoconvexity}, \cite{Low1990}\footnote{  \label{f_Bartolo} A similar characterization of causal simplicity in  terms of {\em maximal null pseudoconvexity} was claimed in \cite[Th. 2]{VPE}, but a counterexample was found in \cite{Hedickeetal}, building on  a previous Riemannian example on convexity  in  \cite{BartoloAGAG}. }.
  Chernov \cite{Chernov} conjectured  that causal simplicity might  suffice for Hausdorffness. A partial positive answer was given by Hedicke and Suhr:

\begin{thm}
\cite[Thm 2.5]{HedickeSuhr} If a causally simple spacetime can be isometrically embedded as an open subset of a  globally hyperbolic spacetime then   $\Ng$ inherits a  Hausdorff contact manifold structure. 
\end{thm}
However, these authors also gave counterexamples to Chernov's conjecture, see \cite[\S 2]{HedickeSuhr} (  see others in    in \cite[\S 5.2]{HS25}). Indeed, they found examples of causally simple spacetimes  non openly embedded in a globally hyperbolic spacetime with both Hausdorff and non-Hausdorff $\Ng$.

Anyway, causal simplicity turns out a natural class where the Hausdorffness of $\Ng$ should be studied. This class does not have a natural candidate for the manifold structure, in contrast to the globally hyperbolic case (Prop.~\ref{p_Ng_glob_hyp}).
A natural  subclass 
to focus is bounded convex open subsets of
Lorentz-Minkowski space $\LL^{n+1}$.  Indeed, Hedicke \cite{Hedicke2021}  considered a more general  class of  open subsets for globally hyperbolic spacetimes, namely, the {\em conformally star-shaped} ones. For this class, $\Ng$ becomes contactomorphic to the
space of cone geodesics of the surrounding globally hyperbolic spacetime, being  then
its contact type that of a spherical cotangent bundle  (in particular, $ST^*\R^n$ for the case of  $\LL^{n+1}$), see \cite[Th. 2.3]{Hedicke2021}.
 
Another interesting  class of open neighborhoods, the  {\em strongly null convex} ones, has also been   considered in \cite{Hedicke2021}. The results are restricted to dimension 3, where they connect to other parts of contact geometry. For such a neighborhood $K$, the boundary $\partial \Ng_K$ of its space of cone geodesics $\Ng_K$ becomes a  convex surface included in the cone geodesic space $\Ng$ of the globally hyperbolic spacetime \cite[Th. 2.8]{Hedicke2021}. Here, {\em convex surface} is used in the contact sense  introduced by Giroux \cite{Giroux}, namely, a closed surface $S$ with a transverse contact vector field defined on a neighborhood of $S$ (this author  also proved that these  surfaces are generic among  closed surfaces), thus, the two notions of convexity  are linked. 
Under further hypotheses, $\Ng_K$ becomes contactomorphic to the spherical cotangent bundle of $\R^2$, \cite[Th. 2.9]{Hedicke2021}.


In the case of spacetimes with timelike 
boundary  seen in \S \ref{s3_1}, the recent reference \cite{HS25} considers both, Lorentz-Finsler metrics and spacetimes with timelike boundary (although the outcomes are equal in the Lorentz and the Lorentz-Finsler cases, the proof in the latter becomes subtler).  
Indeed, both causal simplicity and Hausdorffness of $\Ng$, are ensured by the (infinitesimal) {\em lightlike convexity} of the boundary, that is, the positive semidefiniteness of the $\partial M$ second fundamental form  on lightlike 
directions. More precisely,

\begin{thm} [\cite{HS25}, Theorem 1.2] For a globally hyperbolic-with-timelike-boundary cone structure $(\bar M,\C)$ they are equivalent:
	\begin{enumerate}
		\item The boundary $\partial M$ is lightconvex. 
		\item The interior $M$ of $\bar M$  is causally simple. 
		\item The space of cone  geodesics $\Ng$ of $M$ is  Hausdorff. 
	\end{enumerate}
	\end{thm} 

  The proof of the equivalence between the two first items relies on a technical result on the equivalence of infinitesimal and local convexity for $\partial M$ in the semi-Finsler case \cite[Theor. 1.1]{HS25}, 
 which goes back to convexity equivalences in Riemann and Finsler geometries  studied  by Bartolo et al. \cite{BartoloCVPDE} and the reconstruction problem by Hintz and Uhlmann \cite{HintzUhlmann}.  A previous study by Caponio et al. \cite{CaponioGerminarioSanchez} focused on  the case of stationary spacetimes (where the Lorentz-Finsler metric becomes essentially a Finsler one of Randers type, see \S \ref{ss_stationary} below)    providing further applications to visibility and gravitational lensing. 

It is worth pointing out that, then, 
 $\Ng$  is a smooth Hausdorff $(2n-1)$-manifold (without boundary), and it  can be explicitly described  in terms of: (a)  the projectivization of the tangent bundle $T \bar S$ of any spacelike  Cauchy hypersurface  $\bar S$ in $\bar M^{n+1}$ (also used in the case without boundary), (b) the set of inwards lightdirections starting at points in $J^+(\bar S)\cap \partial M$ and (c) the set of  outwards lightdirections  starting at points in $J^-(\bar S)\cap \partial M$.  In particular \cite[Corollary 4.14]{HS25}: 
 \begin{quote}
{\em  if every cone geodesic has a future  (resp. past) endpoint at $\partial M$, then $\Ng$ is diffeomorphic to $\R\times T \, \partial S$.}
\end{quote}
This result opens prospects to extend some of the aforementioned results to a considerably bigger setting.

\subsection{Singularity theorems}\label{s3_3}
The extension of classical relativistic  singularity theorems to the Lorentz-Finsler setting becomes interesting not only as a technical achievement but also to stress the robustness of these results, in the same vein     
 as  recent results on low regularity for these  theorems \cite{Graf, Oy}.

Aazami and Javaloyes \cite{AJ16} proved the first Finslerian singularity theorem by transplanting Penrose's one.  They solved  technical issues (as those related with the notion of trapped surfaces or the lack of smoothability of the exponential map at 0), 
and 
clarified the appropiated tools. Indeed, they noticed  that the involved elements  (lightlike 
geodesics, focal points, trapped surfaces, Ricci scalar) depended exclusively on the 
nonlinear connection and cone structure, then using effectively the anisotropic Chern 
connection. Other classical  singularity theorems where extended to Finsler spacetimes 
focusing on Raychaudhuri equation in 
\cite{Min_singular}. Moreover,   weighted Raychaudhuri equations and inequalities were obtained  in \cite{LuMinOhta} by using weighted 
Ricci curvature, then obtaining various 
singularity theorems in this setting. Such a weighted tensor extends the weighted Ricci 
curvature introduced for   Finsler metrics by Ohta 
 \cite{Ohta}.

\section{Lorentz-Finsler applied to classical geometries}  \label{s5}

Next,  the Lorentz-Finsler viewpoint applied to both,  Lorentz Geometry, and   Finsler one (including the Riemannian case) is developed. In particular, the relativistic {\em causal boundary}  will emerge as a unifying concept for some relevant boundaries in the Riemannian, Finslerian and Lorentzian settings. Most of these applications  have already been explained in the review \cite{JPS22}. So, we will make here just a short summary with some updates.

\subsection{Application to causality in  Lorentz Geometry}\label{s5_2} 
 Consistently with the Hamilton-Jacobi viewpoint explained in \S \ref{ss_HJ}, next the causal futures and pasts  in a stationary spacetime    will   be computed by using a convenient cone triple $(\Omega=dt, T,F)$, where $F$ is a Finsler metric of Randers metric which is used as a cost function. Then, the causal futures and past of a subset are described in terms of a Randers distance function, whose possible non-symmetry reflects the asymmetry of the time $t$ (the map $t\mapsto -t$ is not an isometry). 
This leads to a precise characterization of the causality of these spacetimes, including the position in the causal ladder. Moreover, the applications will be extended to the class of $SSTK$ spacetimes by dropping the hypothesis that $T$ in the cone triple is timelike.   
Here, our basic references are \cite{CJS11, CJS24}, apart from  the aforementioned \cite{JPS22}.

\subsubsection{Stationary spacetimes}\label{ss_stationary} 
In Lorentzian Geometry, a {\em stationary spacetime} is a Lorentzian manifold admitting a timelike vector field  $K$ (which will be asume to time-orientate it). Locally, such a spacetime looks like a {\em standard stationary spacetime}, that is,   a Lorentzian manifold $(M=\R\times S,g)$ such that $g$ is written as
\begin{equation}\label{statmet}
	g((\tau,v),(\tau,v))=-\Lambda \tau^2+2\omega(v)\tau+g_0(v,v),
\end{equation}
where $(\tau,v)\in  T_{(t_0,x_0)}(\R\times S)$,  $(t_0,x_0)\in \R\times S$,  $\Lambda:S\rightarrow \R$ is smooth and positive, $\omega$ is a one-form  and $g_0$ a Riemmanian metric, all on $S$.  Here, the timelike Killing vector field is the natural one $K=\partial_t$ associated to the projection $t: \R\times S\rightarrow \R$ which becomes a temporal function.
These notions are naturally extensible to the Lorentz-Finsler case\footnote{\label{foot_stationary} A stationary Finsler spacetime admits a timelike vector $K$ 
which is Killing (its local flow preserves $L$); in this case ii is static case when the distribution $K^\perp=\{ v\in TM:  g_K(K,v)=0  \}$ is involutive. However, their local expresions are not ony a transplantion of the Lorentz ones, see \cite{JS20} and also \cite{CaponioStancarone, CJ_Gelomer}.}, but here  we will remain in the Lorentz setting. Such a standard stationary expression admits a natural cone triple $(\Omega, T,F)$ with $\Omega=dt$, $T= \partial_t$ and $F$ then determined by the corresponding cone structure $\C$. Notice that $F$ is properly a Finsler metric on $S$,  the {\em Fermat metric}. Specifically,  
\begin{equation}\label{Fermat}
	F(v)=\frac{1}{\Lambda}\omega(v)+\sqrt{\frac{1}{\Lambda^2}\omega(v)^2+\frac{1}{\Lambda} g_0(v,v)},
\end{equation}
which belongs to the Finslerian class of {\em Randers metrics}. Notice that this metric is {\em non-reversible} (that is, there exist $v\in TS$ such that $F(v)\neq F(-v)$) whenever $\omega\neq 0$, then  reflecting an asymmetry between the future and past-pointing directions. 
 As explained in \S \ref{s2_3_1}, $F$ determines  a generalized distance  $d_F$, which is again non-symmetric if $\omega\neq 0$. Then,  one can define the {\em forward and backward balls}, $B^+(x,r)$, $B^-(x,r)$ associated to $d_F$, whose definition depends on whether they are constructed by using the distance from the center $x$  to the point or the other way round, respectively. In terms of these balls, chronological future and pasts can be computed easily and  the causality of the spacetime can be characterized  in a tidy way, see  \cite[Theorems 4.3 and 4.4]{CJS11} or \cite[\S 3.4.2]{JPS22}. Summing up:

\begin{thm}\label{t_escala_causal} Any standard stationary spacetime $(M=\R\times S,g)$ with associated Fermat metric $F$ (as in \eqref{statmet}, \eqref{Fermat}) 
is causally continuous\footnote{A causality condition stronger then {\em stable causality} (ensured by the existence of the temporal function $t$), where the sets $I^\pm(p)$ vary continuously with $p\in M$, \cite{MinSan}.} and:
	
	\smallskip

	\noi (A) The following properties 
	are equivalent:
	
	(A1) $(M, g)$ is causally simple.
	
	(A2) $F$ is convex (i.e., each $p,q\in M$ can be joined by a minimizing $F$-geodesic).
	
	\smallskip
	
	\noi (B) The following properties 
	are equivalent:
	
	(B1) $(M, g)$ is globally hyperbolic.
	
	(B2) the ball intersections 
	$ B_F^+(p,r)\cap  B_F^-(q,r'), p, q\in S, r, r'>0$ are  precompact. 

	\smallskip

	\noi (C) The following properties 
	are equivalent:

	(C1) The slices $\{t_0\}\times S$ are Cauchy hypersurfaces (so, $t$ is a Cauchy temporal function).
	
	(C2) The Fermat metric $F$ is (forward and backward) complete.
	
\end{thm}

\subsubsection{SSTK spacetimes}
Notice that, as  Causality is invariant under conformal transformations, we could divide $g$ by $\Lambda>0$  in \eqref{statmet}. However, the above generality is useful to consider also the case when  the sign of $\Lambda$ changes, with the unique restriction that the metric remains Lorentzian, namely,
\begin{equation}\label{e_SSTKrestiction}
	\Lambda + \parallel\omega\parallel^2_0>0 
	\end{equation}
(here, $\parallel\cdot\parallel_0$ 	denotes the induced $g_0$-norm). In this case, the spacetime is called {\em standard with a spacetransverse
	Killing vector field} or {\em SSTK}. For these spacetimes, $t$ is still a temporal function, and  the intersection  $\{dt=1\}\cap \C$ provides a wind Finsler  structure $\Sigma$ on $S$  (Fig. \ref{f_ConoCasos3}). Let us analyze it explicitly (see \cite[\S 3.3]{CJS24} for detailed  computations).

	\begin{figure}
		
		\begin{center}
						\includegraphics[height=0.2\textheight]{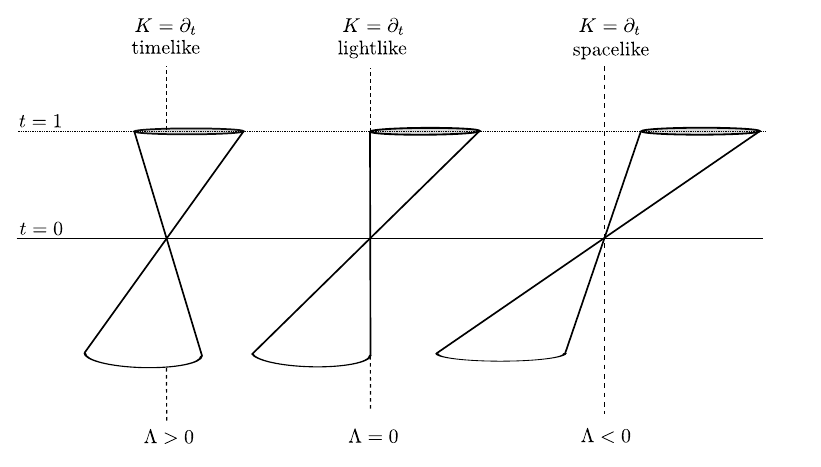}
			\caption{Wind structure on $S$ for a SSTK spacetime: Finsler (weak wind, $\partial_t$ timelike,   Fig. \ref{f_cone_triple}), Kropina (critical wind $\partial_t$ lightlike), strong wind ($\partial_t$ spacelike). [Credit \cite{CJS24}.]
			}\label{f_ConoCasos3}
		\end{center}
		
	\end{figure}
	
	In  case $\Lambda >0$ the Fermat metric $F$ in \eqref{Fermat} can be described in terms of its Zermelo navigation data $(g_R,W)$. These data directly provide the indicatrix of $F$ as follows: first, consider the indicatrix $\Sigma_R$ of the Riemannian metric $g_R$ and then, translate pointwise its center using the vector field 
	$W$.\footnote{The term {\em navigation} refers to the interpretation of  this translated  indicatrix as the pointwise and direction-dependent velocity of a sailboat  affected by the wind $W$.}  Zermelo data for $F$ in \eqref{statmet} are 	 determined by:
	\begin{equation}\label{e_FZ}
		\omega=-g_0(W,\cdot), \qquad g_R =\frac{1}{\Lambda+\|\omega\|_0^2}g_0.
	\end{equation} 
	However, these  data make sense whenever the SSTK condition \eqref{e_SSTKrestiction} holds. Indeed, the case $\Lambda=0$ corresponds to the critical case $0\in \Sigma$ (or Zermelo data with $g_R(W,W)=1$), so that $F$ turns into a {\em Kropina metric}, which is a type of singular Finsler metric (explicit in \ref{e_kropina} below). The case $\Lambda<0$ (Zermelo data with $g_R(W,W)>1$) corresponds to a strong wind structure. Here, the name {\em wind Riemannian } structure can be also used, as the desplaced indicatrix comes from a Riemannian metric.  
	
	The Finsler metric $F$ and the  improper Lorentz-Finsler metric\footnote{The name {\em  improper} Lorentz-Finsler reminds that, here, the condition of smooth extension to the cone for  Lorentz-Minkowski norms (Def. \ref{d_LorentzNorm}) is not properly fulfilled.} $F_l$ associated with  the wind structure $\Sigma$  (recall Fig. \ref{f_wind}) is explicitly:
		$$
		F=\frac{g_0}{-\omega+ \sqrt{\Lambda g_0+\omega^2}}, \qquad  F_l=-\frac{g_0}{\omega+ \sqrt{\Lambda g_0+\omega^2}}$$	
	(see the table in	 \cite[Fig. 2, p. 44]{CJS24} for additional details). As shown exhaustively in \cite[Chapter 2]{CJS24}, this wind structures admit natural notions on lengths, geodesics, convexity and completeness. In terms of them, the causal properties of the SSTK spacetime can be determined precisely, \cite[Chapters 3-5]{CJS24}, in particular extending Theorem \ref{t_escala_causal}, see \cite[Th. 5.9]{CJS24}. 
	
	It is worth pointing out that  SSTK structures appear in several relativistic spacetimes, especially, black holes. Indeed,  slow Kerr spacetime presents such a Killing $K$, which is ligthlike in the  {\em limit stationary hypersurface}   which bounds the so-called  {\em ergosphere}\footnote{The ergosphere has relevant physical properties. Indeed, Penrose process can occur, theoretically allowing for energy to be extracted from the black hole's rotation.}. For the specific description of this region and physically related properties, see \cite{JS17b, D}.

\subsection{Application to Riemannian and Finslerian Geometries} \label{s5_1}

\subsubsection{Wind Riemannian manifolds of constant flag curvature} 
In the previous section we have seen that, in a natural way, Randers metrics with Zermelo data $(g_R,W)$ turns out Kropina metrics if one permits $g_R(W,W)=1$ and strong wind Riemannian if 
$g_R(W,W)>1$. Kropina metrics are typically studied as singular Finsler metrics written as a quotient $g_R/\omega$, where $\omega$ is a non-vanishing 1-form and only the region $\omega>0$ is taken into account. 
In Finslerian Geometry becomes natural the 
evolution of the centroid of the indicatrix,
in contrast with the Riemannian case where, in natural coordinates, the indicatrices must be always ellipsoids centered at 0. Then, it is not surprising  that wind Riemannian (or, in general, wind Finslerian) structures may arise as complete solutions for a problem initially posed for Randers (or, in general Finsler) manifolds. In this case, the viewpoint of cone structures provides a natural ``non-singular'' description of the problem, which may lead to a more comprehensive solution.

Remarkably, this was the case of the problem of  classification for Randers metrics of constant flag curvature (CFG). As explained in \S \ref{ss2_5}  flag curvature plays a role with analogies  to sectional curvature in (semi-) Riemannian geometry. However, in striking difference with the semi-Riemannian case, the problem of finding the model spaces of constant flag curvature is completely open in Finsler Geometry. Indeed, a milestone in Finslerian Geometry obtained by Bao, Robles and Shen \cite{bao2004} was the classification of such manifolds  restricted to the class of Randers metrics. Bao et al. {\em local} solution states that the Zermelo data $(g_R,W)$ of a CFC Randers manifold must be: $g_R$ a Riemannian metric of constant curvature $c$ and $W$ any of its Killing vector fields, and, in the case $c=0$, $W$ a homothetic vector field too.    For the global result, however,  $g(W,W)$ may reach the value $1$  in some of the previous cases, then leading to incomplete solutions. These solutions  are inextensible as Randers metrics but can be extended as wind Riemannian ones, as proved in \cite{JS17}.  This leads to the following global result,  \cite[Theorem 3.12]{JS17}, obtained taking into account a natural notion of {\em completeness} for wind Finslerian structures (see \cite[Def. 2.45, Prop. 6.6]{CJS24}).

\begin{thm}\label{e_cfc}
The complete simply connected wind Riemannian structures with constant flag curvature lie in one of the following two exclusive cases, determined by the Zermelo data $(g_R,W)$:

(i) $(M, g_R)$ is a model space of constant curvature and $W$ is any of its Killing
vector fields.

(ii) $(M, g_R)$ is isometric to $\R^n$ and $W$ is a properly homothetic (non-Killing)
vector field.

Moreover, the inextensible simply connected Randers (resp. Kropina) manifolds with constant flag curvature are the maximal simply connected open subsets of the previous wind Riemannian structures where the wind is mild (resp. critical).   
\end{thm}
Such a result also extends the classification of Kropina metrics by Yoshikawa and Sabau \cite{YS}.

\subsubsection{Geodesic refocusing}\label{sss_refocusing} Refocusing is a relevant  property in Riemannian Geometry, studied in its own right by B\'erard-Bergery in \cite{B-B},  which is related to the problem of manifolds whose geodesics are closed. A Riemannian manifold is called a $Z^x$ manifold when all geodesics starting at $x$ return to $x$, and a $Y^x_l$ manifold if every unit-speed geodesic starting at
$x$ arrives at $x$ at time $l>0$. A basic open question is whether there are $Z^x$ non $Y^x_l$ manifolds. Considering two points $x,y$ and geodesics connecting them, analogous definitions $Z^{(x,y)}$, $Y^{(x,y)}_l$ follow and they formalize the idea of {\em strong refocusing} for the geodesics departing from $x$.

 In the Lorentz case, this issue is relevant for the space of cone geodesics $\Ng$ (as pointed out   below Th. \ref{t_CN}).  Indeed,   Chernov and Rudyak \cite{ChernovRudyak}  and then   Kinlaw  \cite{Kinlaw}  stressed the importance of   {\em (null) refocusing},  that is, the property that all the geodesics starting at an event $p$ arrive close to a second one $q$. In the case of a product spacetime $\R\times S$, if $S$ is a $Y^{(x,y)}_l$ then all the lightlike geodesics starting at $p=(x,0)$ will arrive at $q=(y,l)$, i.e. a  {\em strong (null) refocusing} holds. Consistently, Chernov, Kinlaw and Sadykov
\cite{ChernovKinlawSadykov} studied Riemannian refocusing too. 
Recently, Bauermeister \cite{Bauermeister} has  introduced 
the notion of  {\em observer refocusing} in spacetimes
and obtained as a Riemannian corollary 
a partial  solution to the 
 open question on $Z^{(x,y)}$, $Y^{(x,y)}_l$ manifolds: 
{\em an analytic Riemannian $Z^{(x,y)}$ manifold is also an $Y^{(x,y)}_l$} one for some $l>0$. 
The proofs use techniques  transplantable to the Lorentz-Finsler setting too
(including the splitting explained in \S \ref{s_items}, item~\ref{item7}), opening  the possibility to obtain new applications of Lorentz-Finsler geometry.

\subsection{The causal boundary: a link among the three geometries}\label{s5_3}
The {\em causal boundary}  of a spacetime, {\em c-boundary} for short, is a construction initially proposed by Geroch, Kronheimer and Penrose \cite{GKP} as an intrinsic alternative to the extrinsic  conformal boundary commonly used in General Relativity. In their seminal article, these authors pointed out that their own definition of the topology should be improved. This question became a major issue, that hindered the applicability of the c-boundary. What is more,  the c-boundary turned out non-Hausdorff in some simple cases, making it awkward. After the careful study by Flores et al. \cite{FHS11} (see also previous Harris' \cite{Ha00}), it was clear that  a minimal {\em chr-topology} should be always present at this boundary and that  this topology might be non-Hausdorff. Anyway, it was left open the possibility that, eventually,  some authors might prefer to choose a refinement of this topology in certain particular cases. 

The systematic computation of the so-defined  c-boundary for a standard stationary spacetime was carried out by Flores et al. in \cite{FHS11} (see \cite{Ha01} for a previous study of the static case). This research provided some remarkable  links between the c-boundary and others boundaries that could be defined in the Riemannian and in the Finslerian setting. 

Following \cite{FHS11}, the links with the Riemannian setting appeared when considering the c-boundary of a product $(\R\times S, g=-dt^2+g_0)$. For a  Riemannian manifold $(S,g_0)$, one can consider the following three general boundaries: the basic  Cauchy one for incomplete metrics, the celebrated Gromov's compactification for any complete Riemannian manifold \cite{Gr, Gromov}  and the Busemann one for Cartan-Hadamard manifolds (i.e., simply connected Riemannian manifolds with non-positive sectional curvature).  

Essentially, Gromov's identify each $x\in S$  with  $d_x+\R$, that is,  the distance function $d_x$ to $x$, up to an (additive)  constant. Then, $S$ can be regarded as a subset of the space of Lipschitz functions on $S$ up to a constant. As this space is compact, the compactification of $M$ is attained just by taking the closure therein, being Gromov's boundary the frontier points. 

Busemann compactification provides a sphere at infinite (computable by using {\em Busemann functions}) for any Cartan-Hadamard manifold, as proven by Eberlein and O'Neill  \cite{EO}. It turns out equivalent to Gromov's for these spaces. As proven in \cite{FHS11}, the construction is generalizable to any complete Riemannian manifold, now using Busemann-type functions on diverging curves. 
Using a suitable topology, this generalized {\em Busemann boundary} naturally matches with the spacetime c-boundary and may be non-Hausdorff in some patological cases. It is worth pointing out that, in these cases  Gromov boundary presents   awkward properties,  as the existence of points in the boundary non-reachable as limits of continuous curves in the manifold (see \cite[Figures 2, 3]{Sa22}).   When Busemann's is Hausdorff,   however, Gromov and Busemann boundaries are equivalent and become the natural Riemannian counterpart to the c-boundary. 

Moreover, for incomplete Riemannian metrics the Cauchy boundary can be naturally added to both Gromov and Busemann boundaries. When considering the c-boundary of the product spacetime,  this Cauchy boundary  admits a natural interpretation in terms of  {\em naked singularities} (associated with the lack of global hyperbolicity).

In the general case of standard stationary spacetimes, we have seen that the causality is determined by the (non-reversible) Fermat Randers metric $F$. This suggest extending the previous Riemannian setting with counterparts in  Lorentz products to the Finslerian one with counterparts for Randers metrics in stationary spacetimes. 

The extensions of Gromov and Busemann boundaries to the reversible Finslerian setting  have a similar fashion. However,  non-reversible Finsler metrics admit a variety of new elements associated with the non-symmetry of the distance. Indeed, one has a forward and backward completion for each one of the three Riemannian completions, which will correspond to the different behavior of the c-boundary of stationary spacetimes to the future and the past.
A further  explanation of this topic exceeds our purpose. The reader is referred to \cite{FHS_Memo} for the original exhaustive  study,   the Section 5 in \cite{JPS22} for a summary, and the Section 4.4  in \cite{Sa22} for an   instructive particular case. 

  As a final remark, it is  worth pointing out the comparison of the c-boundary with a boundary introduced by Beem for the class of causally simple spacetimes \cite{Be77},  recently used in the setting of low regularity and Lorentz length spaces (see Burgos' thesis \cite{Bu} and references therein). When tested in static spacetimes, one notices that  Beem's boundary does not involve a space compactification, as in the case of Busemann's and Gromov's. Specifically,  Beem's topology is a metrizable topology finer than the minimal  chr-one which permits to introduce a structure of Lorentzian pre-length space in the completion of any globally hyperbolic spacetime. Anyway, the best results are obtained when it coincides with the chr-topology---a property again equivalent to the Hausdorffness of the chr-topology. So, under Hausdorffness, the c-boundary becomes a unifying concept for Riemann, Finsler and Lorentz manifolds, including the framework of low regularity.

\section{Applications in Everyday Physics}\label{s4}
Finsler and Lorentz-Finsler Geometry have been applied to study moving objects and wave propagation in many situations. These include 
Zermelo's navigation problem (the motion of a ship in a current or a zeppelin in the air),  \cite{bao2004,CJS24,JS20,MPR25}, 
general anisotropic waves  \cite{JPS21}, or other specific ones,   such as sound waves \cite{gibbons2010,gibbons2011}, wildfire fronts \cite{JPS23,markvorsen2016} and seismic waves \cite{antonelli2003,yajima2009}, the latter including refraction between different layers of the Earth \cite{BS02,SW99}. 

 The specific Lorentz-Finsler viewpoint  was stressed since the  first version  of \cite{CJS24}. There it is shown the possibility to model, in a unified  and non-singular way, the motion of objects whose  velocities yield a wind-Finsler structure by using a cone structure (as  in Fig. \ref{f_conewind}).   Then, the systematic application of cone structures and Lorentz-Finsler metrics  in \cite{JPS21}, and  subsequent applications as \cite{JPS23}, showed a number of advantages, which will appear next: 

\ben \item 
The possible time-dependence of the model, or {\em rheonomy}, is  geometrized in a natural intuitive picture. 

\item The non-rheonomic case is also benefited, as in the already studied setting of SSTK spacetimes (\S \ref{s5}), or in  the geometric interpretation of  Snell law below (see Fig. \ref{f_Snellbasic}). 

\item Finsler spacetime techniques can solve technical issues about propagation which are  more involved from a purely Finslerian viewpoint (for example,  
see 
\S \ref{ss_Wildfires} below). 

\item 
The relativistic machinery of causality and, especially, event horizons can be used (see Fig. \ref{f_Wild}). 
In particular, cone geodesics  provide first arriving curves and satisfy Fermat principle i.e., they are the critical curves for the arrival time 
not only when first or last arriving  (see Fig. \ref{f_Fermat}) but also in other critical cases, typically  with  noticeable global properties. 

\item Geometric properties admit further interpretations in specific models, as the case of focal and cut points for wildfires,    \S \ref{ss_Wildfires}. 

\een

What is more, the recent spacetime viewpoint for the study of refraction  in \cite{JMPS25}, including Snell formula, not only provides a better modeling for specific cases (as seism propagation through possibly moving layers) but also introduces an essential ingredient for an efficient discretization of classical wave propagation as well as for Numeric Relativity (see  \S \ref{ss_discretization}). 

Next, we will give a brief account of these Lorentz-Finsler applications and references.

\subsection{Fermat, Zermelo and Snell problems}\label{s4_1}

 Hamilton Jacobi equation and optical geometrics (see \S \ref{ss_HJ}) are linked to the classical Fermat principle which states that the trajectories of light rays are the critical points of the arrival time functional, $\int n(\gamma(s))ds$, on curves $\gamma$ connecting the source with each point. This underlies the relativistic and classical applications below, where the time is embedded in the manifold structure even in the classical setting, providing a richer geometric picture. 

\subsubsection{An increasingly general Fermat principle}
Let us pose the  basic 
Fermat principle for a cone structure $(M,\C)$ admitting a time function $t: M\rightarrow \R$.
 Choose  $p\in M$  and an embedded  timelike curve $\alpha$ parametrized by the time $t$.  Consider the set of piecewise smooth lightlike curves from $p$ to $\alpha$ 
$$\Ng_{p,\alpha}:=\{ \gamma: [0,1] \rightarrow M | \gamma \;  \hbox{lightlike},  \; p=\gamma(0), \,   \gamma(1) \in \alpha \}.$$
 Then, define the {\em arrival} functional:
  
 $$\mathcal{T}: \Ng_{p,\alpha} \rightarrow \R \qquad  \gamma \mapsto \mathcal{T}[\gamma]:=\alpha^{-1}(\gamma(1)).$$

\begin{figure}
\begin{center} 
		\includegraphics[width=0.4\textwidth]{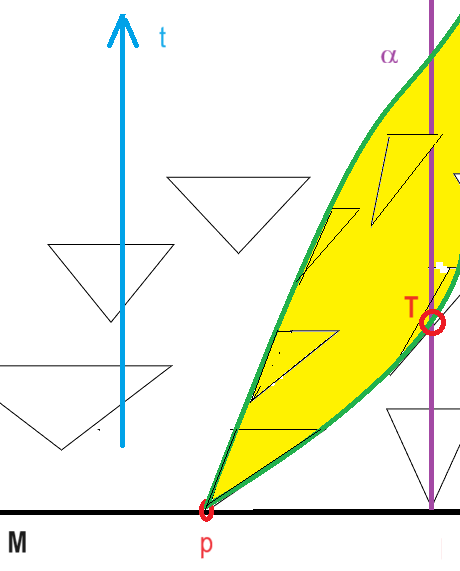}
\end{center}	
\caption{The first and last causal curves from $p$ to $\alpha$ (in green) must be lightlike geodesics by causality considerations ---in higher dimensions    other critical points of the arrival functional can be pictured.}\label{f_Fermat}
\end{figure}
\noindent Fermat's principle is then the following result:
\begin{quote}
{\em  	The set of critical points of $\mathcal{T}$ is equal to the  set of cone geodesics from $p$ to $\alpha$.  }
\end{quote}

\begin{note}
In particular,   the first-arriving causal curves must be cone geodesics. Of course, this can be also achieved trivially from standard causality,  as necessarily the first arriving point must lie in the horismos $E^+(p) (=J^+(p)\setminus I^+(p)$). However, even this case will be relevant in the Snell setting, as it will permit define causality for discontinuous media.  
\end{note}

In the relativistic setting, the Fermat principle was posed by Kovner \cite{K90}  and studied rigurously by Perlick \cite{P90}. This author also noticed that the existence of the global  time $t$ was irrelevant and, indeed the parameter of $\alpha$ for any regular   parametrization (i.e., smooth  with non-vanishing velocity) could play its role for the functional $\mathcal{T}$. As the computations of critical points only uses variations of the arriving curve, global hypotheses in $\C$ are  required only to ensure existence of critical points.   Perlick also made a first  extension of Fermat principle  to the Lorentz-Finsler case \cite{P06}. 

Perlick and Piccione \cite{PP98} also extended the principle to the case when the source of the lightlike curves is a spacelike submanifold $P$ (instead of a point $p$)  and the curves arrive at a timelike submanifold (now actually endowed with a time function in order to redefine properly $\mathcal{T}$). In the setting of SSTK spacetimes above, Caponio et al. \cite{CJS24} considered also the case when the arrival curve $\alpha$ is not timelike,    noticing that just the non-orthogonality between  $\alpha$ and the incident curve $\gamma$ is sufficient to establish Fermat principle (notice that this condition is  automatically  satisfied when $\alpha$ is timelike), see \cite[Th. 7.4]{CJS24}. 

These previous approaches and hypotheses have been  taken into account and optimized in the setting for Snell formula for cone structures  recently developed in \cite{JMPS25}.

\subsubsection{Zermelo navigation}\label{ss_Zermelo}
The Zermelo problem admits now a very general formulation. Consider 
a moving object and let $(M,\C)$ be a cone structure modelling its 
possible (maximum) velocities at each event (point of $M$) and direction. 
Let $p\in M$ be the starting event and $\alpha$  a  target    curve.   Which are the candidates of first arriving (for the 
parameter of $\alpha$) permitted curves  and when do they exist?

It is clear now that the solutions must be  cone geodesics from $p$ to the intersection of $E^+(p) \cap $Im$(\gamma)$, when non-empty. The existence of such a geodesic depends on the global structure of $\C$. In case that $\C$ is causally simple (in particular, when  globally hyperbolic) and $\alpha$ continuosly inextendible towards the infimum of its parameter,   such a geodesic will exist if and only if  $\alpha$ contains at least one point in $J^+(p)$ and  one  $M\setminus I^+(p)$, as this forces the existence of a first point  in $E^+(p) \cap $Im$(\gamma)$. 

If $\C$ is a  SSTK structure\footnote{The following assertions could be extended to Finsler SSTK spacetimes, taking into account footnote \ref{foot_stationary}.} then  its whole  causality  is encoded in a wind Finsler structure $\Sigma$  on a slice $S=\{t=0\}$ through $p$. Thus,  the general conditions which ensure  causal simplicity or global hyperbolicity for $\C$ can be expressed accurately in terms of 
$\Sigma$ as explained in \S \ref{s5_2}. If  $\alpha$ is an integral curve of the Killing vector field $K$ intersecting $S$ at a point  $q$, then the existence of Zermelo solutions is equivalent to the existence of minimizing geodesics $\rho$ for 
the wind Finsler structure from $p$ to $q$. If, additionally, $\alpha$ is parametrized by the  time $t$, then the arrival time is the wind Finsler length of $\rho$. In case that the SSTK is  a 
standard stationary spacetime, the wind Finsler structure $\Sigma$ becomes just a Randers metric and $\rho$ is a usual geodesic. Moreover,  as the integral curves of $K$ are timelike,  the Fermat principle holds with no restriction about  orthogonality at the arrival point in $\alpha$.  From the geometric  viewpoint, this generalizes widely the previous study by Perlick about the stationary case in \cite{P90b}, where the possibility of a  reduction in one dimension less for Fermat principle considered.

\subsubsection{Snell law}
Next, we will follow \cite{JMPS25}. Classical Snell law applies to light rays  crossing two different media, in order to  find the fastest trajectory between two prescribed points. This is characterized by the  elementary formula   
	$
	 \sin (\theta_1)/v_1=  \sin (\theta_2)/v_2$, where $v_i$ is the speed of light at each medium and $\theta_1, \theta_2$ the incident and refracted angles, according to Fig. \ref{f_SnellClassic}. 
	 
\begin{figure}
	\begin{center} 
	\includegraphics[height=0.3\textheight]{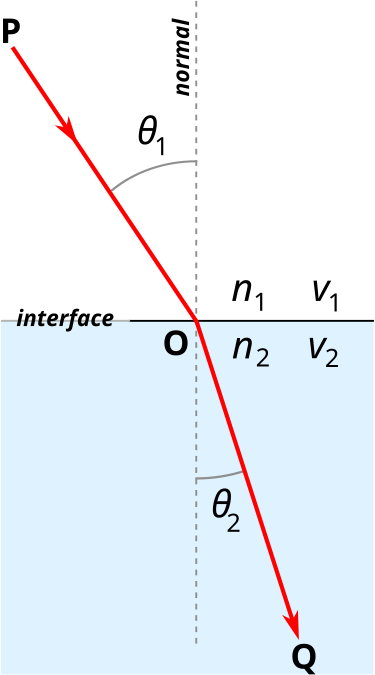}
	\caption{Classic Snell (Credit: Wikipedia ``Snell's law'')}\label{f_SnellClassic}
			
	\end{center}

\end{figure}
	
 When considering the propagation between two cone structures $(Q_1,\C_1)$, $(Q_2,\C_2)$, Snell law becomes a Fermat  problem in a  discontinuous medium. Here light ray departs from an initial event $p\in Q_1$ and must arrive at a curve $\alpha$ modelling the target observer as in Fig \ref{Snellsetting},  namely:
	
\begin{figure}

	\includegraphics[width=0.54\textwidth]{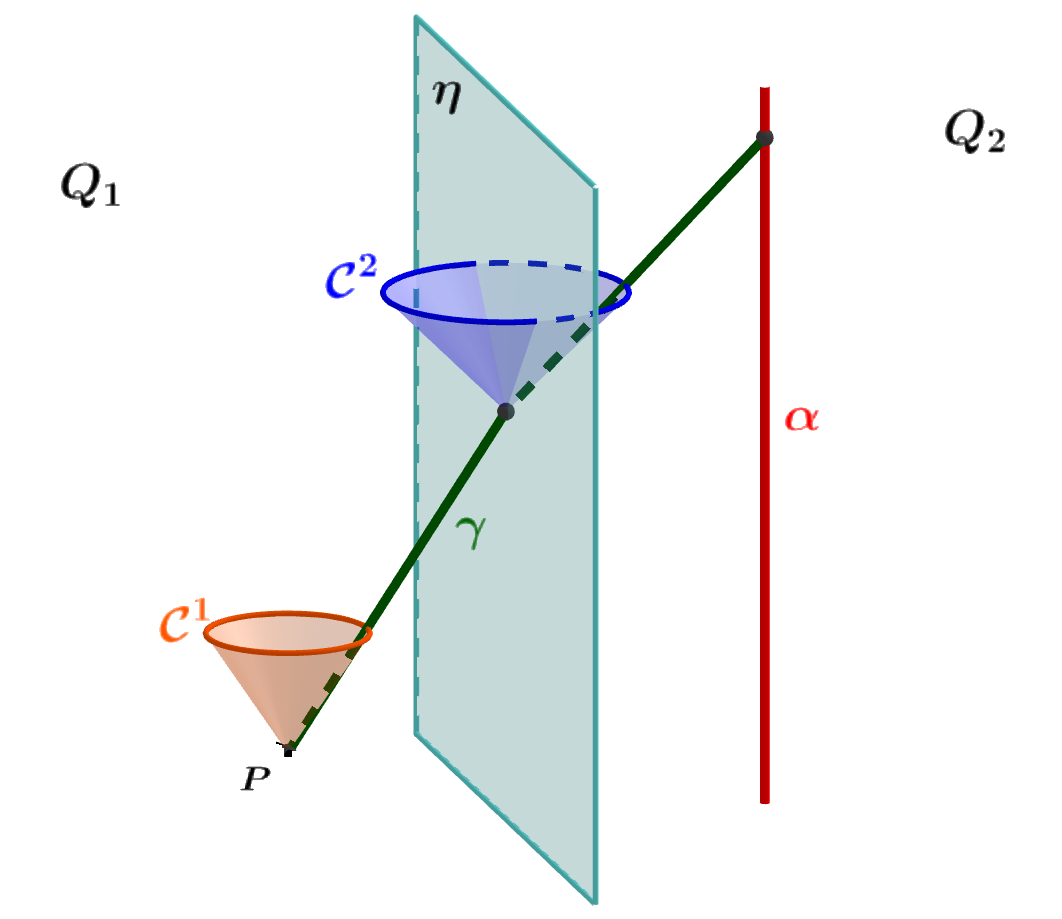}
	\caption{Snell rheonomic setting (Credit \cite{JMPS25})}\label{Snellsetting}
	\end{figure}
	
	\bit 
	\item Event $p$ (source) at medium $Q_1$ with cone ${\C }_1$..
	\item Observer $\alpha$  at medium $Q_2$ with cone $\C_2$.
	(wider cone means $v_2>v_1$)
	\item $\eta$ interface, the critical curve $\gamma$ will be a  lightlike geodesic broken at a point $o\in\eta$.  
	\eit

	\begin{figure}
	
	\begin{center} 
		
			\includegraphics[width=0.5\textwidth]{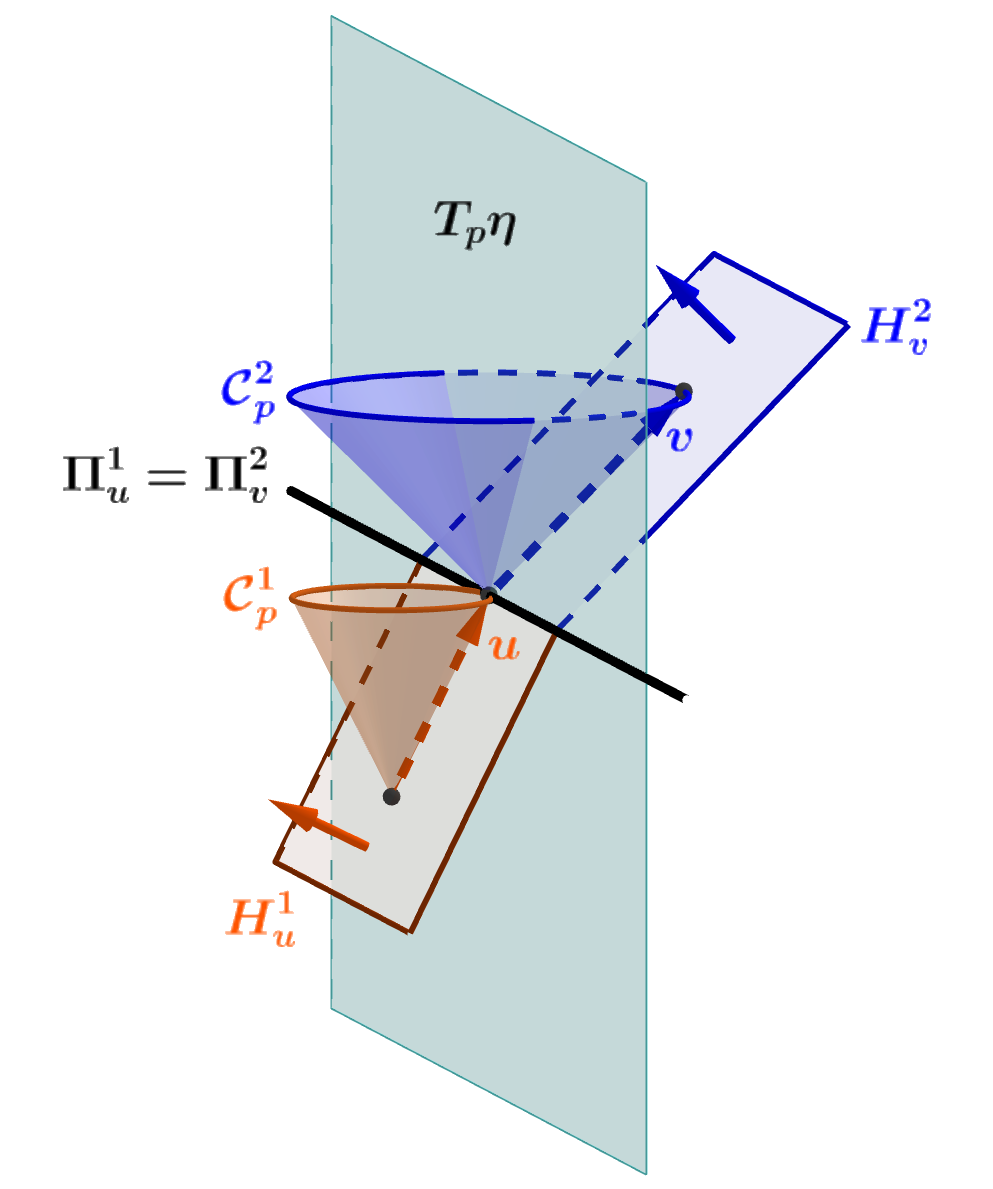}
	\caption{Basic case of Snell formula (Credit \cite{JMPS25})} \label{f_Snellbasic}
	\end{center} 
	\end{figure}

The basic Snell law  for (anisotropic) cones is (\cite[Th. 4.3]{JMPS25}, see Fig. \ref{f_Snellbasic}): 
	\begin{equation}\label{e_Snellbasic}
\Pi_u^1 (:=	H_u^1 \cap T_o\eta) = \Pi_v^2 (:=H_v^2 \cap T_o\eta ) ,
	\end{equation}
where $u$ (resp $v$) is the incident (resp. refracted) $\C_1$-lightlike (resp. $\C_2$-lightlike) direction at $o\in \eta$,   $H_u^1$ (resp. $H_v^2$) is the hyperplane tangent to $\C_1$ at $u$ (resp. to $\C_2$ at $v$) and Snell's formula states that the intersections  $\Pi_u^1, \Pi_v^2$ of these hyperplanes with the tangent space to the interface $T_o\eta$ must be equal. 
There is, however, a number of issues in relation to \eqref{e_Snellbasic} and  Fig.~\ref{f_Snellbasic}:

\bit\item The incident hyperplane $H_u^1$ (and then the refracted one $H_v^2$) are assumed to be transverse to $\eta$.

\item The orientations in these hyperplanes (selected by the side where the corresponding cone lives) must be consistent, as in Fig.~\ref{f_Snellbasic}. This consistency is the additional hypothesis  to ensure local minimization for Fermat principle and, then, to obtain a  {\em Snell geodesic} i.e. a locally horismotic curve.
Indeed, if the blue cone lived on the other side of $H_v^2$, then the broken lightlike geodesic could not minimize the arrival time (and could not be used to compute the causal future of $p$). 

\item  Classical effects as  reflection and total  reflection also appear as a consequence of the Fermat problem.

\item Purely relativistic issues appear when the interface $\eta$ is not timelike (indeed, the causal character of $\eta$ might be  different for $Q_1$ and $Q_2$). 

For example, when $\eta$ is spacelike for $Q_1$ and $Q_2$, the reflection is forbidden. 
Moreover, if the cones point out in consistently oriented directions at $\eta$,   two refracted  trajectories  would appear, see 
Fig.~\ref{f_2refracted}. In this case, only one of them is minimizing (a fact determined again by orientations) providing then a Snell geodesic.
In case of non-consistent orientation, no refracted one  can appear, see Fig.~\ref{f_NOrefracted}.\footnote{In these figures, past cones (which were not considered in $\C$) are depicted and filled for an easier comparison with usual relativistic diagrams.}  
\eit

\begin{figure}
\begin{center} 
				\includegraphics[width=0.5\textwidth]{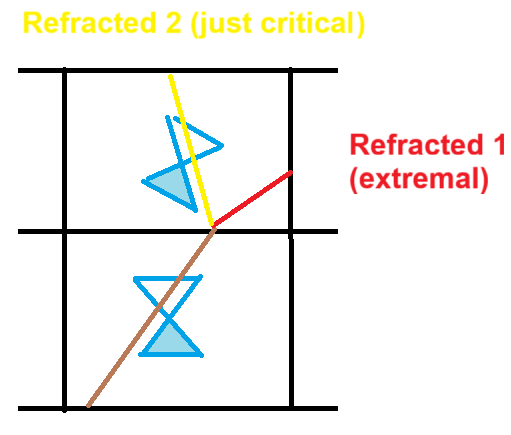}
	\caption{Two refracted trajectories for a spacelike interface. Only the red one is  locally minimizing (Snell geodesic)}	\label{f_2refracted}
\end{center}	

\end{figure}

\begin{figure}

\begin{center} 
	\includegraphics[width=0.5\textwidth]{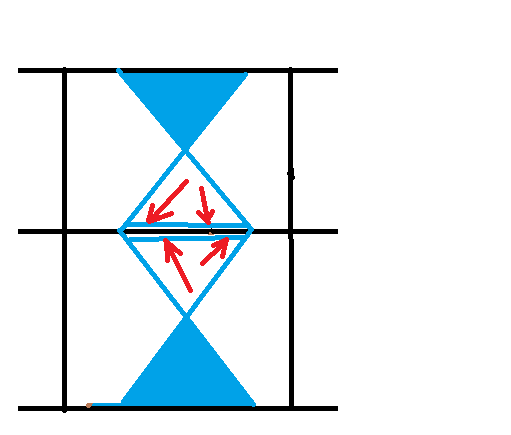}
	\caption{No reflected nor  refracted trajectories in a spacelike interface.}	\label{f_NOrefracted}
\end{center}	

\end{figure}

All these possibilities are analyzed in 
\cite[\S 5]{JMPS25}, as well as the corresponding 
causality in such discontinuous media  
\cite[\S 6]{JMPS25}. About the latter, it is worth pointing out the higher technical difficulty of the Lorentz-Finsler case in comparison with the Lorentz one, stressed in the proof of \cite[Lemma 6.6]{JMPS25}.

\subsection{Applications to wave propagation and discretization}  \label{s4_2}

\begin{figure}

\begin{center}

\includegraphics[height=0.4\textheight]{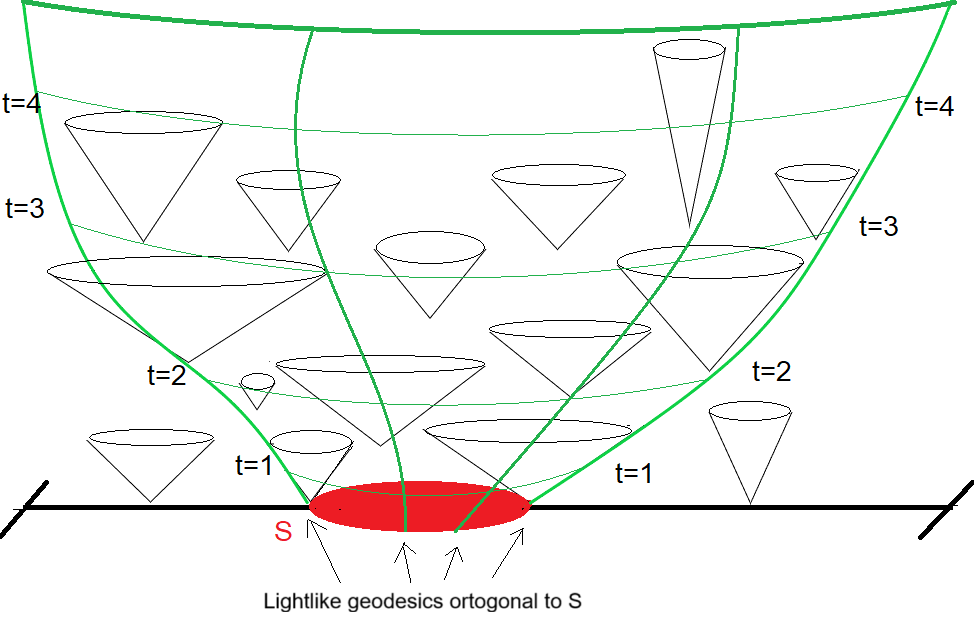}
\caption{Propagation of a wildfire (cone structure viewpoint)}\label{f_Wild}
\end{center}

\end{figure}

\subsubsection{General wave propagation
and wildfires} \label{ss_Wildfires}
Reasoning as for Zermelo navigation in \S \ref{ss_Zermelo}, let $(M,\C)$ be a cone structure modelling the velocity of  wave propagation  in each  event   and direction permitted by a medium.
In the case of Cosmology, the medium is vacuum, but in the case of a problem in classical Mechanics (the propagation of a wildfire or a sound wave), a cone triple for $(\Omega, T, F)$ on $\C$ appears naturally. Namely,  $T$ would  model the observers ``at rest'' on Earth,   $\Omega$ is $dt$, with  $t$ the Newtonian time, and $F$ (which depends on $t$) specifies $\C$. 

Focusing on a wildfire,  it will start at a bounded region $R$ which can be modelled as a compact open subset with boundary $S=\partial R$ of the (spacelike) hypersurface  $t=0$. Then, the evolved burned region is equal to its $\C$-causal future $J^+(R)$. The spacetime frontfire will be its boundary in $t>0$,  which 
matches the horismos $E^+(S)$. The frontfire at each instant $t_0>0$ is $E^+(S)\cap \{t=t_0\}$ (see Fig. \ref{f_Wild}). Now, Lorentz-Finsler geometry applies to show that $E^+(S)$ is composed by the cone geodesics starting at $S$ with initial velocity pointing outside $\R$. Notice that,  at each point $p\in S$, there are exactly two lightike directions orthogonal to $S$, the (non-depicted) second one pointing inwards $R$.   

 This  general viewpoint of anisotropic wave propagation was developed in \cite{JPS21}. It includes the short time existence of solutions \cite[Theorem 4.8]{JPS21}. This result is more involved now than in the  purely Lorentz case (see  \cite[Lemma 4.7]{JPS21}), indeed, it had not been carried out previously (i.e., with a non-spacetime approach) as far as we know. 
 
 The particular case of wildfires is analyzed detailedly in \cite{JPS23}. This includes a comparison with actual wildfire monitoring, which is based in rough  approximations to the wave equation.     
As emphasized therein  \cite[\S 5.2]{JPS23}, focal and cut points   along cone geodesics become points where wildfires are more intense and firefighters might get trapped.

\begin{figure}

\begin{center}

\includegraphics[height=0.23\textheight]{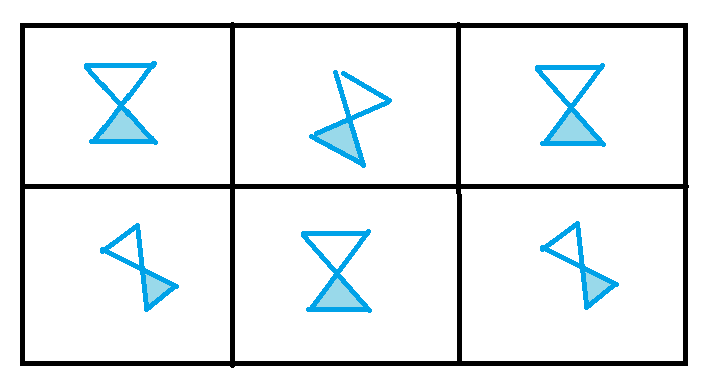}
\caption{A discrete grid for both classical  wavefront propagation and Numeric Relativity}\label{f_Grid}
\end{center}

\end{figure}

\subsubsection{From seisms to Numerics and Particle Physics}\label{ss_discretization}
The modelization of seisms is similar to  wildfires. However,  it is natural  to consider  discontinuous layers for seisms  and, thus, a systematic use of Snell formula. In this classical setting,  the interface between different layers yields a timelike interface $\eta$ where the basic Snell formula \eqref{e_Snellbasic} is enough. Indeed, typically, the interface will be ``constant in time'' 
so that its spacetime evolution 
is trivial\footnote{Certainly, the effect of the seism itself could modify the interface, however, this would be negligible for the computation of the sismic wavefront. Moreover, in any case, it is not expected that the velocity of the movement of the interface could be compared    
with the velocity of seism propagation so that \eqref{e_Snellbasic} would remain.}.

Anyway, the use of Snell formula can go far beyond the scope of seisms, and require the use of other cases explained below \eqref{e_Snellbasic}. Specifically:

\ben\item Computer aided modelizations of classical systems  must use  discretizations. In a natural way, this modeling  yield a grid where each cell would contains a cone (Fig. \ref{f_Grid}). Then, the propagation of the wave between each two grids will be controlled by the basic Snell formula \eqref{e_Snellbasic}.

\item  Such a discretization migth be useful for Numerical Relativity. The usual way to proceed in this field is considering at each cell the Lorentzian metric at one of its points. Notice, however, that such a Lorentzian metric is determined by the  cone structure $\C$  up to a conformal factor. Thus, the separate study of $\C$ as above and this factor might be more efficient. 

For this purpose,  the interface would be generically timelike or spacelike. 
The first case is essentially governed by  Snell formula \eqref{e_Snellbasic}. 
For spacelike interfaces, on the one hand, the possibility of double refraction (as in Fig. \ref{f_2refracted}) appears naturally. Then,  the possibility to distinguish the Snell geodesics by using orientations would be useful. On the other hand,  the case of no refraction (Fig. \ref{f_NOrefracted}) would mean either that the grid approximation is overly rough or  
that a singularity appears.

\item \label{item_effective_discrete_cone} Cone structures might be useful also at the fundamental level of particle physics\footnote{These ideas were discussed in the meeting ``Causal fermionic systems 2025'', Regensburg October 6-10th.}. Indeed, whatever the fundamental theory one chooses, 
a semiclassical propagation of interactions should emerge for matter and energy. This propagation would yield   cones  (resembling infinitesimal  skies for the space $\Ng$ of cone geodesics),   possibly anisotropic as considered here. Depending of the theory, these cones would match in  a smooth cone structure or a quantum discretization \cite{BOW,Cuzi,Loll,Surya}. In both cases, they would be modelizable with the introduced tools.        
\een

\section{Finslerian Relativity and Einstein equations}\label{s6}
The topic of Finslerian modifications of Relativity would merit a survey far more extensive than the present one.
 Here, following some physical foundations in \S \ref{s6_1}, presented in a style accessible to mathematically oriented readers, we will focus on Finslerian versions of the vacuum Einstein equations.
    Specifically, we will examine their variational formulations within both the Hilbert (\S \ref{s6_20}) and Palatini 
    (\S \ref{s6_4}) approaches ---highlighting their inequivalence when Lan$_i \neq 0$, in contrast to the Lorentzian case--- and discuss some vacuum solutions (\S \ref{s6_3}).

\subsection{Physical foundations}\label{s6_1}

\subsubsection{Very Special and General Relativity}
In 2006, Cohen and Glashow \cite{CG} proposed \emph{Very Special Relativity} (VSR), suggesting that much of the observed physics of Special Relativity follows not from the full Lorentz 6-dim. group  SO(1,3), but from a smaller 4-dim. subgroup SIM(2), 
the group of Lorentz transformations preserving a particular lightlike direction Span \( \{n^\mu\} \). This group still preserves the constancy of the speed of light and remains consistent with quantum field theory. A smaller 3-dim. subgroup  HOM(2)  preserving exacty \( n^\mu \)  would also satisfy these requirements.
Such groups  break rotational symmetry partially, but preserve enough structure for most relativistic kinematics to hold. Anyway, they imply that, physically,  Lorentz invariance is only approximate or emergent.

In the 1970s, Bogoslovsky \cite{Bogos} had already proposed a Finslerian generalization of Minkowski spacetime, which fulfilled above requisites. Namely,
\begin{equation}\label{e_Bogos}
   \left( \eta_{\mu\nu} dx^\mu dx^\nu \right)^{\frac{1-b}{2}} \left( n_\rho dx^\rho \right)^b,
\end{equation}
where \( \eta_{\mu\nu} \) is the Minkowski metric,
   \( n^\mu \) is a fixed lightlike direction (\( \eta_{\mu\nu} n^\mu n^\nu = 0 \)) and \( b \) is a small, dimensionless parameter measuring deviation from Lorentz invariance\footnote{\label{f_smooth} Bogoslovsky metric is an {\em improper} Lorentz norm, according to our definitions, as it is not smooth at the distinghished lightlike direction. 
   	Anyway,  this means only that other specific geometric  tools might be required for their treatment. Notice also that if a power of $L$ is smooth outside the zero section (as in \cite{PW}),  this is typically  enough for the regularity of both, the  cone and the indicatrix.}.    
   In 2007, VSR as well as its natural generalization to possibly curved spaces (called subsequently {\em General Very Special Relativity}), were clearly  identified by Gibbons, Gomis and Pope as Finsler Geometry  \cite{GGP};  shortly after,   Finslerian Cosmology was considered by Kouretsis et al. \cite{KSS0}.  Fuster, Pfeifer and Pabst \cite{FPP} studied Berwald spacetimes as a non-flat generalization of the line element used in VSR, and promoted the tidier name  {\em Very General Relativity}.  Pfeifer and Wohlfart \cite{PW} carried out a systematic construction of gravitational dynamics, including field equations (see \S \ref{ssPW}), matter coupling, and observer definitions. Hohmann \cite{Hoh} explored relations between Finsler spacetimes and Cartan geometry.
   
In any case, to drop the beloved Lorentz symmetry is a drastic  proposal, and different arguments have appeared from both  the {\em effective} and the {\em fundamental viewpoints}.   

\subsubsection{The effective viewpoint} As pointed out above (\S \ref{ss_discretization}, item \eqref{item_effective_discrete_cone}), aniso\-tropic cones may emerge as an effective model for propagation coming from  any fundamental theory. In this vein, 
Girelli, Liberati and Sindoni \cite{GLS}  noticed  in 2007 that Planck-scale modifications to the usual energy-momentum relation (dispersion relations) in Quantum Gravity, can be understood by using Finsler geometry. These modifications are a common feature of many Quantum Gravity (QG) phenomenology approaches, and some models proposed that this phenomenology could be associated with
an energy dependent geometry called ``rainbow metric''. The authors showed that rainbow metrics are naturally described by Finsler geometry, then providing a 
mathematical framework for the possible geometric structures  arising in
the semiclassical regime of QG. 
Moreover,  Edwards and  Kosteleck\'y \cite{EK} connected effective field theories with Lorentz violation  to Finsler geometry and  showed that  classical point-particle Lagrangians from Lorentz violating  operators correspond to Finsler structures. Major cosmological issues as dark matter have been also modeled effectively by using Finsler elements, but the requirements for such a modelization may be disputed and go beyond our scope, see for example \cite{ChanLi, Vacaru}.

\subsubsection{The fundamental viewpoint}  
The foundations  of Finsler spacetimes from the observers viewpoint have been widely discussed in \cite{BJS_Universe}, including classical   axiomatic  approaches by von Ignatowsky (1910) \cite{Ignato} and  Ehlers, Pirani and Schild (1972) \cite{EPS}.  Bububianu and Vacaru \cite{Bubu} introduced a different axiomatic approach on gravity theories with modified dispersion relations,  and  other viewpoints can be seen in Kouretsis et al. \cite{KSS} or Vacaru \cite{Vacaru0}.  Next, a very brief summary of the classical ones is presented.  

 \smallskip \smallskip

\noindent Following  \cite[Sect. 2-4]{BJS_Universe} (which builds on \cite{BLS, Ignato}),  the  theories for space and time  
such as  Classical Mechanics and Special Relativity, rely on the existence of inertial frames of references (IFR). Such a frame $R$  
assigns four coordinates $(t,x^i)$ to the whole set of events. The first one $t$ is named {\em temporal}, the other three $x^i$ {\em spatial}, and they obey the 
following
two postulates: 

(1) linearity: the transformation of coordinates 
between two IFR's $R, \bar R$ is an affine isomorphism $\R^4\rightarrow \R^4$, 

(2) time interchangeability and space 
interchangeability: if $(\bar t,\bar x^i)$ are the coordinates of $\bar R$ then:
\begin{align}\label{e_post2}
\partial\bar t/\partial t&=   \partial t/\partial\bar t,
&&
&\partial \bar x^i/\partial x^j&=\partial x^j/\partial\bar x^i, \qquad \qquad \forall i,j=1,2,3,
\end{align} 
or, in a more physical language:  the temporal coordinate $\bar t$  (resp. the spatial coordinate $\bar x^i$) of $\bar R$  measured by using the physical clock (resp. the rod for the direction $x^j$) of $R$,  goes by as  the temporal coordinate $t$  (resp. the spatial coordinate $x^j$) of $R$  measured by using the physical clock (resp. the rod for the direction $\bar x^i$) of $\bar R$.

Then, a computation shows the existence of one (generically unique) group where all the changes of coordinates among IFR's fall. This group lies in one of the following four cases: 

(a) the Galilei group of classical Mechanics, 

(b) the Lorentz group with lightspeed $c>0$ (where $c$ can be interpreted as a unit normalization of the quotient between any spatial coordinate and the temporal one),  conjugate to the orthochronous Lorentz group O(1,3), 

(c) the orthogonal group up a normalization $c>0$  (where $c$ can be interpreted as above),  conjugate to the orthogonal one O(4), and 

(d) a group dual to Galilei, some times called {\em Carroll group}. 

\smallskip

\noindent Each one of the  postulates (1) and (2) above assumes the possibility of a linearization for measurements of  spacetime. Considering the case of the group (b), to drop the first postulate means to go from Special to  General Relativity  (under the assumption  that the now pointwise constant $c$ can be regarded as a true constant on the manifold of all the events). Then, dropping also the second postulate, one arrives at Lorentz-Finsler Geometry.

Summing up, Lorentz-Finsler metrics emerge from General Relativity  after a relaxation of the assumptions on  symmetries similar to the emergence of General Relativity from the Special one. In \cite{JSV_Universe}, it is carried out a detailed development of the meaning of measurements (for both, kinematics and stress-energy tensor) by comparing  the Lorentz-Finsler case with the standard relativistic one.  

\smallskip
\smallskip
\noindent
Ehlers, Pirani, and Schild's approach (EPS) \cite{EPS} offers a constructive axiomatic framework for General Relativity based on observable structures rather than assuming a pre-existing spacetime metric. Their goal was to derive the geometric structure of spacetime from physical measurements starting with two fundamental observational structures:

(A) Light rays, which will define a cone structure. 

(B) Free-falling particles, which will define a projective structure. 

\smallskip

\noindent By combining these two structures under certain compatibility conditions, EPS argues that the Lorentz metric of spacetime can be reconstructed ---light rays and free-falling particles being modelled as lightlike and timelike pregeodesics---  then providing the seeked foundation for General Relativity. 

However, one can wonder why general Lorentz-Finsler metrics are excluded for EPS. Tavakol and Van den Berg \cite{Tavakol} pointed out  that a Lorentz norm (or, whith more generality a Berwald space) should  satisfy  EPS axioms, as their geodesics underlie an affine space. L\"ammerzahl and Perlick \cite{LP}  argued against the unjustified requirement of  smoothness for the radar coordinates  introduced around each event $e$ in the EPS first axiom\footnote{Specifically,  the smoothness  at $e$ of a function  $g(p)=-t(e_1)t(e_2)$, meaning the product of local arrival parameters  of two light rays (then satisfying $g(e)=0, g_{,a}(e)=0$).}. 

The careful analysis in \cite[\S 5]{BJS_Universe} identifies the two (unjustified) physical hypotheses on smoothness which turns out the exclusion of general  Lorentz-Finsler metrics, essentially: 

(i) The requirement of smoothness pointed out by L\"ammerzahl and Perlick would imply to ask for  the smoothness at 0 of $F^2$, where $F$ is a  Finsler metric, specifically, the one in a cone triple $(\Omega, T, F)$ associated with the cone $\C$ provided by the light rays see \cite[Example 5.1]{BJS_Universe}. By Lemma \ref{l_norms}, this would make $F^2$ Riemannian  and, then, $\C$ will be Lorentzian. Moreover, even in Lorentz-Minkowski spacetime there are situations where such a smoothness fail, see \cite[footnote 24]{BJS_Universe}.

(ii) In order to deduce the existence of the projective structure, EPS considers a general {\em law of inertia}  
for freely falling observers. They write it in coordinates  as
\begin{equation}\label{e_law_of_Finsler_inertia}
	\ddot{x}^a + \Pi^a_{bc} \dot{x}^b \dot{x}^c=\lambda \dot x^a
\end{equation} (see \cite[formula (7)]{EPS}). Here, $\lambda$ depends on the parameterization $x^a(u)$  of the curve and $\Pi^a_{bc}$ 's depend on $x^a$ and  are called the {\em projective coefficients} (indeed,  they  determine a projective structure $\mathcal{P}$ compatible with some affine connection). However,  $\Pi^a_{bc}$ 's are not permitted to depend on the  direction of the  velocities $\dot x^j$, and this not justified, see \cite[\S 5.2.2]{BJS_Universe}. Indeed, if permitted   a non-linear connection would be obtained (its exponential being as regular as explained in~\eqref{e_Bryant1}).

 \smallskip
\smallskip
\noindent
Summing up, from the purely geometrical viewpoint, Lorentz-Finsler metrics are physically permitted (whenever additional conditions on isotropy are not explictly imposed) as a consequence of the possible infinitesimal non-linearity of spacetime. 

\subsection{Finsler Einstein-Hilbert variational approach}\label{s6_20}

\subsubsection{First approaches} \label{s6_21} Given any semi-Finsler metric $L$ with fundamental tensor $g$ and any vector field $V$ lying in its domain, one can define the {\em osculating} semi-Riemannian metric $
	g_V(x):=g(x,V(x))$.
  Asanov \cite{Asanov} carried out a variational approach for such  metrics leading to  an equation formally similar to Einstein's vacuum one. However, the subtleties of the osculating metrics (see for example the interpretation of Chern anisotropic connection in \S \ref{ss_Anisoropicconnections}) make unclear its implications.  
  
  Miron  \cite{Miron} considered Finsler metrics from a general Lagrange viewpoint. Starting at the fundamental tensor $g$, a Sasaki-type metric on $A\subset TM\setminus\mathbf{0}$, which would admit a classical Einstein equation, is constructed. 
No further physical interpretation was given then, anyway, the
Lagrangian theory and its physical applicability is analyzed further by Miron et al. \cite{MRA}.

Rutz \cite{Rutz} considered the  analogy with General Relativity and Newtonian Gravity, where the infinitesimal geodesic deviation must vanish in vacuum. As the Ricci scalar in 
\eqref{eq:curvature and torsion} depends only on the nonlinear connection,  she proposed
 as  Finslerian vacuum  equation 
\begin{equation}\label{e_Rutz}
	\mathrm{Ric}=0. 
\end{equation}
This choice, however,  did not take into account the variational viewpoint and, in fact, this equation cannot be obtained as the Euler Lagrange one for an action functional. Anyway, the systematic procedure of variational completion in \cite{VoicuKupka} is applicable. As proven by   Hohmann,  Pfeifer and  Voicu \cite{HPV}, the variational completion of Rutz's  equation will coincide with the variational equations  below obtained by Pfeifer and Wohlfart in \cite{PW}, then providing an extra support to  the latter.
	
\subsubsection{Pfeifer--Wohlfarth (PW) equation} \label{ssPW}
Pfeifer and Wohlfarth~\cite{PW} proposed an action functional for the gravitational sector of Finsler geometry by integrating the Ricci scalar on the indicatrix or, equivalently \cite[\S 3]{{JSV22}}, 
\begin{equation}
    \mathscr{S}[L] := \int_{D\subset \mathbb{P}^+A} g^{\alpha\beta}\,
    \mathrm{Ric} 
    _{\cdot \alpha \cdot \beta}
    \, d\Sigma^+ ,
\end{equation}
 where
 $d\Sigma^+$ denotes the natural volume form on the positively projectivized tangent bundle $\mathbb{P}^+A$. This action is closely related to the Einstein--Hilbert functional in Riemannian geometry, and  represents the simplest choice based on the canonical curvature derived from the connection coefficients $\mathring{G}^\mu$.   The Cartan tensor (a relevant ingredient of the theory) is not involved here, but it might be incorporated at a further stage of the theory, eventually  including couplings with  particles. 

Varying the action $\mathscr{S}$ with respect to the Finsler function $L$ and adding a matter contribution on the right-hand side leads to the field {\em PW equation}
\begin{equation}\label{e_PW}
	\mathrm{Ric} - \frac{1}{6} g^{\alpha\beta} \mathrm{Ric}_{\cdot \alpha \cdot \beta} L + \{\text{terms involving } \mathrm{Lan}_{\mu}\} = \mathfrak{T},
\end{equation}
(recall the dot notation for vertical derivatives in \eqref{e_Berwald_et_al}) where $\mathfrak{T}$ denotes the energy--momentum distribution function on $A$, and $\mathrm{Lan}_{\mu}$ is the mean Landsberg tensor (\S \ref{ss_Anisoropicconnections}).  PW equation makes sense in any signature. For semi-Riemannian metrics in vacuum ($\mathfrak{T}=0$), it is equivalent to impose Ricci flatness, but not to Rutz equation \eqref{e_Rutz} in general.

A conservation law associated with diffeomorphism invariance was derived by  Hohmann  et al. \cite{HPV22}:
\begin{equation}
    \int_{\mathbb{P}^+(A_x)} y^\alpha \mathfrak{T}_{\mid \alpha} y_\mu \, d\Sigma_x^+ = 0, \qquad \forall x \in M,
\end{equation}
which is consistent with known physical conservation laws, such as those arising in the Liouville equation for kinetic gases. What is more, the authors  conjectured that the refinement of $\mathfrak{T}$ from a perfect fluid description to a kinetic gas, combined with the effects of the Landsberg tensor, could give rise to an accelerated cosmic expansion without invoking dark energy \cite{HPVkineticgas, PVFP}.

\subsection{Exact Vacuum Solutions to PW's} \label{s6_3} To date, only a limited number of explicit solutions to the Pfeifer--Wohlfarth equation are known. The examples discussed next illustrate the rich geometric structure of Finslerian gravity and its potential relevance for cosmology and gravitational physics.
A comprehensive review of exact solutions to the Pfeifer--Wohlfarth equation can be found in the PhD thesis by S. Heefer~\cite{Heefer}, which is our main reference for the brief summary here. 

The solutions are typically analytic, defined on the whole $TM$ but possibly containing singularities in some directions,  which may lie in the  causal cone.	 

\subsubsection{$(\alpha,\beta)$-metrics and pp-waves}
The so-called $(\alpha,\beta)$-metrics  are a special class of Finsler metrics which can be expressed as:
\begin{equation}
	F = \alpha \, \phi\!\left(\frac{\beta}{\alpha}\right),
\end{equation}
where $\alpha$ is (the 1-homogeneous root of) a Riemannian metric, $\beta$ is   a 1-form 
and $\phi$ is smooth function satisfying conditions ensuring the Finsler character;
this definition extends directly to  Lorentz case. As a first result:

\begin{thm} \cite[Th. 8.8.1]{Heefer}. Let $F$ be any $(\alpha,\beta)$-metric such that
	$\alpha$  solves the classical Einstein equations in vacuum and $\beta$ is parallel respect to $\alpha$.
	Then, the nonlinear connection of $F$ coincides with the Levi-Civita connection
	of $\alpha$ and  $F$ becomes a (Ricci-flat) vacuum solution to PW equation.
	\end{thm}
This result demonstrates that there exist infinitely many Finslerian generalizations of the Minkowski vacuum.

In dimension 4, the hypotheses lead naturally either to Lorentz norms or to the following {\em Finsler pp-waves}.

\begin{thm} \cite[Th. 8.8.2]{Heefer}
Let $(u,v,x^1,x^2)$ be local coordinates. Any $(\alpha,\beta)$-metric defined by
\begin{equation}
    \alpha = \sqrt{ | 2\dot{u}(\dot{v} + H(u,x^1,x^2)\dot{u}) - (\dot{x}^1)^2 - (\dot{x}^2)^2 | }, \qquad \beta = \dot{u},
\end{equation}
where $H$ satisfies the harmonic condition $(\partial_{x^1}^2 + \partial_{x^2}^2)H = 0$, is a solution of the PW vacuum equation.
\end{thm}
 Finsler pp-wave represents plane-fronted gravitational waves with parallel rays, and they are a relevant class of Finsler spacetimes \cite{Fuster_pp_wave}.
 
 Another type of  $(\alpha,\beta)$- metrics are the Kropina ones (which has appeared sometimes above) and correspond to the case  $m=1$ of the $m$-Kropina metrics, defined by $\phi(s)=1/s^m$, that is,
 \begin{equation}\label{e_kropina}
	F = \frac{\alpha^{1+m}}{\beta^{m}}, \qquad m\in \R.
\end{equation}	
These metrics also include  Bogoslovsky's in \eqref{e_Bogos}, and we refer to 
 \cite{PHF, Heefer} for a systematic study.

\subsubsection{Unicorns and Finslerian Cosmologies}  \label{ss_unicorn}

Building upon Elgendi's unicorns in Finsler metrics with singularities, Heefer et al. \cite{HPRF} 
arrive at  semi-Finsler unicorns with singularities and, remarkably, some of them  are solutions of PW equations. More precisely (see \cite{Heefer}, formulas (10.8), (10.30) ---and below--- and Prop. 10.2.2):

\begin{thm} 
Let $\Phi(\hat{y}) = c_{ij} y^i y^j$, with $\hat{y} = (y^1, y^2, y^3)\in \R^3$, be a non-degenerate quadratic form, and  
$f>0$ a smooth function. A semi-Finsler metric $F$  (with singularites) is  Landsberg non-Berwald, if it is written  locally 
\begin{equation}
    F(x,y) = f(x^0) 
    \left( |y^0| + \mathrm{sgn}(\Phi)\sqrt{|\Phi|} \right)
    \exp\!\left( \frac{|y^0|}{|y^0| + \mathrm{sgn}(\Phi)\sqrt{|\Phi|}} \right),
\end{equation}
where the last expresion is defined as 0 when so is the denominator of $\exp$.

Moreover, $F$ is then a solution of the PW vacuum equation if and only if there are constants $a_0,a_1$ with $a_0>0$ such that $f(x^0)= a_0 e^{a_1 x^0}$, for all $x^0$.  
\end{thm}

The light cones for $F$ are equal to those of Lorentz-Minkowski if  the signature of $\Phi$ is either Lorentzian $(+,+,-)$ or negative $(-,-,-)$ . 
However, the signature of the fundamental tensor $g$ is not Lorentzian inside the cone in these cases, \cite[Prop. 10.1.4]{Heefer}.
Indeed, for Lorentzian $\Phi$, one can consider the sets
 $\mathcal{S}(\Phi), \mathcal{T}(\Phi)$ containing, resp., $\Phi$-spacelike and $\Phi$-timelike directions. Then,  $g$ has Lorentz signature on  $\R\times \mathcal{S}(\Phi)$ but neutral $(+,+,-,-)$ on $\R\times \mathcal{T}(\Phi)$. For negative definite $\Phi$ the signature of $g$ is also negative definite.
 
 These properties may  obscure the relativistic interpretation of such  unicorns. However, they present further interesting properties to be explored in both the cosmological and purely geometric setting. Indeed, in the case of negative definite $\Phi$, the unicorns  correspond to a linearly
expanding (or contracting) universe with a cosmological
symmetry which make them reminiscent of FLRW cosmology.  They are not only conformally flat as FLRW's but also flat for the curvature of the non-linear connection,  that is $\RN_{ij}^k=\delta_j N_i^k-\delta_i N_j^k \equiv 0$ (recall \eqref{e_Berwald_et_al}). This is precisely a consequence of the non-Berwaldian character of the metric and, thus, a  potentially measurable Finsler effect, beyond the scope of standard General Relativity.
 It is worth pointing out that  the (unique) unicorn generalization of Friedmann-Lemaitre-Robertson-Walker geometry has been recently obtained in  \cite{FPVPH}.

\subsection{Finsler Einstein-Hilbert-Palatini approach} \label{s6_4}
The huge freedom introduced by the Finslerian viewpoint permits to consider a good number of functionals, depending on the focused  geometric level of the structure (non-linear, anisotropic or Finsler connection), see the discussion in \cite{SV}.  Anyway, the so-called Palatini\footnote{It is kwown that this aproach was introduced by Einstein himself, but we use it following \cite{JSV22}  because it is widely spread (an example of  
Stigler's Law of Eponymy) and distinguishes   from other related formalisms.} formulation is a natural extension of the Hilbert variational approach which fits especially well to the Finsler viewpoint, developed in \cite{JSV22} (our main reference below). 
Indeed, Palatini's promote the nonlinear connection (which determines the freely falling observers) to a geometric field  independent of the metric (which provides numbers from observations), as they come from a priori independent phenomena. This theory   has an independent variational interest in semi-Finsler Geometry and formally recovers Rutz's \eqref{e_Rutz}, a desirable property  obtained by PW only when the mean Landsberg tensor vanishes. What is more,  it diverges significantly of PW  when  $\mathrm{Lan}_\mu \neq 0$, as the geodesics of the connection  will differ then (in particular,  predicting  experimentaly testable differences).

\subsubsection{Palatini Formalism in Semi-Riemannian Geometry}

Recall that the Einstein-Hilbert action in the Palatini formulation is
\begin{equation}
    \mathscr{I}_{\mathrm{EHP}}[g_{\mu\nu}, \Gamma^\mu_{\nu\rho}] 
    = \int g^{\alpha\beta}\,\mathrm{ric}_{\alpha\beta}\, d\mathrm{Vol},
\end{equation}
where the metric $g_{\mu\nu}$ and the affine connection $\Gamma^\mu_{\nu\rho}$ are treated as independent variables,  $\mathrm{ric}_{\alpha\beta}$ is the Ricci tensor of the connection, and the metric provides the contraction and volume. A pair $(g_{\mu\nu}, \Gamma^\mu_{\nu\rho})$ is a critical point of $\mathscr{I}_{\mathrm{EHP}}$ if and only if $g_{\mu\nu}$ satisfies the Einstein field equations and the affine  $\Gamma^\mu_{\nu\rho}$ satisfies, in terms of  Christoffel symbols $(\Gamma^g)_{ij}^k$ 
of the Levi-Cita connection of $g_{\mu\nu}$ 
(see, for example \cite{Bernal_et_al}), 
\begin{equation}\label{e_Palat}
	\Gamma_{ij}^k= (\Gamma^g)_{ij}^k + \mathcal{A}_i \otimes \delta_j^k  
\end{equation}
where $\mathcal{A}$ is a 1-form (with no restriction) and $\delta_j^k$ is Kronecker's delta.
The unique symmetric connection in \eqref{e_Palat} is the Levi-Civita one, then providing a remarkable support to standard General Relativity.

\subsubsection{Palatini Approach to Finslerian Einstein Equations}

Following \cite{JSV22}, the starting point is a generalization of Pfeifer and Wohlfarth's action,
\begin{equation}
    \bar{\mathscr{S}}[L, N^\mu_\nu] := 
    \int_{D\subset \mathbb{P}^+A} g^{\alpha\beta}\, \mathrm{Ric}_{\cdot\alpha\cdot\beta}\, d\Sigma^+,
\end{equation}
where $L$ and $N^\mu_\nu$ are treated as independent variables, and $\mathrm{Ric}$ denotes the Ricci scalar constructed from $N^\mu_\nu$.

Variation of $\bar{\mathscr{S}}$ with respect to $L$ and $N^\mu_\nu$  yields, after removing torsion contributions,  the coupled system:
\begin{align}
    \mathrm{Ric} - \frac{1}{n+2} g^{\alpha\beta}\mathrm{Ric}_{\cdot\alpha\cdot\beta}\, L &= 0,
    \label{eq:M}\\[2mm]
    \mathcal{Z}^\alpha_{\cdot\alpha\cdot\mu} 
    + \left( C_\alpha - 3\frac{y_\alpha}{L} \right)\mathcal{Z}^\alpha_{\cdot\mu}
    &= \mathrm{Lan}_\mu,
    \label{eq:A}
\end{align}
where $
    \mathcal{Z}^\mu := \tfrac{1}{2} N^\mu_\alpha y^\alpha - \mathring{G}^\mu.
$

Equation~\eqref{eq:M} resembles the PW equation \eqref{e_PW} but notably contains no explicit Landsberg terms. The Ricci scalar $\mathrm{Ric}$, however, is computed from $N^\mu_\nu$ rather than from the canonical connection $\mathring{G}^\mu_{\cdot\nu}$, motivating the need to understand $\mathrm{Ric}$ in terms of $\mathcal{Z}$ for direct comparison.

\subsubsection{Some results on the Finslerian Palatini Equations} The following sampler of results in \cite{JSV22} is illustrative  (see more in  \cite[Th. D]{JSV22}). The first one stresses an aforementioned fundamental difference with PW.  

\begin{thm} For solutions to \eqref{eq:M} and \eqref{eq:A}:
	
\begin{enumerate}
    \item If $\mathrm{Lan}_\mu = 0$, then the canonical connection $\mathring{G}^\mu_{\cdot\nu}$ satisfies equation~\eqref{eq:A}.
    \item If $\mathrm{Lan}_\mu \neq 0$, then no solution of~\eqref{eq:A} can share the same geodesic trajectories as $\mathring{G}^\mu_{\cdot\nu}$.
\end{enumerate}
\end{thm}

Equation \eqref{eq:M} resembles  Einstein's, which reduced to the vanishing of the Ricci tensor  just taking traces.  Palatini vacuum reduces to the vanishing of the  Ricci scalar (in the spirit of Rutz's \eqref{e_Rutz}) but the proof is not trivial because now a differential equation must be solved \cite[Lemma 5.13]{JSV22}. Surprisingly, 
the result is obtained only for Lorentzian signature, because the proof is then based on the Riemannian maximum principle on the indicatrix
\cite[Th. C]{JSV22} (in the case  that $L$ is Riemannian, an alternative proof could be given, based on the eigenvalues of the Laplacian on a round sphere, cf. \cite[Th. 5.17]{JSV22}).

\begin{thm}
	In Lorentzian signature, the solutions $(L,N)$  of the Palatini equations  \eqref{eq:M}, \eqref{eq:A} such that $N$ is defined on the whole causal cone of $L$ satisfy $
		\mathrm{Ric} = 0.$
	\end{thm}

\begin{rem}
 In the semi-Riemannian case, an equation type \eqref{eq:M} or Einstein vacuum, implies that the manifold is Einstein by Schur's theorem. This is a result for $\dim M\geq 3$ stating that, whenever $
		\mathrm{ric}^g_{\mu\nu}(x)=\lambda(x)\,g_{\mu\nu}(x)
		$
		for some function $\lambda$ on $M$   then $\lambda=\mathrm{constant}$.  In the Finslerian case, this is a major open question. Indeed, the weaker version of Schur's theorem where the hypotheses is the pointwise dependence of the  sectional\footnote{This was the original result by Schur, see the historic summary in \cite{V}. } or flag curvature, is known from Berwald's \cite{BerwaldSchur}. However, the  refined one for the scalar Ricci  was known only for Randers metrics \cite{Robles} and  other very particular cases. A recent result by F.F. Villase\~nor \cite{V} shows that  such a Finsler-type Schur result holds for a class of manifolds including weakly Landsberg ones. Noticeably, the proof is based on the conservation law associated to the invariance of Einstein-Hilbert functional by diffeomorphisms.   
\end{rem}

As in the case of semi-Riemannian Palatini \eqref{e_Palat}, given a pseudo-Finsler metric $L$, the solutions of
the affine equation \eqref{eq:A} have a fibered structure on the set of symmetric
solutions (now with fiber isomorphic to the space of anisotropic 
0-homogeneous 1-forms A), \cite[Th. A]{JSV22}. However, the uniqueness of the symmetric case is not trivial. Again, one can obtain  uniqueness, by using a fiberwise argument of analyticity (notice that  even non-smooth semi-Riemannian metrics are fiberwise analytic, as the metric at each point is a quadric, thus analytic). The highly original argument is valid for indefinite signature, as it will use crucially analiticity when $L$ vanishes \cite[Th. B]{JSV22}, and the solutions are assumed to be well defined therein, i.e. {\em proper}.  



\begin{thm}
Any analytic proper indefinite pseudo-Finsler metric
$L$ admits at most one analytic proper symmetric solution $N$
of the affine variational equation \eqref{eq:A}.
\end{thm}


\section*{Acknowledgments}
The author acknowledges warmly D. Azagra (UCM, Madrid) for his advice on smoothening in the proof of Th. \ref{t_splitting2},     F. Fern\'andez Villase\~nor (U.~Castilla la Mancha) for his comments on the manuscript, especially sections \ref{ss_nonlinearconncetions}  to \ref{ss_Finslerconnections} and~\ref{s6}, J.~Hedicke (U. Nijmegen)   and S. Nemirovski (U.~Wuppertal)  for their   comments on section \ref{s3_2},  M.A.~Javaloyes (U. Murcia)  and C. Pfeiffer (U. Bremen) for their reading of the manuscript, suggestions and references, and the referees, who found many misprints. The author is partially supported by the grant PID2024-156031NB-I00 funded by MICIU / AEI / 10.13039/501100011033 / ERDF / EU, as well as the framework IMAG/  Mar\'{\i}a de Maeztu,   CEX2020-001105-MCIN/ AEI/ 10.13039/501100011033.

\section*{Data Availability Statement}
Data sharing is not applicable to this article as no datasets were generated or analysed during the current study.

\end{document}